\documentclass{article}

\usepackage[utf8]{inputenc}

\usepackage[T1]{fontenc}

\usepackage{animate}

\usepackage{fullpage}
\usepackage{cite}

\usepackage{xcolor}
\usepackage[shortlabels]{enumitem} 
\usepackage{physics}

\usepackage{tikz}
\usetikzlibrary{fadings}

\usepackage[colorlinks=true,citecolor=blue,linkcolor=blue]{hyperref}

\usepackage{subcaption}
\usepackage{caption}

\usepackage{algorithm}
\usepackage{caption}
\usepackage[noend]{algpseudocode}
\makeatletter
\newenvironment{phase}[1][htb]{%
    \renewcommand{\ALG@name}{Phase}
       \begin{algorithm}[#1]%
  }{\end{algorithm}}
\makeatother
\usepackage{amssymb}
\usepackage{amsthm,amsmath}
\usepackage[capitalise, nameinlink]{cleveref}

\newtheorem{theorem}{Theorem}[section]
\numberwithin{theorem}{section}
\newtheorem{lemma}[theorem]{Lemma}

\newtheorem{corollary}[theorem]{Corollary}
\usepackage{thmtools}
\usepackage{thm-restate}

\newcommand{\fieldinput}{\mathtt{input}}
\newcommand{\fieldoutput}{\mathtt{output}}
\newcommand{\fieldphase}{\mathtt{phase}}
\newcommand{\fieldrole}{\mathtt{role}}
\newcommand{\fieldbias}{\mathtt{bias}}
\newcommand{\fieldcounter}{\mathtt{counter}}
\newcommand{\fieldbiases}{\mathtt{opinions}}
\newcommand{\fieldlevel}{\mathtt{exponent}}
\newcommand{\fieldexponent}{\mathtt{exponent}}
\newcommand{\fieldfull}{\mathtt{full}}
\newcommand{\fieldactive}{\mathtt{active}}

\newcommand{\fieldhour}{\mathtt{hour}}
\newcommand{\fieldminute}{\mathtt{minute}}
\newcommand{\fieldTlevel}{\mathtt{hour}}
\newcommand{\fieldsample}{\mathtt{sample}}
\newcommand{\fieldestimatelogn}{\mathtt{logn}}

\newcommand{\fieldassigned}{\mathtt{assigned}}
\newcommand{\fieldopinion}{\mathtt{opinion}}

\newcommand{\rolemain}{\mathsf{Main}}
\newcommand{\roleclock}{\mathsf{Clock}}
\newcommand{\rolereserve}{\mathsf{Reserve}}

\newcommand{\roleMCR}{\mathsf{Role_{MCR}}}
\newcommand{\roleCR}{\mathsf{Role_{CR}}}
\newcommand{\True}{\mathsf{True}}
\newcommand{\False}{\mathsf{False}}
\newcommand{\unbiased}{\mathcal{O}}

\newcommand{\hourstart}{\ensuremath{\mathrm{start}}}
\newcommand{\hourend}{\ensuremath{\mathrm{end}}}

\newcommand{\kconstant}{45}

\newcommand{\roleStableSize}{\ensuremath{\mathsf{Size}}}

\usepackage[disable,textsize=scriptsize]{todonotes}

\newcommand{\A}{\mathsf{A}}
\newcommand{\B}{\mathsf{B}}
\newcommand{\T}{\mathsf{T}}
\newcommand{\mains}{\mathcal{M}}
\newcommand{\clocks}{\mathcal{C}}
\newcommand{\reserves}{\mathcal{R}}

\usepackage{listings}
\lstdefinelanguage{json}
{
    string=[s]{"}{"},
    stringstyle=\color{blue},
    comment=[l]{:},
    commentstyle=\color{black},
    basicstyle=\small\ttfamily,
}

\newcommand{\polylog}{\mathrm{polylog}}

\usepackage{mathtools}
\DeclarePairedDelimiter\ceil{\lceil}{\rceil}
\DeclarePairedDelimiter\floor{\lfloor}{\rfloor}

\let\vec\mathbf
\newcommand{\N}{\mathbb{N}}

\newcommand{\vc}{\vec{c}}
\newcommand{\vo}{\vec{o}}
\newcommand{\vi}{\vec{i}}

\newcommand{\calP}{\mathcal{P}}
\newcommand{\reach}{\Rightarrow}
\renewcommand{\Pr}{\mathbb{P}}
\newcommand{\E}{\mathbb{E}}

\newcommand{\countersubroutine}{\hyperref[alg:standard-counter-subroutine]{Standard Counter Subroutine}}
\newcommand{\phaseInitialize}{\hyperref[phase:initialize]{Phase~0}}
\newcommand{\phaseDiscreteAveraging}{\hyperref[phase:discrete-averaging]{Phase~1}}
\newcommand{\phaseConsensus}{\hyperref[phase:consensus]{Phase~2}}
\newcommand{\phaseMainAveraging}{\hyperref[phase:main-averaging]{Phase~3}}
\newcommand{\phaseDetectTie}{\hyperref[phase:detecttie]{Phase~4}}
\newcommand{\phaseReserveSample}{\hyperref[phase:reserve-sample]{Phase~5}}
\newcommand{\phaseReserveSplit}{\hyperref[phase:reserve-split]{Phase~6}}
\newcommand{\phaseHighMinorityElimination}{\hyperref[phase:highminorityelimination]{Phase~7}}
\newcommand{\phaseLowMinorityElimination}{\hyperref[phase:lowminorityelimination]{Phase~8}}
\newcommand{\phaseConsensusTwo}{\hyperref[phase:consensus2]{Phase~9}}
\newcommand{\phaseStableBackup}{\hyperref[phase:backup]{Phase~10}}
\newcommand{\stableBackup}{\hyperref[phase:backup]{Phase~10}}
\newcommand{\majorityNonuniform}{\hyperref[alg:majority-nonuniform]{Nonuniform Majority}}

\usepackage[normalem]{ulem} 

\title{\Large A time and space optimal stable population protocol solving exact majority}

\usepackage{authblk}
\author[1]{David Doty}
\author[1]{Mahsa Eftekhari}
\author[2]{Leszek G\k{a}sieniec} 
\author[1]{Eric Severson}
\author[3]{Grzegorz Stachowiak}
\author[3]{Przemysław Uznański}
\affil[1]{University of California, Davis, {\tt \{doty,mhseftekhari,eseverson\}@ucdavis.edu}}
\affil[2]{University of Liverpool, {\tt l.a.gasieniec@liverpool.ac.uk}}
\affil[3]{University of Wrocław, {\tt \{gst,puznanski\}@cs.uni.wroc.pl}}

\date{}

\begin{document}

\maketitle

\begin{abstract}
    We study population protocols, a model of distributed computing appropriate for modeling well-mixed chemical reaction networks and other physical systems where agents exchange information in pairwise interactions, but have no control over their schedule of interaction partners.
    The well-studied \emph{majority} problem is that of determining in an initial population of $n$ agents, each with one of two opinions $A$ or $B$, whether there are more $A$, more $B$, or a tie.
    A \emph{stable} protocol solves this problem with probability 1 by eventually entering a configuration in which all agents agree on a correct consensus decision of $\A$, $\B$, or $\T$, from which the consensus cannot change.
    We describe a protocol that solves this problem using $O(\log n)$ states ($\log \log n + O(1)$ bits of memory) and optimal expected time $O(\log n)$. 
    The number of states $O(\log n)$ is known to be optimal for the class of 
    polylogarithmic time stable protocols that are
    ``output dominant'' and 
    ``monotone''~\cite{alistarh2018space}.
    These are two natural constraints satisfied by our protocol,
    making it simultaneously time- and state-optimal for that class.
    We introduce a key technique called a ``fixed resolution clock'' to achieve partial synchronization.
    
    Our protocol is \emph{nonuniform}:
    the transition function has the value $\ceil{\log n}$ encoded in it.
    We show that the protocol can be modified to be uniform, while increasing the state complexity to $\Theta(\log n \log \log n)$.
\end{abstract}

\paragraph{Acknowledgements.}
Doty, Eftekhari, and Severson were supported by NSF award 1900931 and CAREER award 1844976.

\todototoc
\listoftodos

\thispagestyle{empty}
\clearpage 
\pagenumbering{arabic} 

\section{Introduction}
\label{sec:intro}

\emph{Population protocols}~\cite{AADFP06}
are asynchronous, complete networks that consist of computational entities called \emph{agents}
with no control over the schedule of interactions with other agents.
In a population of $n$ agents,
repeatedly a random pair of agents is chosen to interact,
each observing the state of the other agent before updating its own state.\footnote{
    Using message-passing terminology, each agent sends its entire state of memory as the message.
}
They are an appropriate model for electronic computing scenarios such as sensor networks
and for ``fast-mixing'' physical systems such as
animal populations~\cite{Volterra26},
gene regulatory networks~\cite{bower2004computational},
and chemical reactions~\cite{SolCooWinBru08},
the latter increasingly regarded as an implementable ``programming language'' for molecular engineering,
due to recent experimental breakthroughs in DNA nanotechnology~\cite{chen2013programmable, srinivas2017enzyme}.

Time complexity in a population protocol is defined by \emph{parallel time}: 
the total number of interactions divided by the population size $n$,
henceforth called simply \emph{time}.
This captures the natural timescale in which each individual agent experiences expected $O(1)$ interactions per unit time.
All problems solvable with zero error probability 
by a constant-state population protocol are solvable in $O(n)$ time~\cite{AAE08, DotHajLDCRNNaCo}.
The benchmark for ``efficient'' computation is thus sublinear time,
ideally $\polylog(n)$,
with $\Omega(\log n)$ time a lower bound on most nontrivial computation,
since a simple coupon collector argument shows that is the time required for each agent to have at least one interaction.

As a simple example of time complexity, 
suppose we want to design a protocol to decide whether at least one $x$ exists in an initial population of $x$'s and $q$'s.
The single transition $x,q \to x,x$ indicates that if agents in states $x$ and $q$ interact, the $q$ agent changes state to $x$.
If $x$ outputs ``yes'' and $q$ outputs ``no'',
this takes expected time $O(\log n)$ 
to reach a consensus of all $x$'s (i.e., $O(n \log n)$ total interactions, including null interactions between two $x$'s or between two $q$'s).
However,
the transitions 
$x,x \to y,y$; \quad
$y,x \to y,y$; \quad
$y,q \to y,y$, 
where $x,q$ output ``no'' and $y$ outputs ``yes'',
which computes whether at least \emph{two} $x$'s exist,
is exponentially slower:
expected time $O(n)$.
The worst-case input is exactly 2 $x$'s and $n-2$ $q$'s, where the first interaction between the $x$'s takes expected $\binom{n}{2} = \Theta(n^2)$ interactions, i.e., $\Theta(n)$ time.

To have probability 0 of error, 
a protocol must eventually \emph{stabilize}:
reach a configuration where all agents agree on the correct output, 
which is \emph{stable},
meaning no subsequently reachable configuration can change the output.\footnote{
    Technically this connection between probability 1 correctness and reachability requires the number of producible states for any fixed population size $n$ to be finite, which is the case for our protocol.
}
The original model~\cite{AADFP06}
assumed states and transitions are constant with respect to $n$.
However, for important 
problems such as
leader election~\cite{LeaderElectionDIST},
majority computation~\cite{alistarh2017time},
and computation of other functions and predicates~\cite{belleville2017time},
no constant-state protocol can stabilize in sublinear time with probability 1.\footnote{
    These problems have $O(1)$ state, sublinear time \emph{converging} protocols~\cite{kosowski2018population}.
    A protocol \emph{converges} when it reaches the correct output without subsequently changing it---though it may remain change\emph{able} for some time after converging---whereas it \emph{stabilizes} when the output becomes unchangeable.
    See~\cite{Berenbrink_majority_2020, LeaderElectionDIST} for a discussion of the distinction between stabilization and convergence.
    In this paper we consider only stabilization time.
}
This has motivated the study of population protocols whose number of states is allowed to grow with $n$, 
and as a result they can solve such problems in polylogarithmic time~\cite{elsasser2018recent, doty2018exact, burman2021self, alistarh2015fastExactMajority, alistarh2017time, alistarh2018space, berenbrink2018simple, bilke2017brief, berenbrink2018population, ben2020log3, doty2019efficient, Berenbrink2019counting, alistarh2015polylogarithmic, GS18, Berenbrink_majority_2020, gasieniec2019almost, berenbrink2020optimal}.

\subsection{The majority problem in population protocols}

\begin{table}
\centering
\caption{Summary of results on the stable exact majority problem in population protocols, including this paper [$\ast$].
Gray regions are provably impossible:
$o(\log \log n)$ state, $o(n)$ time unconditionally~\cite{alistarh2017time},
$o(\log n)$ state, $O(n^{1-\varepsilon})$ time for monotone, output-dominant protocols~\cite{alistarh2018space},
and
$o(\log n)$ time unconditionally. 
}
\vspace{-0.2cm}

\begin{tikzpicture}
\definecolor{lightgray}{gray}{0.8}
\fill[fill=lightgray] (0,0) rectangle (6,1); 
\fill[fill=lightgray, path fading = east] (5.99,0) rectangle (6.7,1);
\fill[fill=lightgray] (0,0) rectangle (1.4,4); 

\draw[->] (0,0) -- (0,5) node[anchor=east] {Time};
\draw	(0,0.33) node[anchor=east] {$O(1)$};
\draw	(0,1) node[anchor=east] {$O(\log n)$};
\draw	(0,1.66) node[anchor=east] {$O(\log^{3/2} n)$};
\draw	(0,2.33) node[anchor=east] {$O(\log^{5/3} n)$};
\draw	(0,3) node[anchor=east] {$O(\log^{2} n)$};
\draw	(0,4) node[anchor=east] {$\Omega(n)$};
	
\draw[->] (0,0) -- (6.3,0) node[anchor=north] {States};
\draw	(0.5+0.3,-0.4) node[anchor=north,rotate=-45] {$O(1)$};
\draw	(1.4+0.3,-0.4) node[anchor=north,rotate=-45] {$O(\log  n)$};
\draw	(2.6+0.3,-0.4) node[anchor=north,rotate=-45] {$O(\log^2  n)$};
\draw	(3.8+0.3,-0.4) node[anchor=north,rotate=-45] {$O(\log^3  n)$};
\draw	(4.8+0.3,-0.4) node[anchor=north,rotate=-45] {$O(n^\varepsilon)$};
\draw	(5.6+0.3,-0.4) node[anchor=north,rotate=-45] {$\Omega(n)$};

\draw [-] (0.5, 4) node[circle, fill,inner sep=2pt] {} -- (0.5, 4) node[anchor=south west]  
{\hspace{-0.05in}\cite{draief2012convergence, mertzios2014determining,AADFP06,AAE08}} ;
\draw [-] (3.8, 3) node [circle, fill,inner sep=2pt] {} -- (3.8, 3) node[anchor=south west] 
{\hspace{-0.05in}\cite{alistarh2017time}};
\draw [-] (2.6, 3) node [circle, fill,inner sep=2pt] {} -- (2.6, 3) node[anchor=south west] 
{\hspace{-0.05in}\cite{bilke2017brief}};
\draw [-] (1.4, 1.66) node [circle, fill,inner sep=2pt] {} -- (1.4, 1.66) node[anchor=south west] 
{\hspace{-0.05in}\cite{ben2020log3}};
\draw [-] (1.4, 2.33) node [circle, fill,inner sep=2pt] {} -- (1.4, 2.33) node[anchor=south west] 
{\hspace{-0.05in}\cite{berenbrink2018population}};
\draw [-] (1.4, 3) node [circle, fill,inner sep=2pt] {} -- (1.4, 3) node[anchor=south west] 
{\hspace{-0.05in}\cite{alistarh2018space,Berenbrink_majority_2020}};
\draw [-] (5.5, 3) node[circle, fill,inner sep=2pt] {} -- (5.5, 3) node[anchor=south west] 
{\hspace{-0.05in}\cite{alistarh2015fastExactMajority}};
\draw [-] (4.7, 1) node [circle, fill,inner sep=2pt] {} -- (4.7, 1) node[anchor=south west] 
{\hspace{-0.05in}\cite{Berenbrink_majority_2020}};
\draw [-] (5.5, 1) node  [circle, fill,inner sep=2pt] {} -- (5.5, 1) node[anchor=south west] 
{\hspace{-0.05in}\cite{MocquardAABS2015,mocquard2016optimal}};
\draw [-] (1.4, 1) node [circle, fill,inner sep=2pt] {} -- (1.4, 1) node[anchor=south west] 
{\hspace{-0.05in}[$\ast$]};//{this paper};
\end{tikzpicture}

\label{tab:majority-summary}
\vspace{-0.8cm}
\end{table}

Angluin, Aspnes, and Eisenstat~\cite{AAE08-2} showed a protocol they called \emph{approximate majority},
which means that starting from an initial population of $n$ agents with opinions $A$ or $B$, if $|A-B| = \omega(\sqrt{n} \log n)$ 
(i.e., the \emph{gap} between the initial majority and minority counts is greater than roughly $\sqrt{n}$), 
then with high probability the algorithm stabilizes to all agents adopting the majority opinion in $O(\log n)$ time.
A tighter analysis by Condon, Hajiaghayi, Kirkpatrick, and Ma\v{n}uch~\cite{condon2020approximate} reduced the required gap to $\Omega(\sqrt{n \log n})$.

Mertzios, Nikoletseas, Raptopoulos, and Spirakis~\cite{mertzios2014determining},
and independently
Draief and Vojnovi\'{c}~\cite{draief2012convergence},
showed a 4-state protocol that solves \emph{exact} majority problem, i.e., it identifies the majority correctly, no matter how small the initial gap.\footnote{
    The 4-state protocol doesn't identify ties, (gap $=0$), 
    but this can be handled with 2 more states;
    see~\hyperref[alg:slow-stable-majority]{Stable Backup}.
}
We henceforth refer to this simply as the \emph{majority} problem.
The protocol of~\cite{draief2012convergence, mertzios2014determining}
is also \emph{stable} in the sense that it has probability 1 of getting the correct answer.
However, this protocol takes $\Omega(n)$ time in the worst case: when the gap is $O(1)$.
Known work on the stable majority problem is summarized in~\cref{tab:majority-summary}.
G\k{a}sieniec, Hamilton, Martin, Spirakis, and Stachowiak~\cite{gasieniec2017deterministic} investigated $\Omega(n)$ time protocols for majority and the more general ``plurality consensus'' problem.
Blondin, Esparza, Jaax, and Ku\v{c}era~\cite{BlackNinjasInTheDark} show a similar stable (also $\Omega(n)$ time) majority protocol that also reports if there is a tie.

Alistarh, Gelashvili, and Vojnovi\'{c}~\cite{alistarh2015fastExactMajority} showed the first stable majority protocol with worst-case polylogarithmic expected time, requiring $\Omega(n)$ states.
A series of positive results reduced the state and time complexity for stable majority protocols~\cite{alistarh2017time, alistarh2018space, bilke2017brief, berenbrink2018population, MocquardAABS2015, mocquard2016optimal, Berenbrink_majority_2020, ben2020log3}.
Ben-Nun, Kopelowitz, Kraus, and Porat showed the current fastest stable sublinear-state protocol~\cite{ben2020log3} using $O(\log^{3/2} n)$ time and $O(\log n)$ states.
The current state-of-the-art protocols use alternating phases of \emph{cancelling} 
(two biased agents with opposite opinions both become unbiased, preserving the 
{difference between the majority and minority counts})
and \emph{splitting} (a.k.a. \emph{doubling}):
a biased agent converts an unbiased agent to its opinion; if all biased agents that didn't cancel can successfully split in that phase, then the 
{count difference} 
doubles.
The goal is to increase the 
{count difference}
until it is $n$;
i.e., all agents have the majority opinion.
See~\cite{elsasser2018recent, alistarh2018recent} for relevant surveys.

Some non-stable protocols solve exact majority with high probability but have a small positive probability of incorrectness.
Berenbrink, Els\"{a}sser, Friedetzky, Kaaser, Kling, and Radzik~\cite{Berenbrink_majority_2020} showed a protocol that with initial gap $\alpha$ uses $O(s+\log \log n)$ states and WHP converges in $O(\log n\log_s(\frac{n}{\alpha}))$ time.\footnote{
    This protocol is \textsc{SimpleMajority} in~\cite{Berenbrink_majority_2020}, which they then build on to achieve multiple stable protocols.
    The stable protocols require either $\Omega(n)$ stabilization time or $\Omega(\log n)$ states to achieve sublinear stabilization time.
}
Setting $\alpha = 1$ and $s = O(1)$, 
their protocol uses 
$O(\log \log n)$ states and converges in 
$O(\log^2 n)$ time. 
Kosowski and Uznański~\cite{kosowski2018population} showed a protocol using $O(1)$ states and converging in poly-logarithmic time with high probability.

On the negative side,
Alistarh, Aspnes, and Gelashvili~\cite{alistarh2018space} showed that any stable majority protocol taking (roughly) less than linear time requires $\Omega(\log n)$ states if it also satisfies two conditions 
(satisfied by all known stable majority protocols, including ours):
\emph{monotonicity} 
and 
\emph{output dominance}.
These concepts are discussed in \cref{sec:conclusion}.
In particular, the $\Omega(\log n)$ state bound of~\cite{alistarh2018space} applies only to \emph{stable} (probability 1) protocols;
the high probability protocol of~\cite{Berenbrink_majority_2020},
for example,
uses $O(\log \log n)$ states and $O(\log^2 n)$ time.

\subsection{Our contribution}

We show a stable population protocol solving the exact majority problem in optimal $O(\log n)$ time (in expectation and with high probability) that uses $O(\log n)$ states.
Our protocol is both monotone and output dominant 
(see \cref{sec:conclusion} or~\cite{alistarh2018space} for discussion of these definitions),
so by the $\Omega(\log n)$ state lower bound of~\cite{alistarh2018space},
our protocol is both time and space optimal for the class of monotone, output-dominant stable protocols.

A high-level overview of the algorithm is given in~\cref{subsec:high-level-overview,subsec:overview}, with a full formal description given in \cref{subsec:pseudocode}.
Like most known majority protocols using more than constant space
(the only exceptions being in~\cite{Berenbrink_majority_2020}),
our protocol is \emph{nonuniform}:
agents have an estimate of the value $\ceil{\log n}$ embedded in the transition function and state space.
\cref{sec:uniform} describes how to modify our main protocol to make it uniform, retaining the $O(\log n)$ time bound, 
but increasing the state complexity to $O(\log n \log \log n)$ in expectation and with high probability.
That section discusses challenges in creating a uniform $O(\log n)$ state protocol.

\section{Preliminaries}
\label{sec:prelim}



We write $\log n$ to denote $\log_2 n$, and $\ln n$ to denote the natural logarithm. We write $x \sim y$ to denote that $x$ and $y$ are asymptotically equivalent (implicitly in the population size $n$), meaning $\lim\limits_{n\to\infty}\frac{x(n)}{y(n)} = 1$.

\subsection{Population protocols}
\label{subsec:prelim-pop-protocols}

A \emph{population protocol} is a pair $\calP=(\Lambda,\delta)$,
where $\Lambda$ is a finite set of \emph{states}, and $\delta: \Lambda \times \Lambda \to \Lambda \times \Lambda$ 
is the \emph{transition function}.\footnote{
    To understand the full generality of our main protocol,
    we include randomized transitions in our model.
    However, there is only one type of randomized transition in the protocol (the ``drip reactions'' of \phaseMainAveraging\ described in~\cref{subsec:overview}), parameterized by probability $p$, and in fact we prove the protocol works even when these transitions are deterministic, i.e., when $p=1$.
}
In this paper we deal with \emph{nonuniform} protocols in which a different $\Lambda$ and $\delta$ are allowed for different population sizes $n$ 
(one for each possible value of $\ceil{\log n}$),
but we abuse terminology and refer to the whole family as a single protocol.
In all cases (as with similar nonuniform protocols),
the nonuniformity is used to embed the value $\ceil{\log n}$ into each agent; 
the transitions are otherwise ``uniformly specified''.
See Section~\ref{sec:uniform} for more discussion of uniform protocols.

A \emph{configuration} $\vc$ of a population protocol is a multiset over $\Lambda$ of size $n$, giving the states of the $n$ agents in the population. For a state $s \in \Lambda$, we write $\vc(s)$ to denote the count of agents in state $s$.
A \emph{transition} (a.k.a., \emph{reaction}) 
is a 4-tuple $\alpha = (r_1,r_2,p_1,p_2)$,
written $\alpha: r_1,r_2 \to p_1,p_2$,
such that $\delta(r_1,r_2) = (p_1,p_2)$.
If an agent in state $r_1$ interacts with an agent in state $r_2$, then they change states to $p_1$ and $p_2$.
For every pair of states $r_1,r_2$ without an explicitly listed transition $r_1,r_2 \to p_1,p_2$, there is an implicit \emph{null} transition $r_1,r_2 \to r_1,r_2$ in which the agents interact but do not change state.
For our main protocol,
we specify transitions formally with pseudocode that indicate how agents alter each independent field in their state.

\newcommand{\inputvector}{\vv_{\mathrm{i},\alpha}}
\newcommand{\outputvector}{\vv_{\mathrm{o},\alpha}}
\newcommand{\transitionvector}{\vv_{\alpha}}

More formally, given a configuration $\vc$ and transition $\alpha: r_1,r_2 \to p_1,p_2$, we say that $\alpha$ is \emph{applicable} to $\vc$ if $\{r_1,r_2\} \subseteq \vc$, i.e., $\vc$ contains 2 agents, one in state $r_1$ and one in state $r_2$.
If $\alpha$ is applicable to $\vc$, then we write $\vc\reach_\alpha\vc'$, where $\vc' = \vc - \{r_1,r_2\} + \{p_1,p_2\}$ is the configuration that results from applying $\alpha$ to $\vc$.
We write $\vc \reach \vc'$,
and we say that $\vc'$ is \emph{reachable} from $\vc$,
if there is a sequence $(\alpha_0,\alpha_1,\ldots,\alpha_k)$ of transitions such that
$\vc\reach_{\alpha_0}\vc_1\reach_{\alpha_1}\ldots\reach_{\alpha_k}\vc'$.
This notation omits mention of $\calP$; we always deal with a single protocol at a time, so it is clear from context which protocol is defining the transitions.

\subsection{Stable majority computation with population protocols}
There are many modes of computation considered in population protocols:
computing integer-valued functions~\cite{CheDotSolNaCo,DotHajLDCRNNaCo, belleville2017time} where the number of agents in a particular state is the output,
Boolean-valued predicates~\cite{AAE06, AAE08} where each agent outputs a Boolean value as a function of its state and the goal is for all agents eventually to have the correct output,
problems such as leader election~\cite{LeaderElectionDIST, alistarh2017time, berenbrink2018simple,  alistarh2015polylogarithmic, GS18, gasieniec2019almost, berenbrink2020optimal},
and generalizations of predicate computation, where each agent individually outputs a value from a larger range, such as reporting the population size~\cite{doty2018exact, doty2019efficient, Berenbrink2019counting}.
Majority computation is Boolean-valued if computing the predicate ``$A \geq B$?'', where $A$ and $B$ represent the initial counts of two opinions $\A$ and $\B$.
We define the slightly generalized problem that requires recognizing when there is a tie,
so the range of outputs is $\{\A,\B,\T\}$.

Formally, if the set of states is $\Lambda$, the protocol defines a disjoint partition of $\Lambda = \Lambda_\A \cup \Lambda_\B \cup \Lambda_\T$.
For $u \in \{\A,\B,\T\}$, 
if $a\in\Lambda_u$ for all $a\in\vc$,
we define \emph{output} $\phi(\vc) = u$ 
(i.e., all agents in $\vc$ agree on the output $u$).
Otherwise
$\phi(\vc)$ is undefined 
(i.e., agents disagree on the output).
We say $\vo$ is \emph{stable} if $\phi(\vo)$ is defined and, for all $\vo_2$ such that $\vo \reach \vo_2$, $\phi(\vo) = \phi(\vo_2)$,
i.e., the output cannot change.

The protocol identifies two special 
\emph{input states} $A,B \in \Lambda$.
A \emph{valid} initial configuration $\vi$ satisfies
$a\in\{A,B\}$ for all $a\in\vi$.
We say the \emph{majority opinion} of $\vi$ is 
$M(\vi) = \A$ if $\vi(A) > \vi(B)$, 
$M(\vi) = \B$ if $\vi(A) < \vi(B)$, 
and
$M(\vi) = \T$ if $\vi(A) = \vi(B)$.
The protocol \emph{stably computes majority} if,
for any valid initial configuration $\vi$,
for all $\vc$ such that $\vi \reach \vc$,
there is a stable $\vo$ such that $\vc \reach \vo$
and $\phi(\vo) = M(\vi)$.
Let $O_{\vi} = \{\vo:\phi(\vo)=M(\vi)\}$ be the set of all correct, stable configurations.
In other words,
for any reachable configuration, it is possible to reach a correct, stable configuration, or equivalently reach a strongly connected component in $O_{\vi}$.


\subsection{Time complexity}

In any configuration the next interaction is chosen by selecting a pair of agents uniformly at random and applying an applicable transition,
with appropriate probabilities for any randomized transitions.
Thus the sequence of transitions and configurations they reach are random variables.
To measure time we count the total number of interactions (including null transitions such as $a,b \to a,b$ in which the agents interact but do not change state),
and divide by the number of agents $n$.
In the population protocols literature, this is often called ``parallel time'':
$n$ interactions among a population of $n$ agents equals one unit of time.

If the protocol stably computes majority, then for any valid initial configuration $\vi$, the probability of reaching a stable, correct configuration, $\Pr[\vi\reach O_{\vi}] = 1$.\footnote{
    Since population protocols have a finite reachable configuration space, this is equivalent to the stable computation definition that for all $\vc$ reachable from $\vi$, there is a $\vo' \in O_{\vi}$ reachable from $\vc$.
}
We define the \emph{stabilization time} $S$ to be the random variable giving the 
time to reach a configuration $\vo \in O_{\vi}$.

When discussing random events in a protocol of population size $n$,
we say event $E$ happens 
\emph{with high probability} 
if $\Pr[\neg E] = O(n^{-c})$,
where $c$ is a constant that depends on our choice of parameters in the protocol,
where $c$ can be made arbitrarily large by changing the parameters.
In other words, the probability of failure can be made an arbitrarily small polynomial.
For concreteness, we will write a particular polynomial probability such as $O(n^{-2})$,
but in each case we could tune some parameter 
(say, increasing the time complexity by a constant factor)
to increase the polynomial's exponent.
We say event $E$ happens 
\emph{with very high probability}
if $\Pr[\neg E] = O(n^{-\omega(1)})$,
i.e., if its probability of failure is smaller than any polynomial probability.

\section{Nonuniform majority algorithm description}
\label{sec:algorithm}

The goal of Sections~\ref{sec:algorithm} through \ref{sec:analysis-final-phases} is to show the following theorem:

\begin{restatable}{theorem}{replicatedThmNonuniformAlgorithm}
\label{thm:nonuniform-algorithm}
There is a nonuniform population protocol \majorityNonuniform,
using $O(\log n)$ states,
that stably computes majority in $O(\log n)$ stabilization time, both in expectation and with high probability.
\end{restatable}

\subsection{High-level overview of algorithm}
\label{subsec:high-level-overview}

In this overview we use ``pseudo-transitions'' such as $A,B \to \unbiased,\unbiased$ to describe agents updating a portion of their states, while ignoring other parts of the state space.

Each agent initially has a $\fieldbias$: $+1$ for opinion $A$ and $-1$ for opinion $B$, so 
the population-wide sum $g = \sum_v v.\fieldbias$ gives the \emph{initial gap} between opinions. 
The majority problem is equivalent to determining $\text{sign}(g)$. 
Transitions redistribute biases among agents but,
to ensure correctness,
maintain the population-wide $g$ as an invariant.
Biases change through \emph{cancel reactions} $+\frac{1}{2^i},-\frac{1}{2^i}\rightarrow 0, 0$ and \emph{split reactions} $\pm\frac{1}{2^i}, 0 \rightarrow \pm\frac{1}{2^{i+1}}, \pm\frac{1}{2^{i+1}}$, down to a minimum $\pm\frac{1}{2^L}$. The constant $L = \ceil{\log_2(n)}$ ensures $\Theta(\log n)$ possible states.
The \emph{gap}
is defined to be $\sum_v \text{sign}(v.\fieldbias)$, the difference in counts between majority and minority biases. Note the gap should grow over time to spread the correct majority opinion to the whole population, while the invariant $g$ should ensure correctness of the final opinion.

The cancel and split reactions average the bias value between both agents,
but only when the average is also a power of 2, or 0. 
If we had averaging reactions between all pairs of biases (also allowing, e.g.,  $\frac{1}{2},\frac{1}{4} \to \frac{3}{8},\frac{3}{8}$), then all biases would converge to $\frac{g}{n}$,
but this would use too many states.\footnote{
This was effectively the approach used for majority in \cite{alistarh2015fastExactMajority, mocquard2016optimal}, for an $O(n)$ state, $O(\log n)$ time protocol.}
With our limited set $\{0,\pm\frac{1}{2}, \pm\frac{1}{4}, \ldots, \pm\frac{1}{2^L}\}$ of possible biases, allowing all cancel and split reactions simultaneously does not work. 
Most biases appear simultaneously across the population, reducing the count of each bias,
which slows the rate of cancel reactions.
Then the count of unbiased $0$ agents is reduced, 
which slows the rate of split reactions, see \cref{fig:ba-no-synchronization}. 
Also, there is a non-negligible probability for the initial minority opinion to reach a much greater count, 
if those agents happen to do more split reactions, see \cref{fig:ba-no-synchronization2}. 
Thus using only the count of positive versus negative biases will not work to solve majority even with high probability.

To solve this problem, we partially synchronize the unbiased agents with a field $\fieldhour$, adding $\log n$ states $0_0, 0_1, 0_2, \ldots, 0_L$. The new split reactions
\[
\pm\frac{1}{2^i}, 0_h \rightarrow \pm\frac{1}{2^{i+1}}, \pm\frac{1}{2^{i+1}}
\quad
\text{ if } h > i
\]
will wait until $\fieldhour \geq h$ before doing splits down to $\fieldbias = \pm\frac{1}{2^h}$. 
We could use existing phase clocks to perfectly synchronize $\fieldhour$, by making each hour use $\Theta(\log n)$ time,
enough time for all opinionated agents to split. 
Then WHP all agents would be in states $\{0_h, +\frac{1}{2^h}, -\frac{1}{2^h}\}$ by the end of hour $h$,
see~\cref{fig:ba-perfect-synchronization}.
The invariant $g = \sum_v v.\fieldbias$ implies that all minority opinions would be eliminated by hour $\ceil{\log_2 \frac{1}{g}} \leq L$. 
This would give an $O(\log n)$-state, $O(\log^2 n)$-time majority algorithm, essentially equivalent to \cite{alistarh2018space, Berenbrink_majority_2020}.

The main idea of our algorithm is to use these rules with a faster clock using only $O(1)$ time per hour. 
The $\fieldhour$ field of unbiased agents is synchronized to a separate subpopulation of clock agents, who use a field $\fieldminute$, with $k$ consecutive minutes per hour. 
Minutes advance by \emph{drip reactions}
$C_i, C_i \rightarrow C_{i}, C_{i+1}$, and catch up by \emph{epidemic reactions}
$C_i, C_j \rightarrow C_{\max(i,j)}, C_{\max(i,j)}$. See 
\cref{fig:ba-our-protocols}
for an illustration of the clock $\fieldminute$ and $\fieldhour$ dynamics.

\begin{figure}
     \centering
     \begin{subfigure}[t]{0.49\textwidth}
         \centering
         \includegraphics[trim={2cm 0 3cm 0}, width=\textwidth]{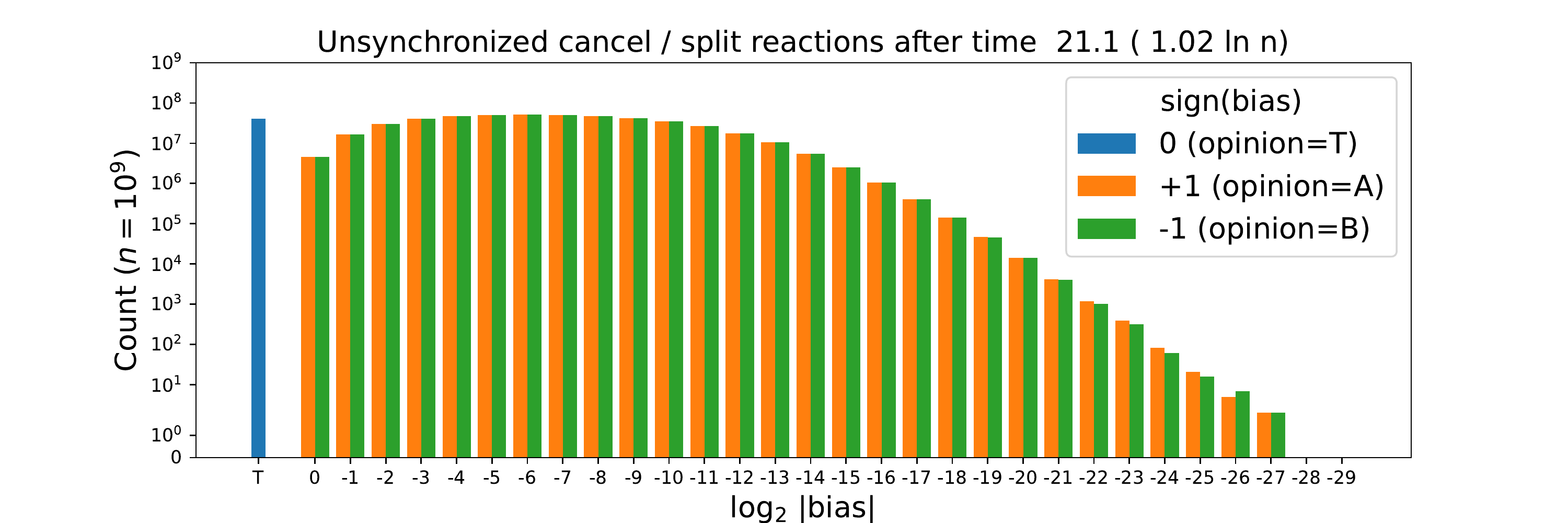}
         \caption{Cancel/split reactions with no synchronization. All states become present, many in about equal counts. 
         Rate of cancel reactions and fraction of 0 agents are $\Theta(\frac{1}{\log n})$.}
         \label{fig:ba-no-synchronization}
     \end{subfigure}
    \hfill
    \begin{subfigure}[t]{0.49\textwidth}
         \centering
         \includegraphics[trim={2cm 0 3cm 0}, width=\textwidth]{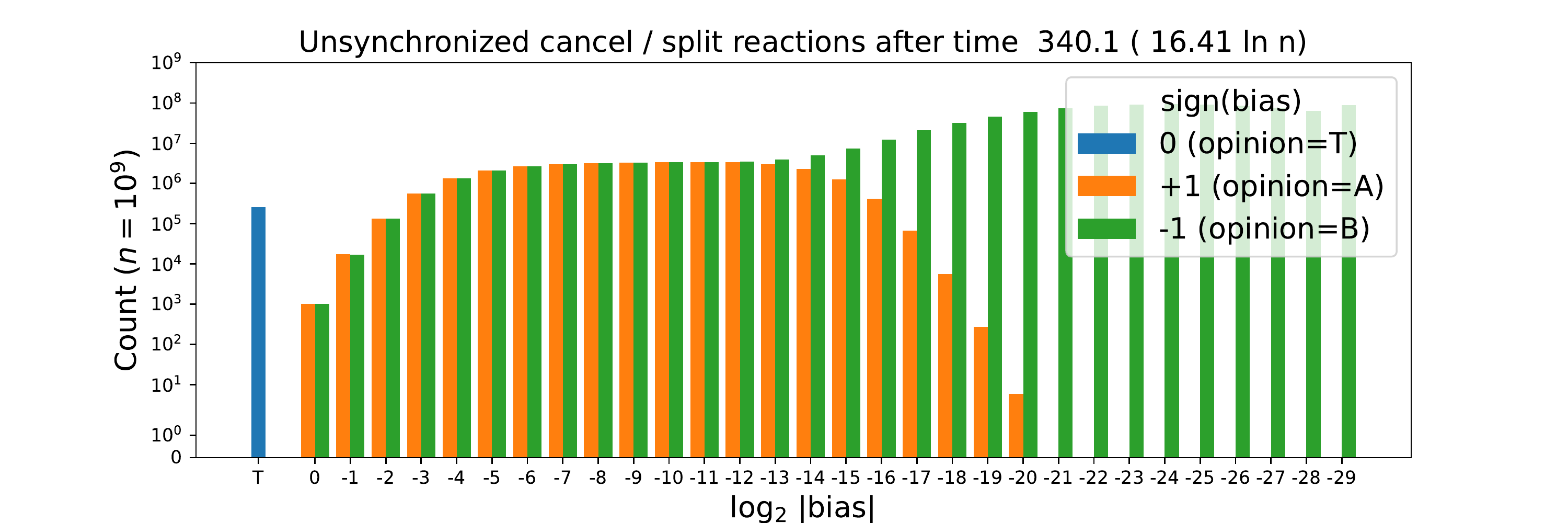}
         \caption{Later snapshot of the simulation in \cref{fig:ba-no-synchronization}. The initial minority $B$ now has a much larger count, because those agents happened to undergo more split reactions.}
         \label{fig:ba-no-synchronization2}
     \end{subfigure}
     \\
     \begin{subfigure}[t]{0.49\textwidth}
         \centering
         \includegraphics[trim={2cm 0 3cm 0}
         , width=\textwidth]{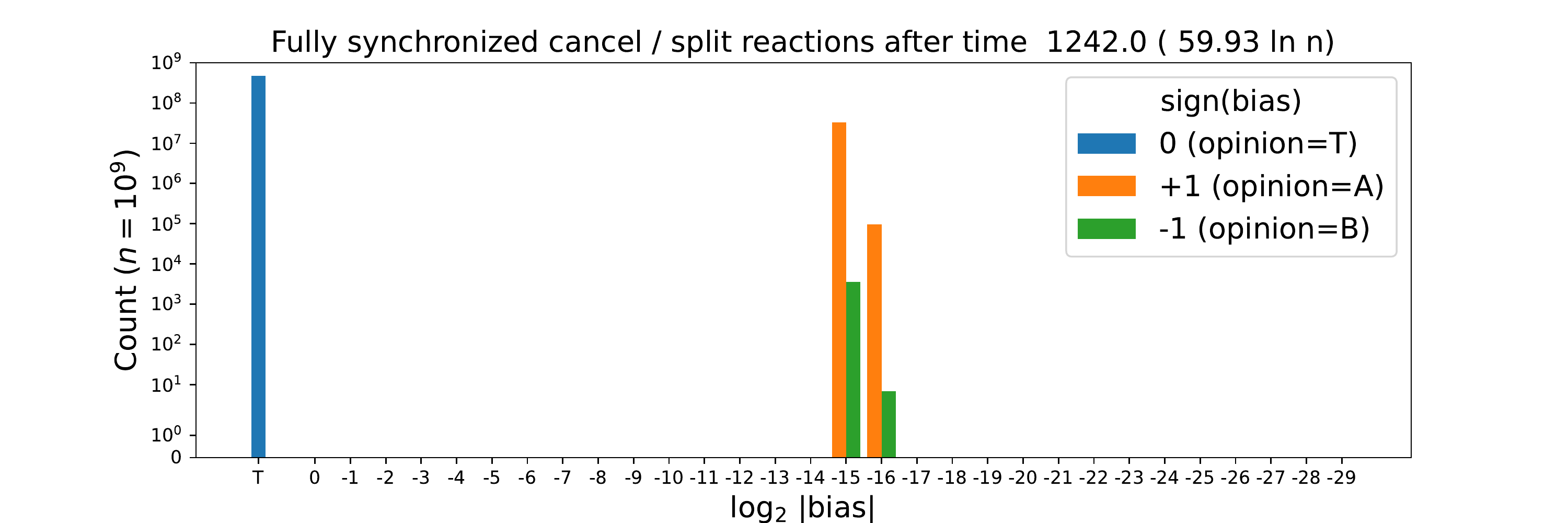}
         \caption{Cancel/split reactions, fully synchronized into $O(\log n)$ time hours, at the beginning of hour 16.
         All minority are eliminated by hour $\log n$ in $O(\log^2 n$) time.}
         \label{fig:ba-perfect-synchronization}
     \end{subfigure}
     \hfill
     \begin{subfigure}[t]{0.49\textwidth}
         \centering
         \includegraphics[trim={2cm 0 3cm 0}
         ,width=\textwidth]{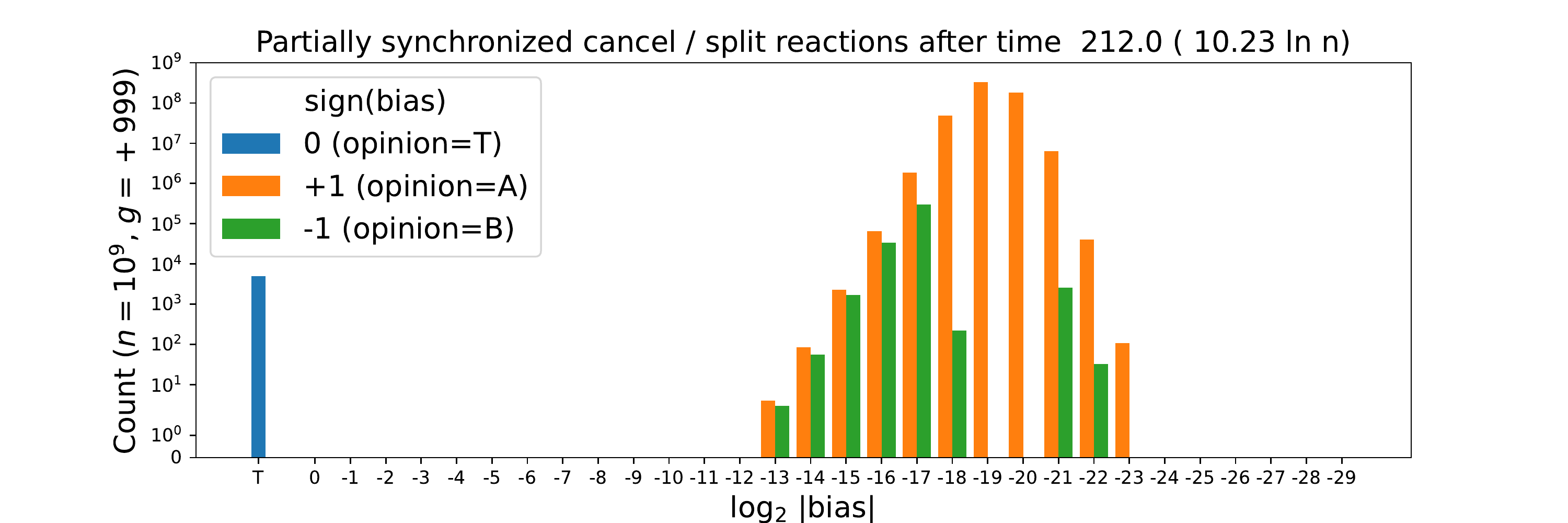}
         \caption{
         Main phase of our protocol,
         split reactions partially synchronized using the clock in~\cref{fig:ba-our-protocols}, 
         at the end of this $O(\log n)$ time phase. Most agents are left with 
         $\fieldbias \in \left\{
         +\frac{1}{2^{18}},
         +\frac{1}{2^{19}},
         +\frac{1}{2^{20}}
         \right\}$.
         Later phases eliminate the remaining minority agents.
         }
         \label{fig:ba-final-averaging-config}
     \end{subfigure}
  \caption{Cancel / split reactions with 
  no synchronization (\ref{fig:ba-no-synchronization},\ref{fig:ba-no-synchronization2}), 
  perfect synchronization
  (\ref{fig:ba-perfect-synchronization}),
  and
  partial synchronization 
  (\ref{fig:ba-final-averaging-config})
  via the fixed-resolution phase clock of our main protocol.
  Plots generated from~\cite{examplenotebook}.}
  \label{fig:ba-simple-examples}
\end{figure}

\begin{figure*}
     \centering
         \centering
         \includegraphics[trim={2cm 0 2cm 0}, width=0.79\textwidth]{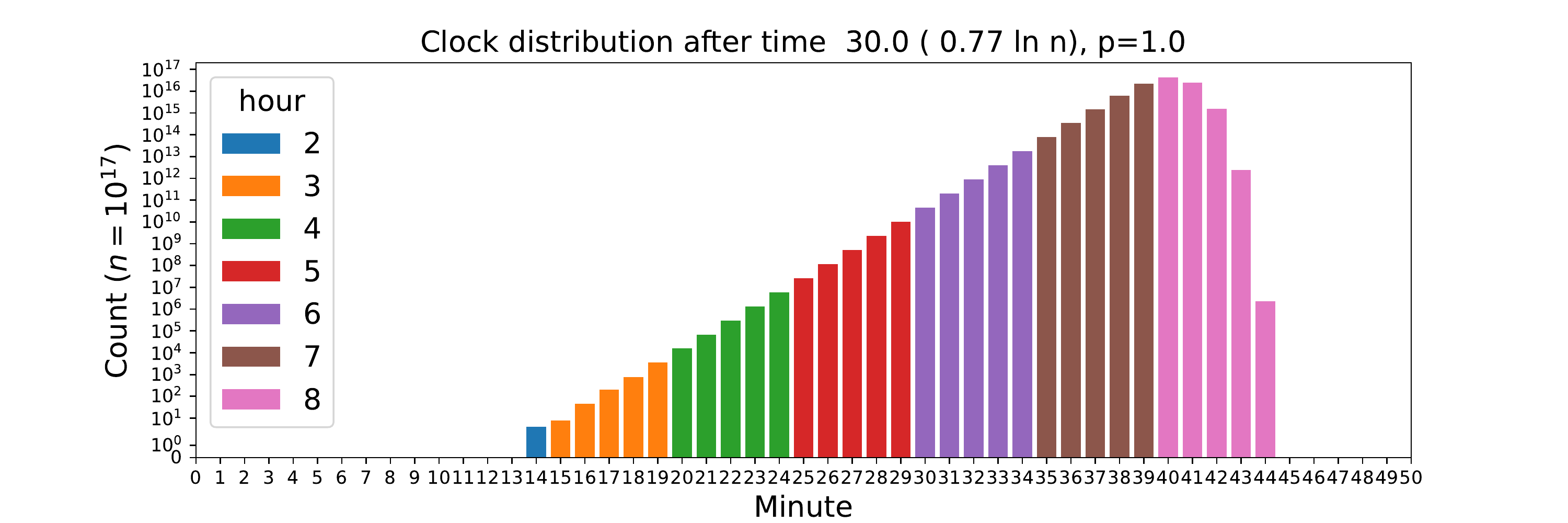}
     
        \caption{Clock rules 
        of our protocol, 
        showing a travelling wave distribution over minutes, on a larger population of size $n=10^{17}$ to emphasize the distribution.
        The distribution's back tail decays exponentially, 
        and its front tail decays doubly exponentially.
        A large constant fraction of agents are in the same two consecutive $\fieldhour$'s (here 7 and 8).
        Plot generated from~\cite{examplenotebook}. 
        }
        \label{fig:ba-our-protocols}
\end{figure*}

Since $O(1)$ time per hour is not sufficient to bring all agents up to the current hour before advancing to the next,
we now have only a large constant fraction of agents, rather than all agents, 
synchronized in the current hour.
Still, we prove this looser synchronization keeps the values of $\fieldhour$ and $\fieldbias$
relatively concentrated, so by the end of this phase, we reach a configuration as shown in~\cref{fig:ba-final-averaging-config}. 
Most agents have the majority opinion
(WLOG positive), 
with three consecutive biases 
$+\frac{1}{2^l},+\frac{1}{2^{l+1}},+\frac{1}{2^{l+2}}$.

\paragraph{Detecting ties.}
    This algorithm gives an elegant way to detect a tie with high probability.
    In this case, $g=0$, and with high probability, all agents will finish the phase with 
    $\fieldbias \in \left\{0, \pm\frac{1}{2^L} \right\}$.
    Checking this condition
    \emph{stably} detects a tie (i.e., with probability 1, if this condition is true, then there is a tie),
    because for any nonzero value of $g$, there must be some agent with $|\fieldbias| > \frac{1}{2^L}$.

\paragraph{Cleanup Phases.}
    We must next eliminate all minority opinions, while still relying on the invariant $g = \sum_v v.\fieldbias$ to ensure correctness. Note that is it possible with low probability to have a greater count of minority opinions (with smaller values of $\fieldbias$), so only relying on counts of positive and negative biases would give possibilities of error.
    
    We first remove any minority agents with large bias, by using an additional subpopulation of $\rolereserve$ agents that enable additional split reactions for large values of $|\fieldbias| > \frac{1}{2^l}$. Then after cancel reactions with the bulk of majority agents, the only minority agents left must have $|\fieldbias| < \frac{1}{2^{l+2}}$.
    
    To then remove minority agents with small bias, we allow agents with larger bias to ``consume'' agents with smaller bias, such as an interaction between agents $+\frac{1}{4}$ and $-\frac{1}{256}$. Here the positive agent can be thought to hold the entire bias $+\frac{1}{4} -\frac{1}{256} = +\frac{63}{256}$, but since this value is not in the allowable states, it can only store that its bias is in the range $+\frac{1}{8} \leq \fieldbias \leq +\frac{1}{4}$. Without knowing its exact $\fieldbias$, this agent cannot participate in future averaging interactions. However, we show there are enough available majority agents to eliminate all remaining minority via these \emph{consumption reactions}. Thus with high probability, all minority agents are eliminated.
    
    A final phase checks for the presence of both positive and negative $\fieldbias$, and if one has been completely eliminated, it stabilizes to the correct output. In the case where both are present, this is a detectable error, where we can move to a slow, correct backup that uses the original inputs.
    Due to the low probability of this case, it contributes negligibly to the total expected time.


\subsection{Intuitive description of each phase}
\label{subsec:overview}

Our full protocol is broken up into 11 consecutive phases.
We describe each phase intuitively before presenting full pseudocode in~\cref{subsec:pseudocode}. 
Note that some further separation of phases was done to create more straightforward proofs of correctness, so simplicity of the proofs was optimized over simplicity of the full protocol pseudocode. It is likely possible to have simpler logic that still solves majority via the same strategy.


\begin{description}
    \item[\phaseInitialize:]
    ``Population splitting''~\cite{alistarh2018recent} divides agents into roles used in subsequent phases:
    $\rolemain, \rolereserve, \roleclock$.
    In timed phases 
    (those not marked as \emph{Untimed} or \emph{Fixed-resolution clock}, 
    including the current phase),
    $\roleclock$ agents count from 
    $\Theta(\log n)$ to 0 to cause the switch to the next phase after $\Theta(\log n)$ time.
    
    ``Standard'' population splitting uses reactions such as $x,x \to r_1,r_2$ to divide agents into two roles $r_1,r_2$.
    This takes $\Theta(n)$ time to converge,
    which can be decreased to $\Theta(\log n)$ time via $r_1,x \to r_1,r_2$ and $r_2,x \to r_2,r_1$,
    while maintaining that $\# r_1$ and $\# r_2$ are both $n/2 \pm \sqrt{n}$ WHP.
    However, since all agents initially have an opinion,
    but $\roleclock$ and $\rolereserve$ agents do not hold an opinion,
    agents that adopt role $\roleclock$ or $\rolereserve$ must first pass off their opinion to a $\rolemain$ agent.
    
    From each interacting pair of unassigned agents, one will take the $\rolemain$ role and hold the opinions of both agents, interpreting $A$ as $+1$ and $B$ as $-1$. 
    This $\rolemain$ agent will then be allowed to take at most one other opinion
    (in an additional reaction that enables rapid convergence of the population splitting), 
    and holding 3 opinions can end up with a
    bias in the range $\{-3,-2,-1,0,+1,+2,+3\}$.

    \item[\phaseDiscreteAveraging:]
    Agents do ``integer averaging''~\cite{alistarh2015polylogarithmic} of biases in the set $\{-3,\ldots,+3\}$ via reactions 
    $i,j \to \floor{\frac{i+j}{2}}, \ceil{\frac{i+j}{2}}$.
    Although taking $\Theta(n)$ time to converge in some cases, 
    this process is 
    known~\cite{mocquard2020stochastic} to result in three consecutive values in $O(\log n)$ time.
    If those three values are detected to be $\{-1,0,+1\}$ in the next phase, the algorithm continues.

    \item[\phaseConsensus:]
    (\emph{Untimed})
    Agents propagate the set of opinions (signs of biases) remaining after \phaseDiscreteAveraging\ to detect if only one opinion remains.
    If so,
    we have converged on a majority consensus,
    and
    the algorithm halts here (see \cref{fig:simulation-linear}).
    At this point, this is essentially the exact majority protocol of~\cite{mertzios2014determining},
    which takes $O(\log n)$ time with an initial
    gap $\Omega(n)$, but longer for sublinear gaps (e.g., $\Omega(n)$ time for a gap of 1).
    Thus, if agents proceed beyond this phase
    (i.e., if both opinions $\A$ and $\B$ remain at this point),
    we will use later that
    the gap was smaller than $0.025\cdot\#\rolemain$.
    With low probability both opinions remain but some agent has $|\fieldbias| > 1$, in which case we proceed directly to a slow stable backup protocol in \phaseStableBackup.

    \item[\phaseMainAveraging:]
    (\emph{Fixed-resolution clock})
    The key goal at this phase is to use cancel and split reactions to average the bias across the population to give almost all agents the majority opinion.
    Biased agents hold a field $\fieldlevel\in\{-L,\ldots,-1,0\}$, describing the magnitude $|\fieldbias| = 2^{\fieldlevel}$,
    a quantity we call the agent's \emph{mass}. 
    Cancel reactions eliminate opposite biases
    $+\frac{1}{2^i},-\frac{1}{2^i}\rightarrow 0,0$ with the same $\fieldlevel$;
    cancel reactions strictly reduce total mass. 
    Split reactions $\pm\frac{1}{2^i},0 \rightarrow \pm\frac{1}{2^{i+1}},\pm\frac{1}{2^{i+1}}$ give half of the bias to an unbiased agent, decrementing the $\fieldlevel$;
    split reactions preserve the total mass. 
    The unbiased $\unbiased$ agents, with $\fieldrole = \rolemain$, $\fieldopinion = \fieldbias = 0$, act as the fuel for split reactions. 
    

    
    We want 
    to obtain tighter synchronization in the exponents than~\cref{fig:ba-no-synchronization},
    approximating the ideal synchronized behavior of the $O(\log^2 n)$ time algorithm of \cref{fig:ba-perfect-synchronization} while using only $O(\log n)$ time.
    To achieve this,
    the $\roleclock$ agents run a ``fixed resolution'' clock that keeps them roughly synchronized 
    (though not perfectly; see~\cref{fig:ba-our-protocols})
    as they count their ``minutes'' from 0 up to $L' = k L$,
    using $O(1)$ time per minute.
    This is done via ``drip'' reactions
    $C_i,C_i \to C_{i},C_{i+1}$ 
    (when minute $i$ gets sufficiently populated, pairs of $C_i$ agents meet with sufficient likelihood to increment the minute) and
    $C_{j},C_i \to C_{j},C_{j}$ for $i < j$
    (new higher minute propagates by epidemic).\footnote{
        This clock is similar to the power-of-two-choices leaderless phase clock of~\cite{alistarh2018space},
        where the agent with smaller (or equal) minute increments their clock ($C_{j},C_i \to C_{j},C_{i+1}$ for $i \leq j$), but increasing the smaller minute by only 1. 
        Similarly to our clock, the maximum minute can increase only with both agents at the same minute.
        A similar process was analyzed in~\cite{berenbrink2006balanced}, and in fact was shown to have the key properties needed for our clock to work---an exponentially-decaying back tail and a double-exponentially-decaying front tail---so it seems likely that a power-of-two-choices clock could also work with our protocol.
        
        The randomized variant of our clock with drip probability $p$ is also similar to the ``junta-driven'' phase clock of \cite{GS18}, 
        but with a linear number $2pn$ of agents in the junta, using $O(1)$ time per minute, rather than the $O(n^\epsilon)$-size junta of~\cite{GS18},
        which uses $O(\log n)$ time per minute.
        There, smaller minutes are brought up by epidemic, and only an agent in the junta seeing another agent at the same minute will increment. 
        The epidemic reaction is exactly the same in both rules. The probability $p$ of a drip reaction can be interpreted as the probability that one of the agents is in the junta.
        For similar rate of $O(1)$ time per minute phase clock construction see also Dudek and Kosowski work \cite{DBLP:conf/stoc/DudekK18}.
    }
    If randomized transitions are allowed,
    by lowering the probability $p$ of the drip reaction,
    the clock rate can be slowed by a constant factor.
    Although we prove a few lemmas about this generalized clock,
    and some of our simulation plots in~\cref{sec:simulation} use $p<1$,
    our proofs work even for $p=1$,
    i.e., a deterministic transition function,
    although this requires constant-factor more states (by increasing the number of ``minutes per hour'', explained next).

    Now the $\unbiased$ agents will use $\Theta(\log n)$ states to store an ``hour'', coupled to the $C$ clock agents via 
    $C_{\floor{i/k}},\unbiased_j \to C_{\floor{i/k}},\unbiased_{\floor{i/k}}$ if $\floor{i/k} > j$, i.e., every consecutive $k$ $\roleclock$ minutes corresponds to one $\rolemain$ hour, and clock agents drag $\unbiased$ agents up to the current hour (see \cref{fig:clock_minute_hour}). 
    Our proofs require $k = \kconstant$ minutes per hour when $p=1$, 
    but smaller values of $k$ work in simulation. 
    For example, the simulation in \cref{fig:ba-final-averaging-config} showing intended behavior of this phase used only $k=3$ minutes per hour with $p=1$.
    
    
    
    This clock synchronizes the $\fieldexponent$s because agents with $\fieldlevel = -i$ can only split down to $\fieldlevel = -(i+1)$ with an $\unbiased$ agent that has $\fieldTlevel \geq i+1$. 
    \todo{
        DD: I think I've brought this up before, but I still find it confusing that $\fieldexponent$ is defined for bias $2^{\fieldexponent}$ rather than $1 / 2^{\fieldexponent}$. I think the high-level discussion and the proofs would all be easier to read if we simply defined $\fieldexponent$ to be positive, and talked about $\fieldexponent$ increasing (along with $\fieldhour$), and didn't need so many negative signs everywhere.
        
        ES: I cannot remember why I ended up settling on bias $2^{\fieldexponent}$, but I don't have a good reason now.
        The biggest reason to keep it is that making sure we switch signs for everything uniformly (in both all formulas but also all english text around them) sounds difficult to do without missing any.
    }
    This prevents the biased agents from doing too many splits too quickly. As a result, during hour $i$, most of the biased agents have $|\fieldbias| = \frac{1}{2^i}$, so the cancel reactions $+\frac{1}{2^i},-\frac{1}{2^i}\rightarrow 0,0$ happen at a high rate, providing many $\unbiased$ agents as ``fuel'' for future split reactions.
    We tune the constants of the clock to ensure hour $i$ lasts long enough to bring most biased agents down to $\fieldlevel = -i$ via split reactions and then let a good fraction do cancel reactions (see \cref{fig:exponent}).
    
    The key property at the conclusion of this phase
    is that unless there is a tie,
    WHP most majority agents end up in three consecutive exponents $-l,-(l+1),-(l+2)$,
    with a negligible mass of any other $\rolemain$ agent 
    (majority agents at lower/higher exponents, minority agents at any exponent, or $\unbiased$ agents).\footnote{
        $l$ is defined such that 
        if all biased agents were at exponent $-l$,
        the difference in counts between majority and minority agents would be between $0.4\cdot\#\rolemain$ and $0.8\cdot\#\rolemain$.
    }
    Phases \ref{phase:reserve-sample}-\ref{phase:highminorityelimination} use this fact to quickly push the rest of the population to a configuration where \emph{all} minority agents have exponents strictly below $-(l+2)$;
    \phaseLowMinorityElimination\ then eliminates these minority agents quickly.
    
    \item[\phaseDetectTie:]
    (\emph{Untimed})
    The special case of a tie is detected by the fact that, 
    since the total bias remains the initial gap $g$,
    if all biased agents have minimal exponent $-L$, 
    $g$ has magnitude less than 1:
    \begin{align*}
    |g| 
    &= \qty|\sum_{a.\fieldrole = \rolemain} a.\fieldbias| 
    \leq \sum_{a.\fieldrole = \rolemain} |a.\fieldbias| \\
    &\leq \sum_{a.\fieldrole = \rolemain} \frac{1}{2^L} < \frac{n}{2^{\ceil{\log_2(n)}}} \leq 1.
    \end{align*}
    The initial gap $g$ is integer valued,
    so $|g| < 1 \implies g = 0$.
    Thus this condition implies there is a tie with probability 1;
    the converse that a tie forces all biased agents to exponent $-L$ holds with high probability.
    If only exponent $-L$ is detected,
    the algorithm halts here with all agents reporting output $\T$ (see \cref{fig:simulation-tie}).
    Otherwise, 
    the algorithm proceeds to the next phase.

    \item[\phaseReserveSample, \phaseReserveSplit:]
    Using the key property of \phaseMainAveraging,
    these phases WHP pull all biased agents above exponent $-l$ down to exponent $-l$ or below using the $\rolereserve$ $R$ agents.
    The $R$'s activate themselves in \phaseReserveSample\ by sampling the exponent of the first biased agent they meet.
    This ensures WHP that sufficiently many $\rolereserve$ agents exist with exponents $-l,-(l+1),-(l+2)$
    (distributed similarly to the agents with those exponents).
    Then in \phaseReserveSplit, 
    they act as fuel for splits, via
    $R_i,\pm\frac{1}{2^j} \to \pm\frac{1}{2^{j+1}},\pm\frac{1}{2^{j+1}}$ 
    when $|i| > |j|$.
    The reserve agents, unlike the $\unbiased$ agents in \phaseMainAveraging, do not change their exponent in response to interactions with $\roleclock$ agents.
    Thus sufficiently many reserve agents will remain to allow the small number of biased agents above exponent $-l$ to split down to exponent $-l$ or below.\footnote{
        The reason we do this in two separate phases is to ensure that the $\rolereserve$ agents have close to the same distribution of exponents that the $\rolemain$ agents have at the end of \phaseMainAveraging. If $\rolereserve$ agents allowed split reactions in the same phase that they sample the $\fieldlevel$ of $\rolemain$ agents, 
        then the splits would disrupt the distribution of the $\rolemain$ agents before all $\rolereserve$ agents have finished sampling.
        This would possibly give the $\rolereserve$ agents a significantly different distribution among levels than the $\rolemain$ agents had at the start. 
        While this may possibly work anyway, we find it is more straightforward to prove if the $\rolereserve$ agents have a close distribution over $\fieldlevel$ values to that of the $\rolemain$ agents.
    }

    \item[\phaseHighMinorityElimination:]
    This phase allows more general reactions to distribute the dyadic biases,
    allowing reactions between agents up to two exponents apart,
    to eliminate the opinion with smaller exponent:
    $\frac{1}{2^i}, -\frac{1}{2^{i+1}} \to \frac{1}{2^{i+1}},0$
    and
    $\frac{1}{2^i}, -\frac{1}{2^{i+2}} \to \frac{1}{2^{i+1}}, \frac{1}{2^{i+2}}$
    (and the equivalent with positive/negative biases swapped).
    Since all agents have exponent $-l$ or below,
    and many more majority agents exist at exponents $-l,-(l+1),-(l+2)$ than the total number of minority agents anywhere,
    these 
    (together with standard cancel reactions $\frac{1}{2^i}, - \frac{1}{2^i} \to 0,0$)
    rapidly eliminate all minority agents
    at exponents $-l,-(l+1),-(l+2)$,
    while maintaining $\Omega(n)$ majority agents at exponents 
    $\geq -(l+2)$
    and $< 0.01 n$ total minority agents, now all at exponents 
    $\leq -(l+3)$.

    \item[\phaseLowMinorityElimination:]
    This phase eliminates the last minority agents,
    while ensuring that if any error occurred in previous phases,
    some majority agents remain,
    to allow detecting the error by the presence of both opinions.\footnote{
        A na\"{i}ve idea to reach a consensus at this phase is to allow cancel reactions $\frac{1}{2^i}, -\frac{1}{2^j} \to 0,0$ between arbitrary pairs of exponents with opposite opinions.
        However, this has a positive probability of erroneously eliminating the majority.
        This is because the majority, 
        while it necessarily has larger mass than the minority at this point, 
        could have smaller \emph{count}.
        For example, we could have 16 $A$'s with $\fieldlevel=-2$ and 32 $B$'s with $\fieldlevel=-5$, so $A$'s have mass $16 \cdot 2^{-2} = 4$ and $B$'s have smaller mass $32 \cdot 2^{-5} = 1$, but larger count than $A$.
    }
    
    The biased agents add a Boolean field $\fieldfull$, initially $\False$, 
    and \emph{consumption} reactions that allow an agent at a larger exponent $i$ to consume 
    (set to mass 0 by setting it to be $\unbiased$) 
    an agent at an arbitrarily smaller exponent $j < i$.
    Now the remaining agent represents some non-power-of-two mass $m = 2^i - 2^j$, 
    which it lacks sufficient memory to track exactly. 
    Thus setting the flag $\fieldfull = \True$ corresponds to the agent having an uncertain mass $m$ in the range $2^{i-1} \leq m < 2^{i}$.
    Because of this uncertainty, full agents are not allowed to consume other smaller levels. 
    However, there are more than enough high-exponent majority agents by this phase to consume all remaining lower exponent minority agents.

    Crucially, 
    agents that have consumed another agent and set $\fieldfull=\True$ may themselves then be consumed by a third agent 
    (with $\fieldfull = \False$) 
    at an even larger exponent.
    This is needed because a minority agent at exponent 
    $i \leq -(l+3)$
    may consume a (rare) majority agent at exponent $j < i$, but the minority agent itself can be consumed by another majority agent with exponent $k > i$.

    \item[\phaseConsensusTwo:]
    (\emph{Untimed})
    This is identical to \phaseConsensus:
    it detects whether both biased opinions $\A$ and $\B$ remain.
    If not (the likely case), the algorithm halts,
    otherwise we proceed to the next phase.

    \item[\phaseStableBackup:]
    (\emph{Untimed})
    Agents execute a simple, slow stable majority protocol~\cite{BlackNinjasInTheDark},
    similar to that of~\cite{mertzios2014determining, draief2012convergence} but also handling ties.
    This takes $\Theta(n \log n)$ time, 
    but the probability that an earlier error forces us to this phase is 
    $O(1/n^2)$,
    so it contributes 
    negligibly
    to the total expected time.
\end{description}


\subsection{Simulation of full algorithm}
\label{sec:simulation}

\begin{figure}[!htbp]
    \vspace{-0.5cm}
    \centering
    \includegraphics[trim=130 0 130 0,clip,width=\textwidth]{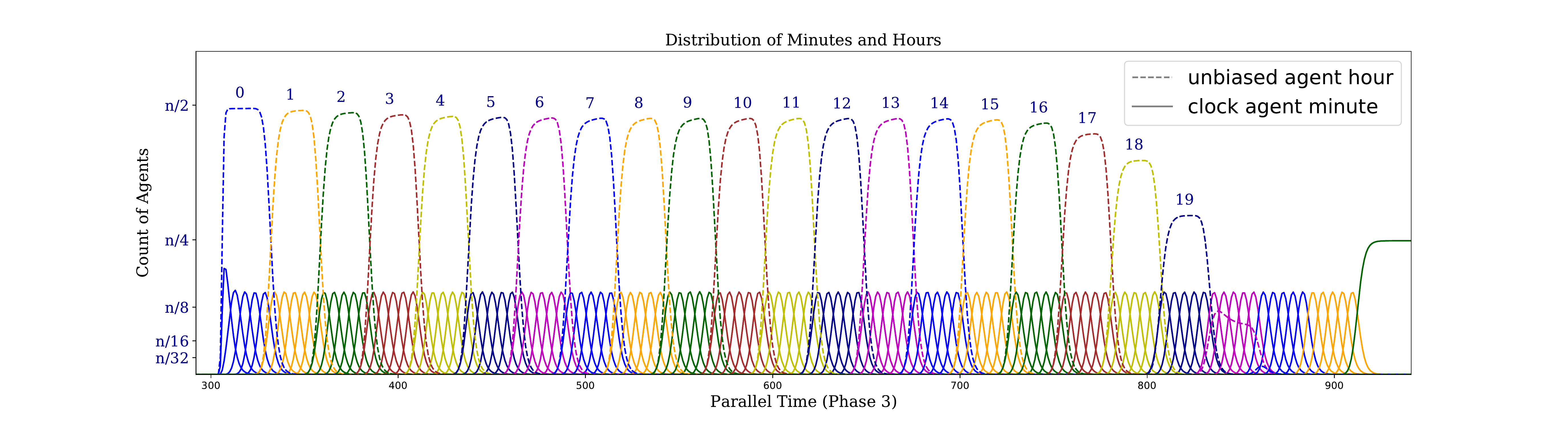}
    \caption{\footnotesize
    Fixed-resolution phase clock used in \phaseMainAveraging,
    with $n \approx 2^{23}$.
    All $\roleclock$ agents set $\fieldminute = 0$ and count up to $k\cdot L$ with special rules defined in \phaseMainAveraging. 
    The solid curves show the distributions of the counts at each value of $\fieldminute$ in the $\roleclock$ agents.
    Blocks of $k = 5$ consecutive minutes correspond to a $\fieldhour$ in the unbiased $\unbiased$ agents with $\fieldrole = \rolemain$ and $\fieldbias = 0$. 
    The dashed curves show the distribution of the counts of each value of $\fieldhour$, with the value written on top.
    We show every $k$ minutes with one color equal to its hour color.
    Near the end of \phaseMainAveraging, the $\unbiased$ count falls to $0$ as the majority count takes over.
    This plot used the value $k = 5$ to emphasize the clock behavior. Later plots will all use the weaker constant $k = 2$.
    In later plots we also omit the $\fieldminute$ distribution and only show the $\fieldhour$ from the $\unbiased$ agents.}
    \label{fig:clock_minute_hour}
\end{figure}

\begin{figure}[!htbp]
     \centering
     \begin{subfigure}[b]{\textwidth}
         \centering
         \includegraphics[trim=130 0 130 0,clip,width=\textwidth]{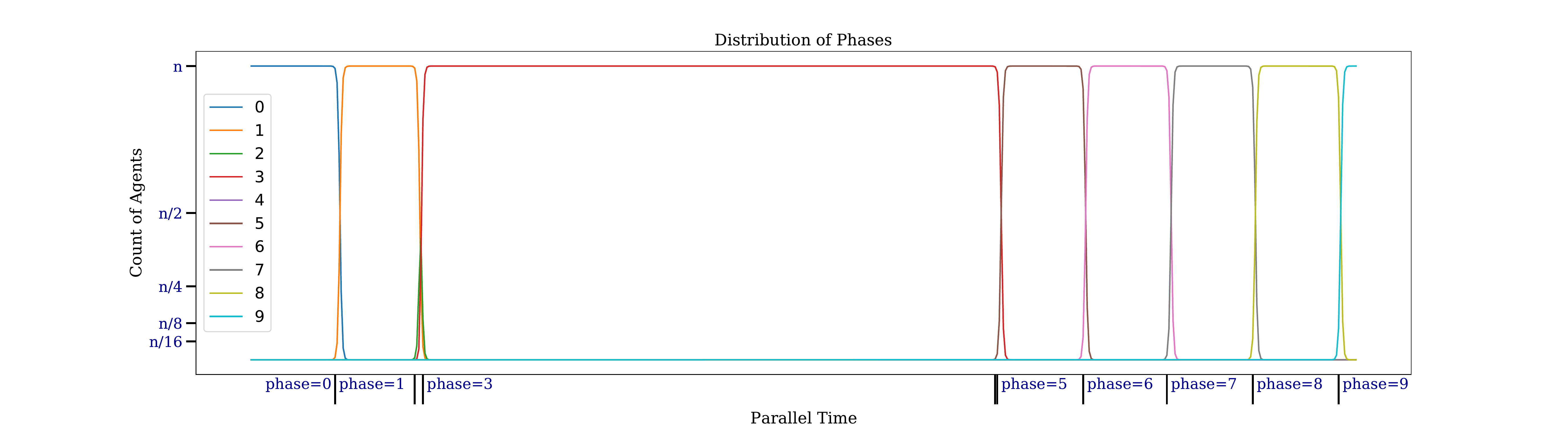}
         \caption{\footnotesize
         The $\fieldphase$ distribution, giving counts of the agents with each value of $\fieldphase$.
         We show these markers on the horizontal axis in later plots. We have set $\max\fieldcounter = 5\log_2(n)$, and removed the counting requirements for $\roleclock$ agents in \phaseInitialize. This makes all timed phases \ref{phase:initialize}, \ref{phase:discrete-averaging}, \ref{phase:reserve-sample}, \ref{phase:reserve-split}, \ref{phase:highminorityelimination}, \ref{phase:lowminorityelimination} take around the same parallel time $2.5\log_2(n)$. The fixed-resolution clock in \phaseMainAveraging\ uses $O(\log n)$ time with a larger constant.
         \phaseConsensus\ and \phaseDetectTie\ are untimed, so they end almost immediately and are not labeled above.}
         \label{fig:phase}
     \end{subfigure}
     \hfill
    
        
    \centering
     \begin{subfigure}[hb]{\textwidth}
         \centering
         \includegraphics[trim=130 0 130 0,clip,width=\textwidth]{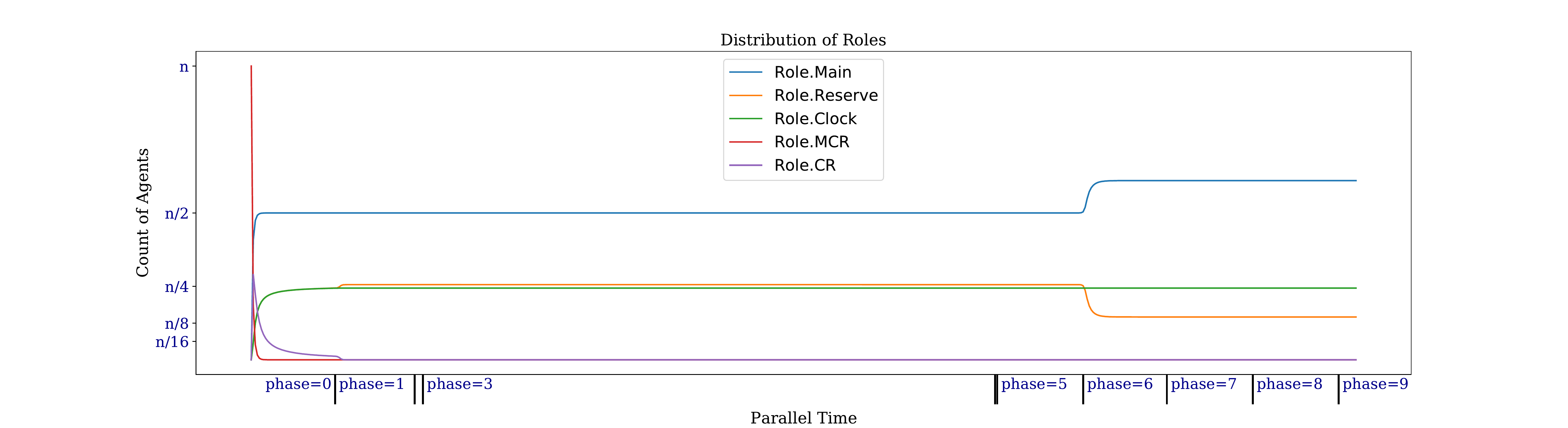}
         \caption{\footnotesize
         The $\fieldrole$ distribution. All agents start in role $\roleMCR$. By the end of \phaseInitialize\, almost all agents decide on a role. They enter \phaseDiscreteAveraging\ with $\approx\frac{n}{2}$, $\approx\frac{n}{4}$, and $\approx\frac{n}{4}$ agents in the respective roles $\rolemain$, $\roleclock$, and $\rolereserve$. An agent's role remain the same in the following phases except \phaseReserveSplit, where split reactions convert $\rolereserve$ agents into $\rolemain$ agents.}
         \label{fig:role}
     \end{subfigure}
     \label{fig:phase-role}
     \caption{The $\fieldphase$ and $\fieldrole$ distributions.}
\end{figure}

\begin{figure}[!htbp]
\ContinuedFloat
\begin{subfigure}[b]{\textwidth}
    \centering
    \includegraphics[trim=130 0 130 0,clip,width=\textwidth]{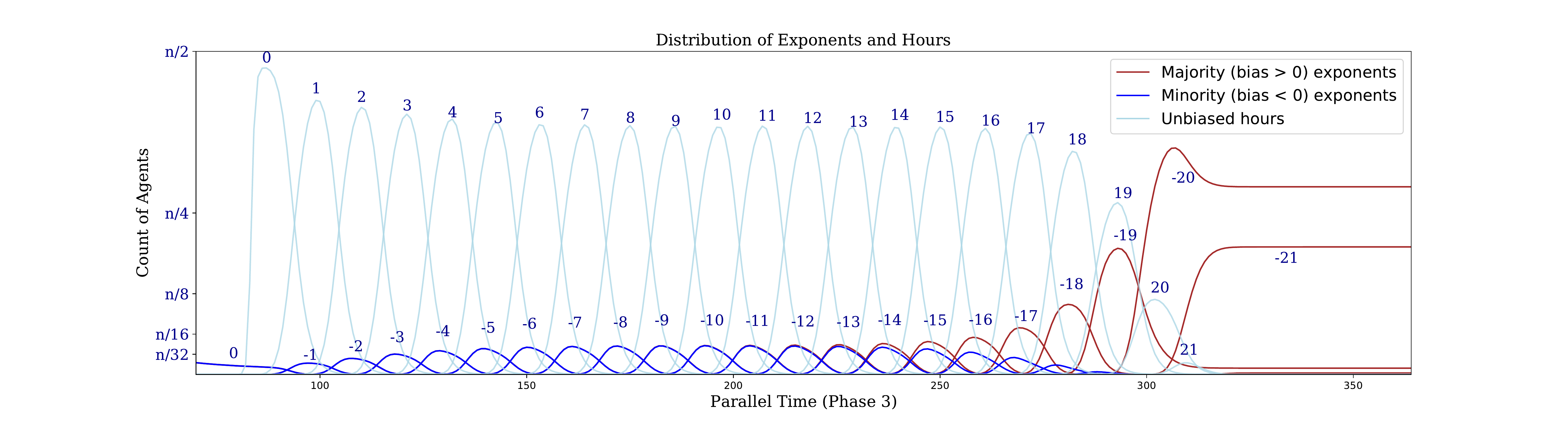}
     \caption{\footnotesize
        The distribution of $\fieldexponent$ and $\fieldhour$ in biased and unbiased $\rolemain$ agents. This plot only shows the time during \phaseMainAveraging, the only time $\fieldhour$ is used. The values for $\fieldexponent$ can decrease from $0$ to $-L$. As described in \phaseMainAveraging, during hour $h$, only agents with $\fieldexponent > -h$ are allowed to split and decrease their $\fieldexponent$ by one. Thus, the changes of $\fieldexponent$ are synchronized with the the changes in the $\fieldhour$ values. In this simulation, the \phaseMainAveraging\ stopped with majority of agents having $3$ consecutive values ($-19, -20, -21$) in their $\fieldexponent$. The $\fieldhour$ values are shown behind the next plot in \cref{fig:opinion-constant}, which shows the totals of these 3 types of $\rolemain$ agents.}
     \label{fig:exponent}
\end{subfigure}

\begin{subfigure}[b]{\textwidth}
    \centering
     \includegraphics[trim=100 40 130 0,clip,width=\textwidth]{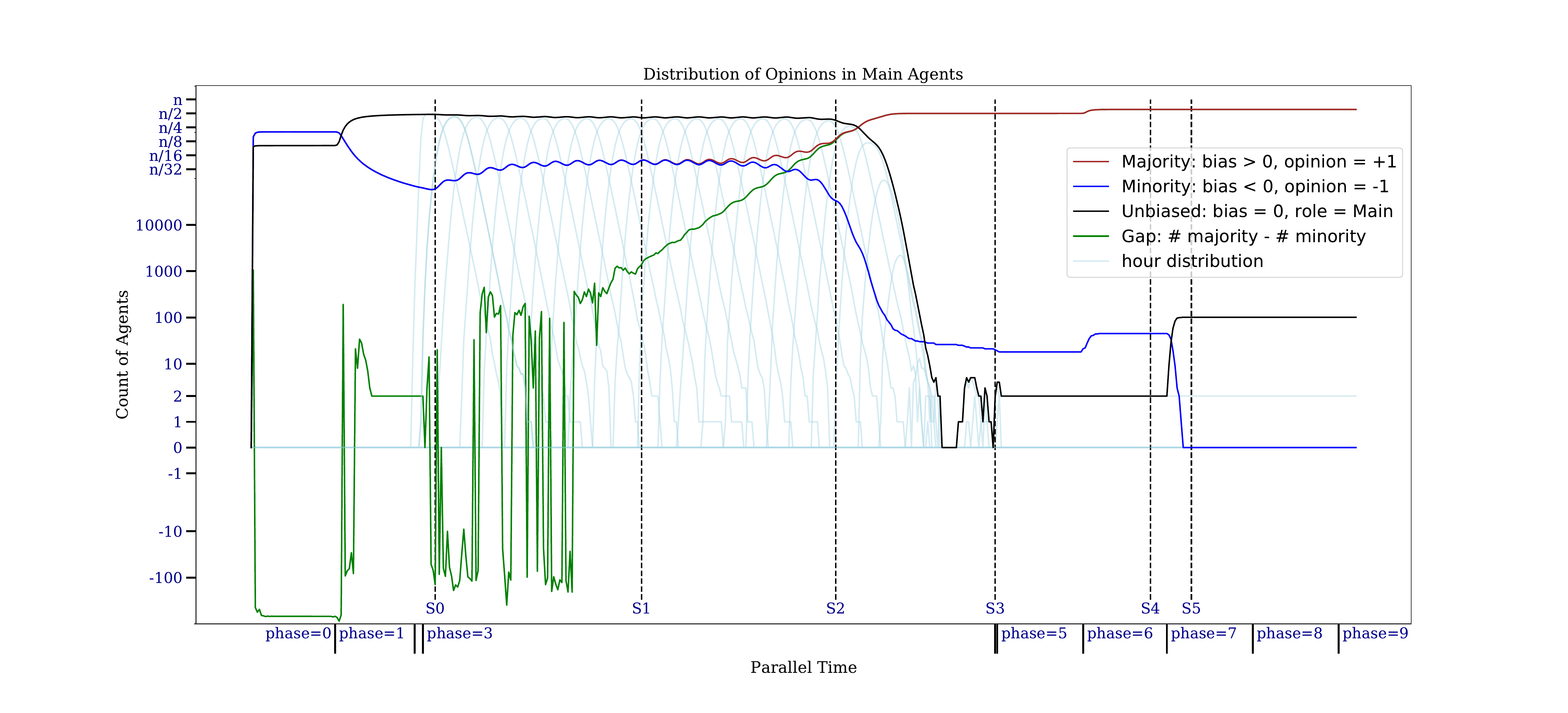}
     \caption{\footnotesize 
     The distribution of $\fieldopinion$ over all $\rolemain$ agents, with count shown on a log scale.
     The green line shows the difference in count between majority and minority agents. Special snapshots at times marked S0, S1, S2, S3, S4, S5 are shown in \cref{fig:snapshots-constant}.
     All $\rolemain$ agents are assigned in \phaseInitialize, with $\fieldbias \in \{0, \pm 1, \pm 2, \pm 3\}$, and all initial biases represented held by this $\rolemain$ subpopulation.
     Here the gap depends randomly on the distribution of $|\fieldbias|\in\{1,2,3\}$.
     Then \phaseDiscreteAveraging\ brings all $\fieldbias \in \{-1, 0, +1\}$, so the gap returns to exactly the initial gap of $2$, and the biased counts decrease polynomially like $\frac{1}{t}$ from cancel reactions.
     In the first part of \phaseMainAveraging\ the gap oscillates randomly about 0, (S0 in \cref{fig:snapshot_0_constant_gap}).
     Once we reach a high enough hour / low enough exponent, 
     the doubling trend takes over and the gap undergoes constant exponential growth (S1 in \cref{fig:snapshot_1_constant_gap}). 
     Finally, this becomes visible as a separation between counts of majority and minority agents (S2 in \cref{fig:snapshot_2_constant_gap}).
     \phaseMainAveraging\ ends with a small but nonzero count of minority agents and the count of unbiased $\unbiased$ agents brought near $0$ (S3 in \cref{fig:snapshot_3_constant_gap}). Then during \phaseReserveSplit, this minority count is amplified slightly by more split reactions bringing the minority exponents down (S4 in \cref{fig:snapshot_4_constant_gap}). 
     During \phaseHighMinorityElimination, additional cancel reactions bring the minority count to $0$ (S5 in \cref{fig:snapshot_5_constant_gap}).
     Since minority agents are gone, 
     \phaseLowMinorityElimination\ has no effect,
     and the protocol stabilizes to the correct majority output in \phaseConsensusTwo.}
     \label{fig:opinion-constant}
\end{subfigure}
    \caption{The case of constant initial gap $g=2$, from simulation with $n = 5122666 \approx 2^{23}$, $p = 0.1$ (drip probability), $k=2$ (number of minutes per hour). The horizontal axis is in units of parallel time, with the ranges corresponding to each phase marked. All agents converge to the correct majority output in \phaseConsensusTwo.}
    \label{fig:constant-gap}
\end{figure}

\begin{figure}[!htbp]
     \centering
     \begin{subfigure}[b]{0.49\textwidth}
         \centering
         \includegraphics[width=\textwidth, trim=10mm 0mm 10mm -2mm, clip=true]{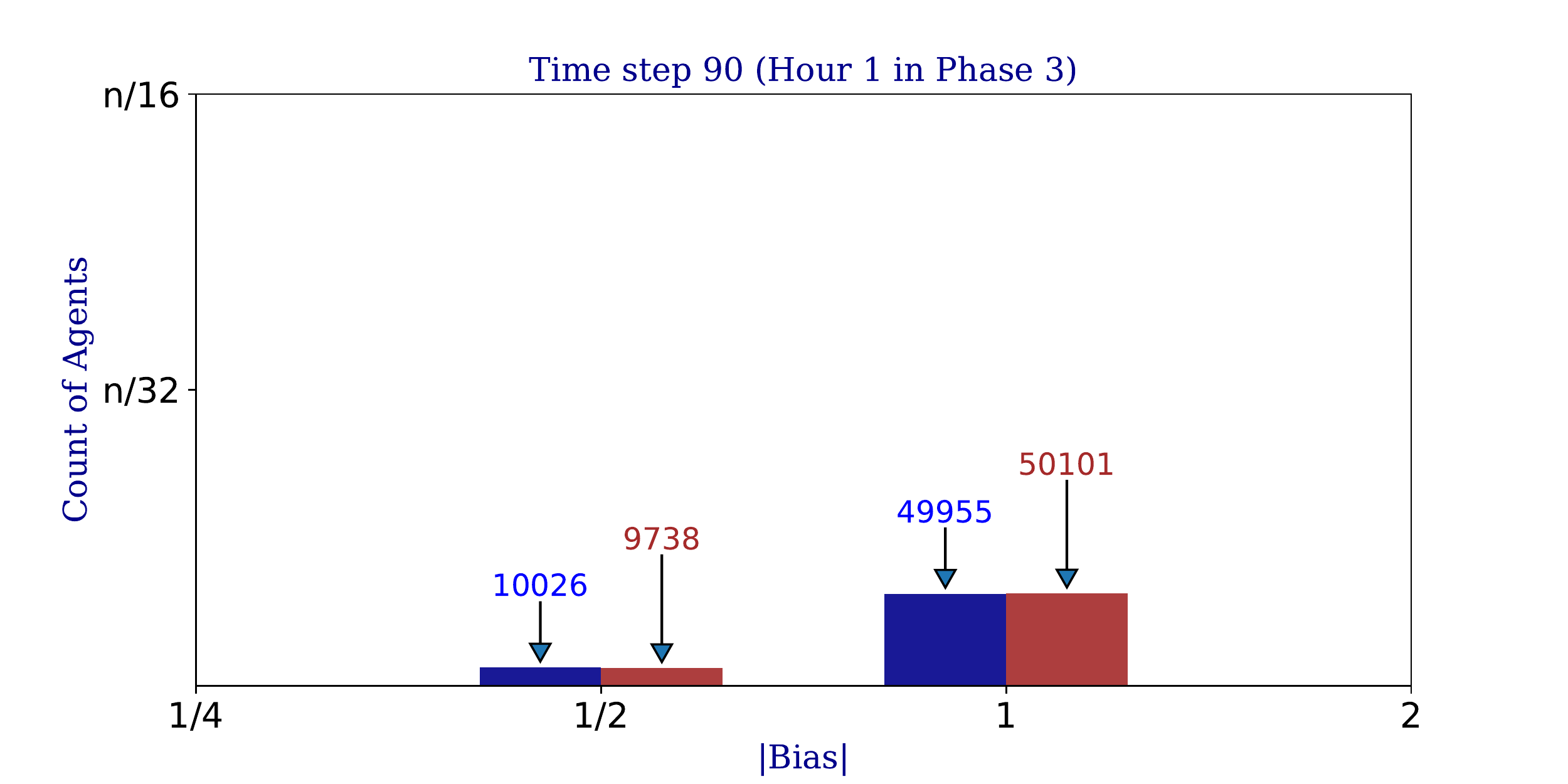}
         \caption{\footnotesize Snapshot S0. At the start of \phaseMainAveraging. The count of minority agents (blue) currently exceeds the count of majority agents because the minorities have done more split reactions. Summing the signed biases, however, gives $+50101 - 49955 + \frac{9738}{2} - \frac{10026}{2} = +146 - 144 = +2$, the invariant initial gap.}
         \label{fig:snapshot_0_constant_gap}
     \end{subfigure}
     \hfill
     \begin{subfigure}[b]{0.49\textwidth}
         \centering
         \includegraphics[width=\textwidth, trim=10mm 0mm 10mm -2mm, clip=true]{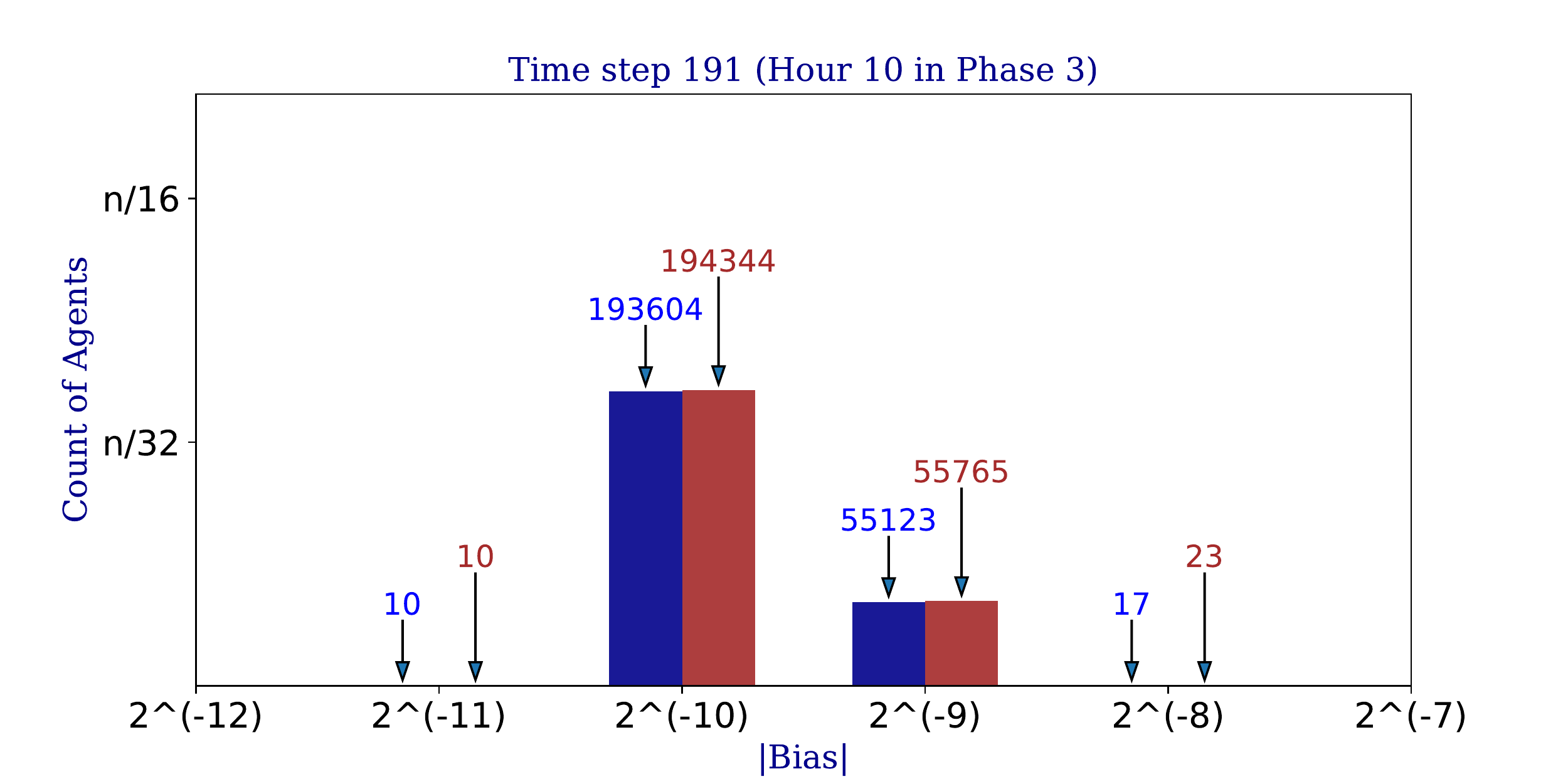}
         \caption{\footnotesize Snapshot S1. A typical distribution at $\fieldhour = 10$, when split reactions have brought most agents to $\fieldexponent = -10$. Only a few $\unbiased$ agents leaked ahead to $\fieldhour = 11$ and enabled splits down to $\fieldexponent = -11$, and no agents have leaked further ahead than this.}
         \label{fig:snapshot_1_constant_gap}
     \end{subfigure}
     \hfill
        
    \centering
     \begin{subfigure}[b]{0.49\textwidth}
         \centering
         \includegraphics[width=\textwidth, trim=10mm 0mm 10mm -2mm, clip=true]{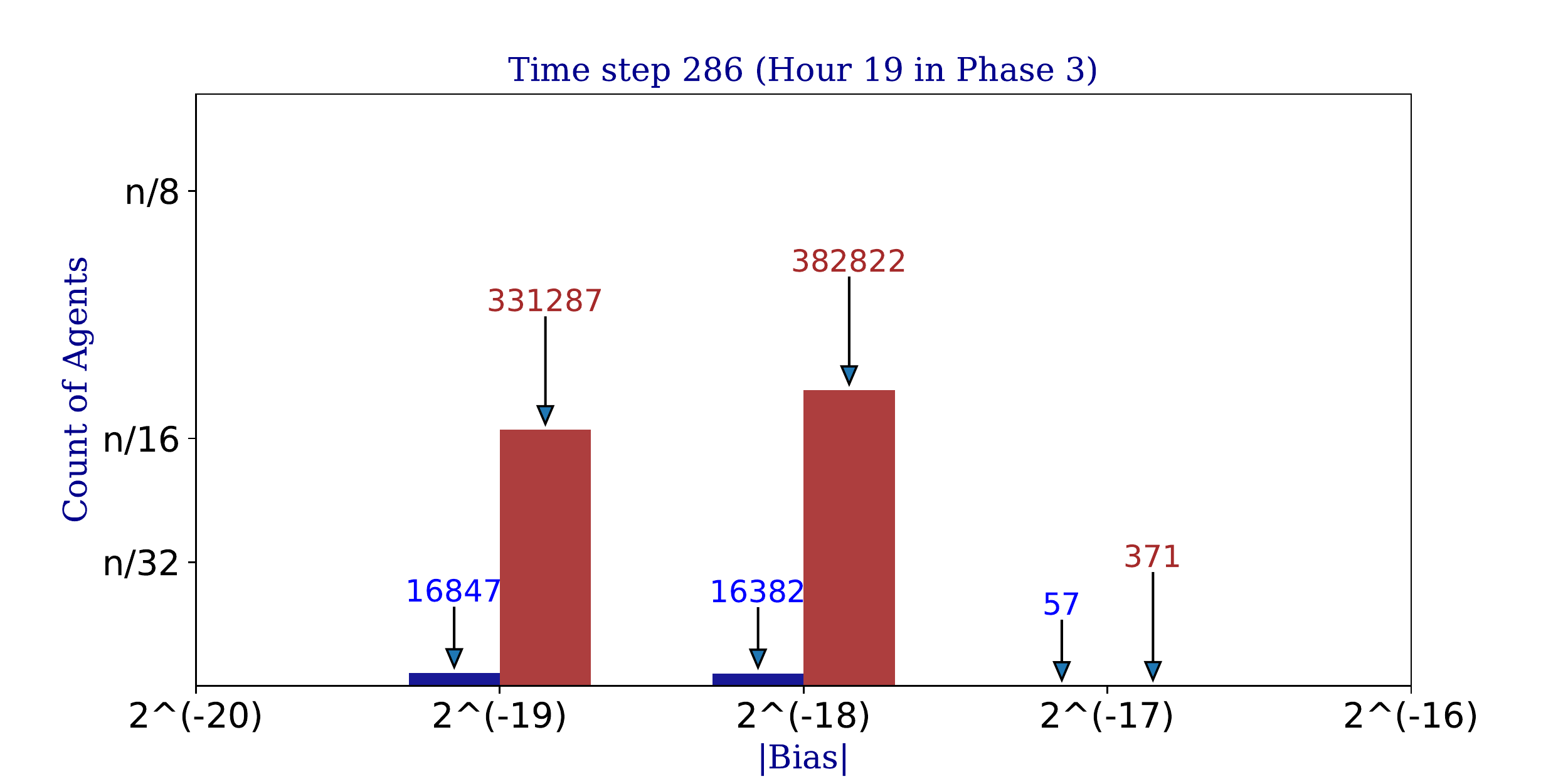}
         \caption{\footnotesize Snapshot S2. We have reached the special exponent $-l = -19$. The count of minority (blue) agents has vastly decreased over the last few hours, and now there will be few more cancel reactions to produce more $\unbiased$ agents.}
         \label{fig:snapshot_2_constant_gap}
     \end{subfigure}
     \hfill
     \begin{subfigure}[b]{0.49\textwidth}
         \centering
         \includegraphics[width=\textwidth, trim=10mm 0mm 20mm -2mm, clip=true]{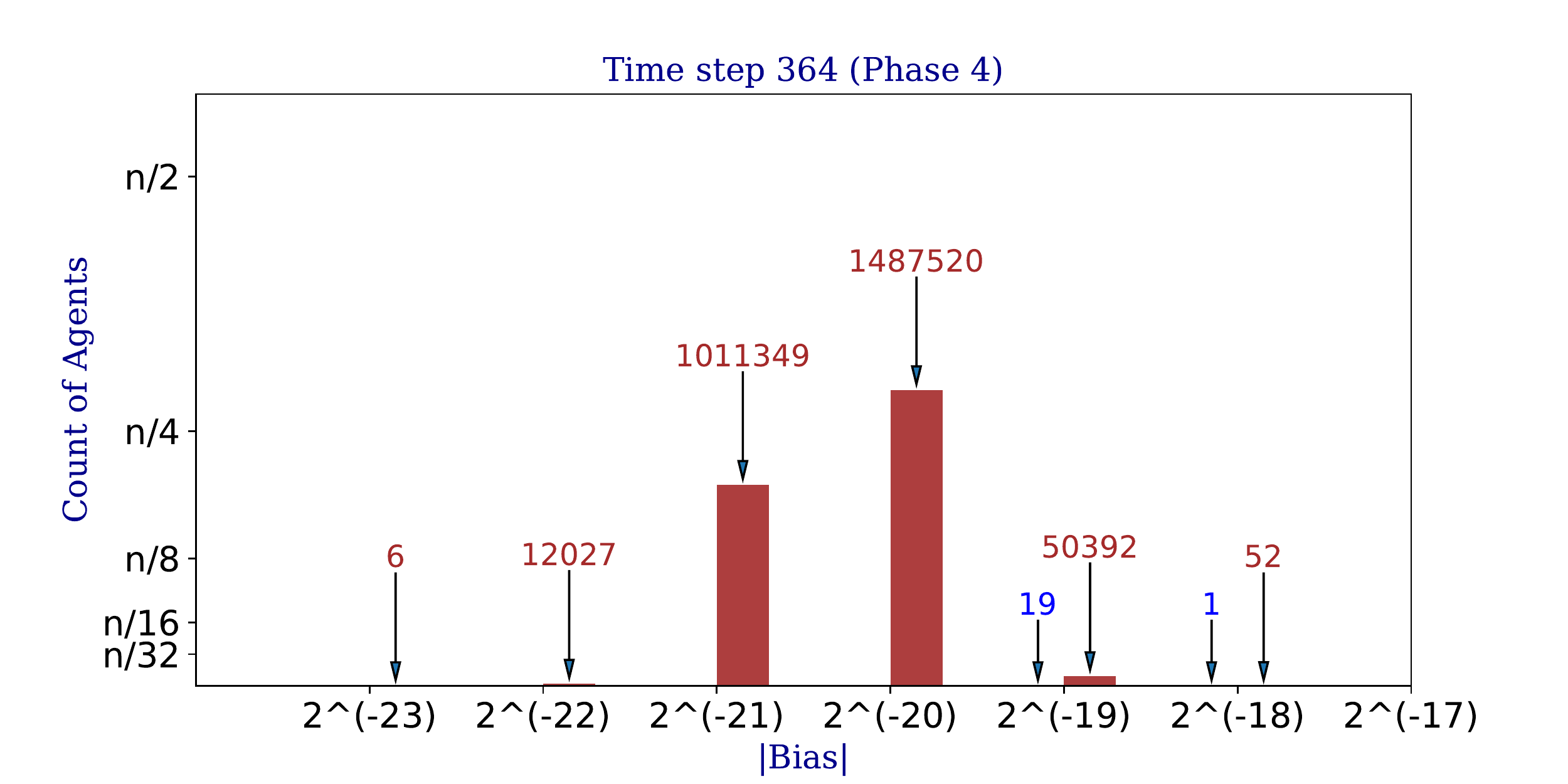}
         \caption{\footnotesize Snapshot S3. We end \phaseMainAveraging\ with most $\rolemain$ agents with the majority opinion, and $\fieldbias\in\{-19,-20,-21\}$ in a range of 3 consecutive values, as shown in \cref{thm:phase-3-majority-result}.
         Only a few minority agents are left.}
         \label{fig:snapshot_3_constant_gap}
     \end{subfigure}
        
    \centering
     \begin{subfigure}[b]{0.49\textwidth}
         \centering
         \includegraphics[width=\textwidth, trim=10mm 0mm 20mm -2mm, clip=true]{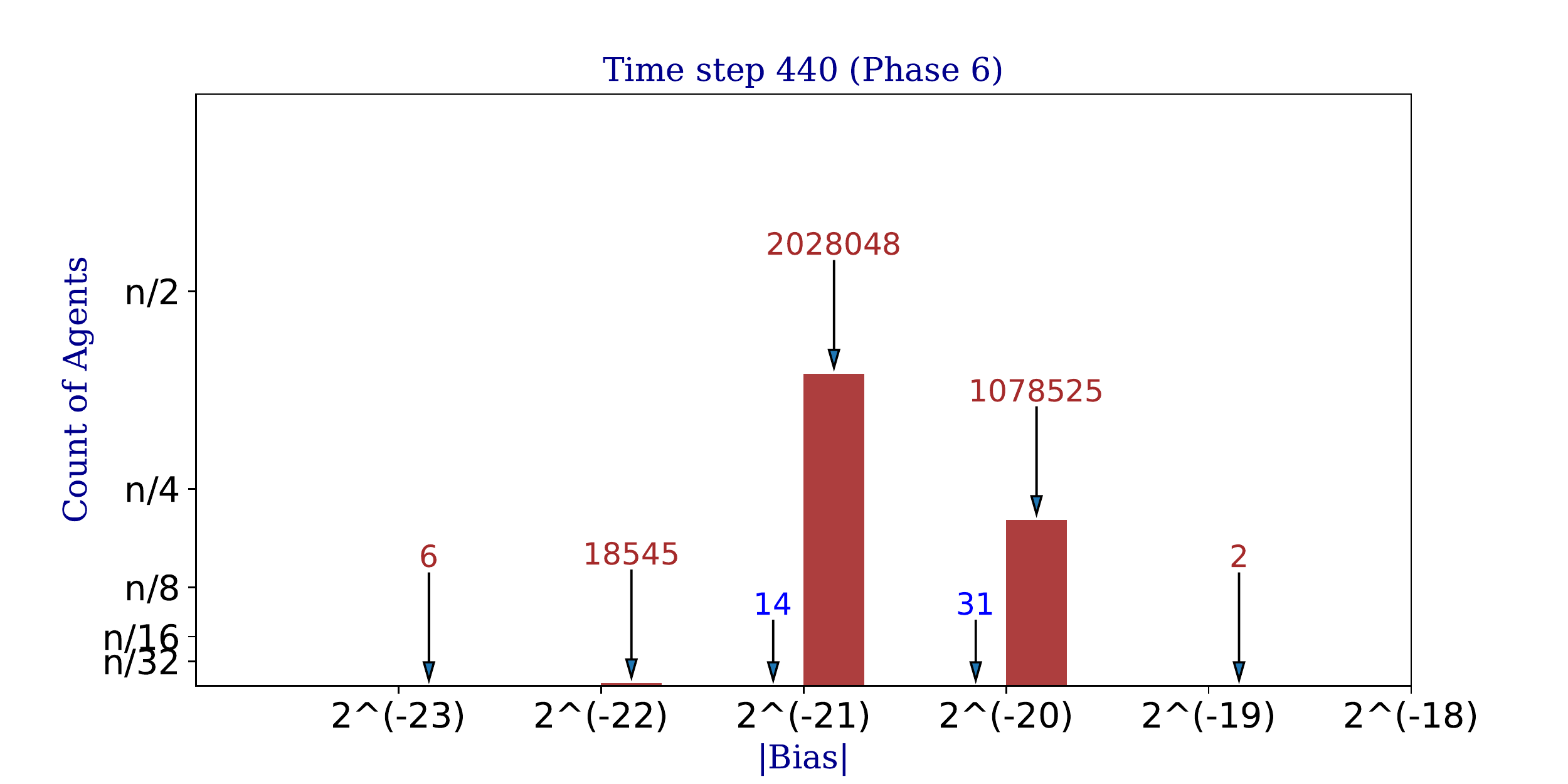}
         \caption{\footnotesize Snapshot S4. After \phaseReserveSplit, where $\rolereserve$ agents with $\fieldsample\in\{-19,-20,-21\}$ enabled additional split reactions that brought all minority agents down.}
         \label{fig:snapshot_4_constant_gap}
     \end{subfigure}
     \hfill
     \begin{subfigure}[b]{0.49\textwidth}
         \centering
         \includegraphics[width=\textwidth, trim=10mm 0mm 20mm -2mm, clip=true]{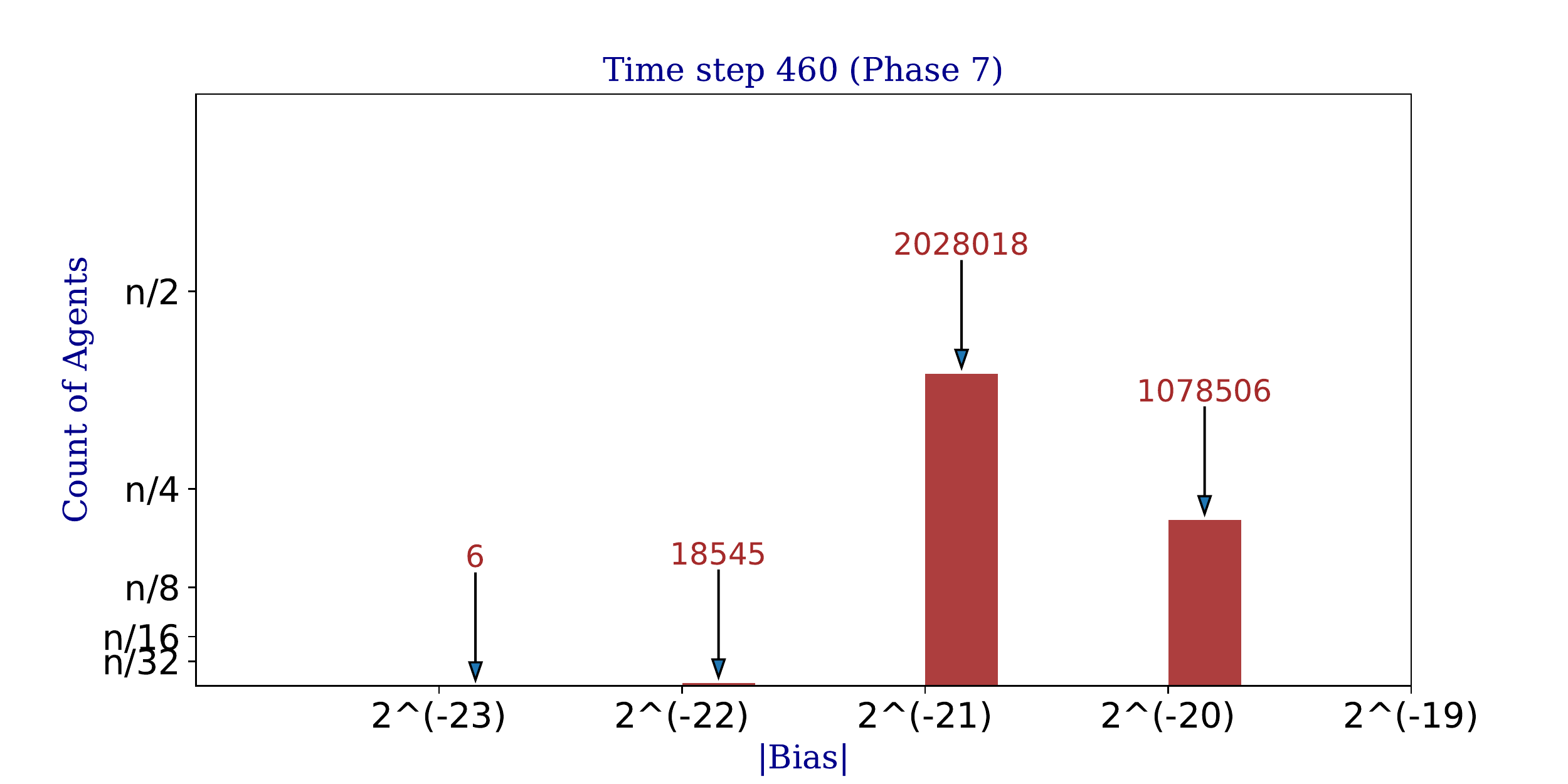}
         \caption{\footnotesize Snapshot S5. After \phaseHighMinorityElimination, additional generalized cancel reactions eliminate all minority agents at the higher $\fieldexponent$ values. At this point, there are no minority agents left, and we will stabilize to the correct output in \phaseConsensusTwo.}
         \label{fig:snapshot_5_constant_gap}
     \end{subfigure}
        \caption{Snapshots S0, S1, S2, S3, S4, S5 from \cref{fig:opinion-constant}, showing the distribution of $\fieldbias$ among agents with $|\fieldbias| > 0$. The red and blue bars give the count of $A$ (majority, with $\fieldbias > 0$) and $B$ (minority, with $\fieldbias < 0$) agents respectively. The exact counts are written above the bars, and everywhere not explicitly written the count of is 0.
        See the link to animations of these plots at the end of this section.
        }
        \label{fig:snapshots-constant}
\end{figure}


\begin{figure}[!htbp]
     \centering
     \begin{subfigure}[b]{0.99\textwidth}
         \centering
         \includegraphics[trim=30 4 30 0,clip,width=\textwidth]{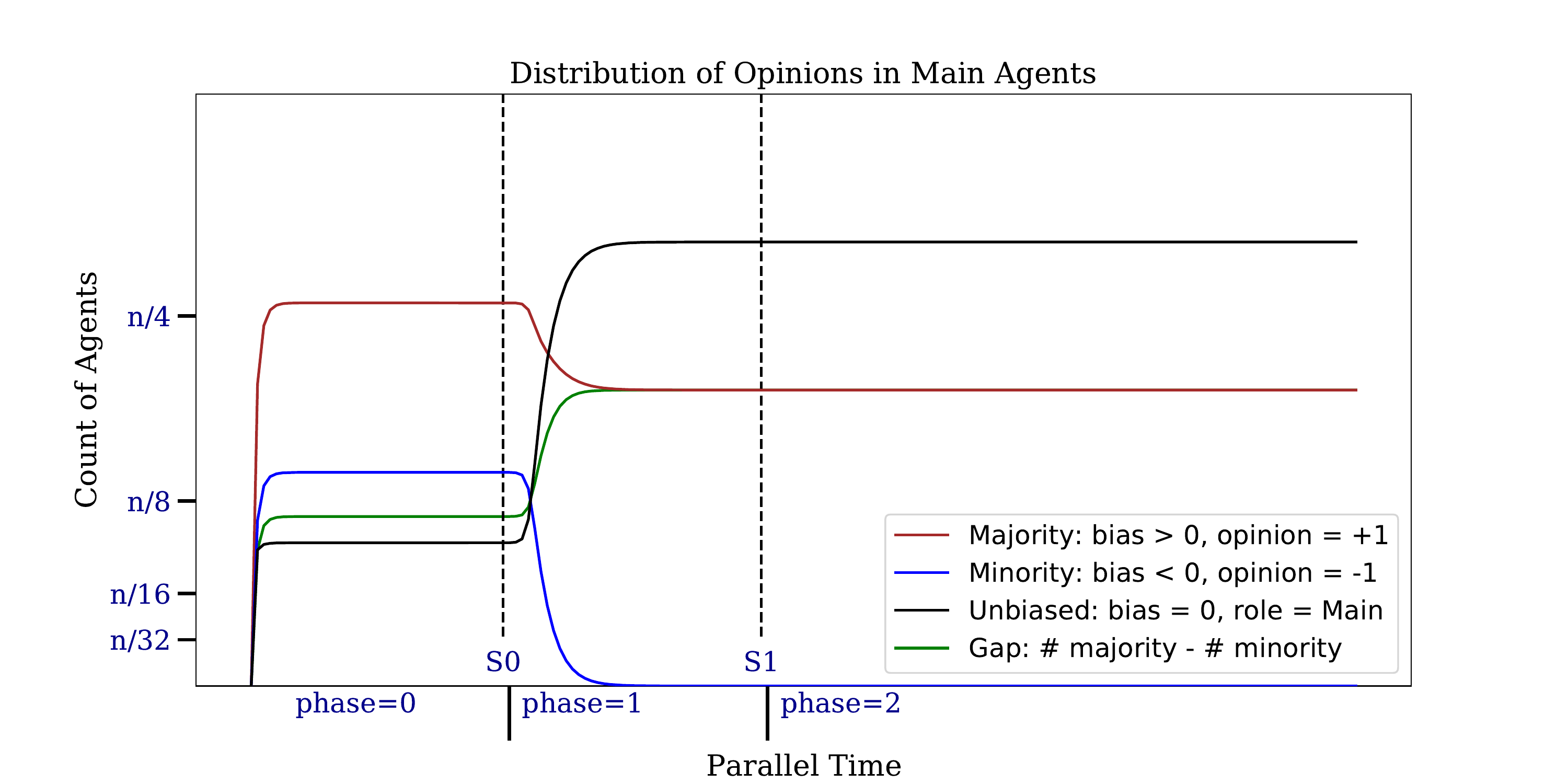}
         \caption{\footnotesize 
         The distribution of $\fieldopinion$ in the $\rolemain$ agents in case of a linear size gap $g=n/10$.
         The red line gives the majority count ($\fieldopinion = +1$), 
         the blue line the minority count ($\fieldopinion = -1$), and the green line their difference. The black line gives the unbiased $\unbiased$ agents.
         }
         \label{fig:opinion_linear}
     \end{subfigure}
     \hfill
     \centering
     \begin{subfigure}[b]{0.49\textwidth}
         \centering
         \includegraphics[width=\textwidth, trim=10mm 0mm 20mm -2mm, clip=true]{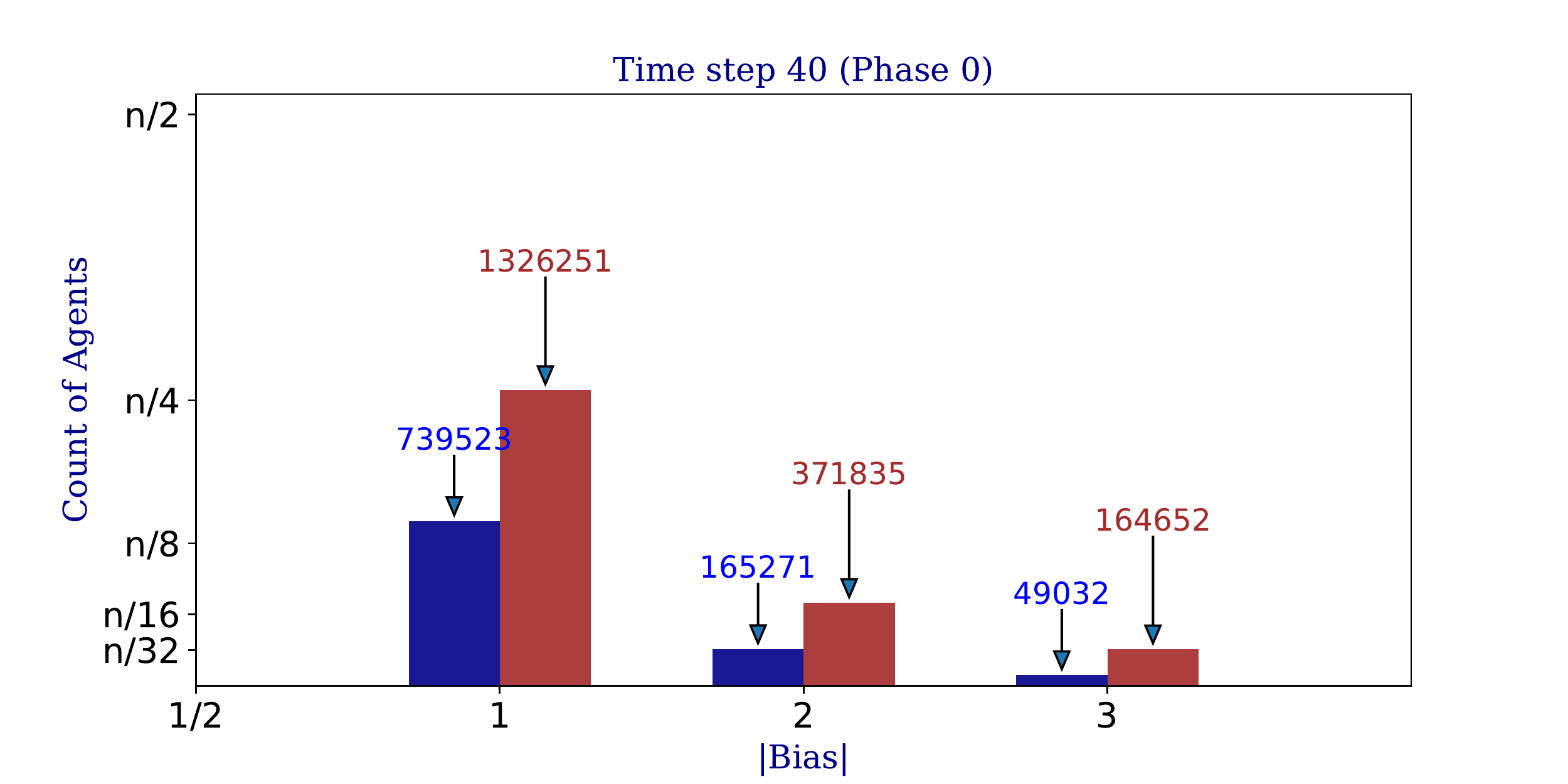}
         \caption{\footnotesize Snapshot S0. At the end of \phaseInitialize, the biased agents have $|\fieldbias|\in\{1,2,3\}$.}
         \label{fig:snapshot_0_linear_gap}
     \end{subfigure}
     \hfill
     \begin{subfigure}[b]{0.49\textwidth}
         \centering
         \includegraphics[width=\textwidth, trim=10mm 0mm 20mm -2mm, clip=true]{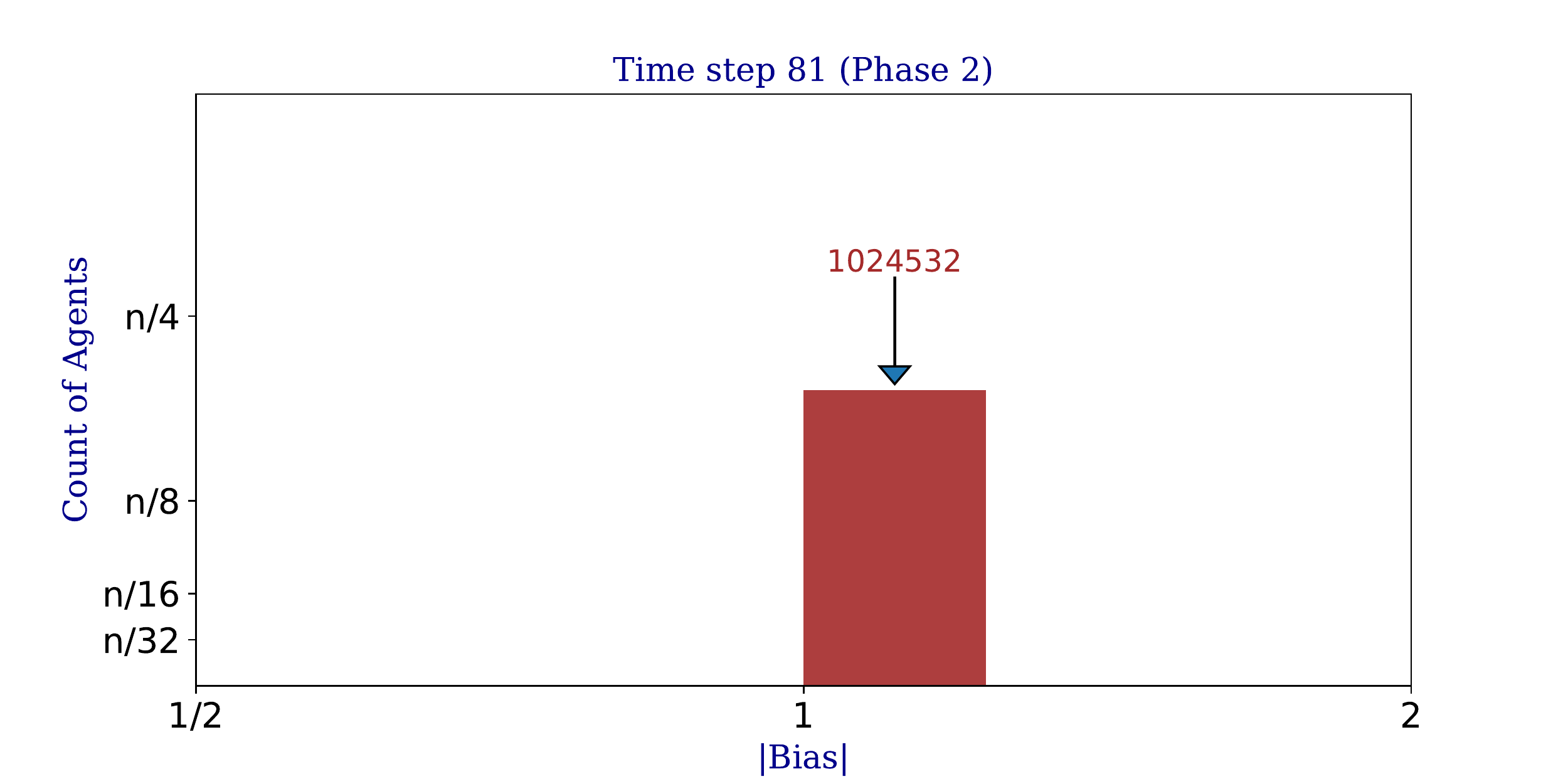}
         \caption{\footnotesize Snapshot S1. The discrete averaging of \phaseDiscreteAveraging\ first brings all $\fieldbias\in\{-1,0,+1\}$. Then all $\fieldbias = -1$ cancel, leaving $\fieldbias = +1$ as the only nonzero bias. With the minority opinion eliminated, we converge in \phaseConsensus.}
         \label{fig:snapshot_1_linear_gap}
     \end{subfigure}
     \hfill
        \caption{The case of linear size gap $g = n/10$, again with $n = 5122666 \approx 2^{23}$, $p  = 0.1$, $k=2$.
        With large initial gap, the simulation converges in \phaseConsensus. \cref{fig:opinion_linear} shows the counts of each opinion among the population of $\rolemain$ agents.
         Two special times S0 and S1, and the configurations of biased agents at these snapshots are shown in \cref{fig:snapshot_0_linear_gap,fig:snapshot_1_linear_gap}.
         See link to animations of these plots at the end of this section.}
        \label{fig:simulation-linear}
\end{figure}

\begin{figure}[!htbp]
     \centering
     \begin{subfigure}[b]{0.99\textwidth}
         \centering
         \includegraphics[trim=30 4 30 0,clip,width=\textwidth]{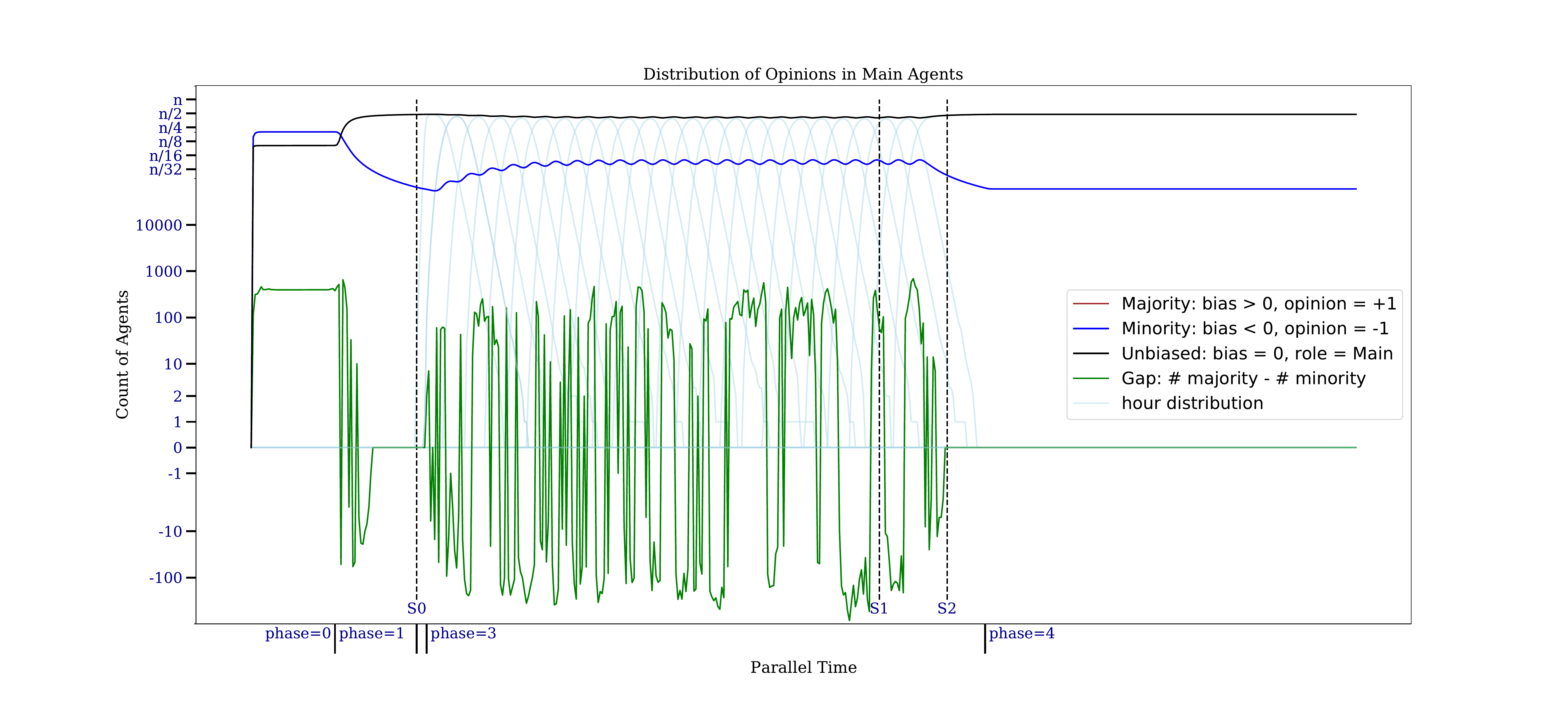}
         \caption{\footnotesize 
         The distribution of $\fieldopinion$ in the $\rolemain$ agents in case of a tie.
         The red line gives the majority count ($\fieldopinion = +1$), 
         the blue line the minority count ($\fieldopinion = -1$),
         though they overlap everywhere, 
         and the green line their difference. 
         The black line gives the unbiased $\unbiased$ agents.
         In \phaseMainAveraging\ we see the same qualitative behavior as the first part of \phaseMainAveraging\ with a constant initial gap in \cref{fig:opinion-constant}. Now this continues the whole phase, with the gap in counts oscillating about $0$ until finally reaching $0$ when all biased agents have $\fieldexponent = -23$ at time S2.
         With no $\fieldexponent > -23$, we stabilize to output $\T$ in \phaseDetectTie.
         }
         \label{fig:opinion_tie}
     \end{subfigure}
     \hfill
     \centering
     \begin{subfigure}[b]{0.49\textwidth}
         \centering
         \includegraphics[width=\textwidth, trim=10mm 0mm 20mm -2mm, clip=true]{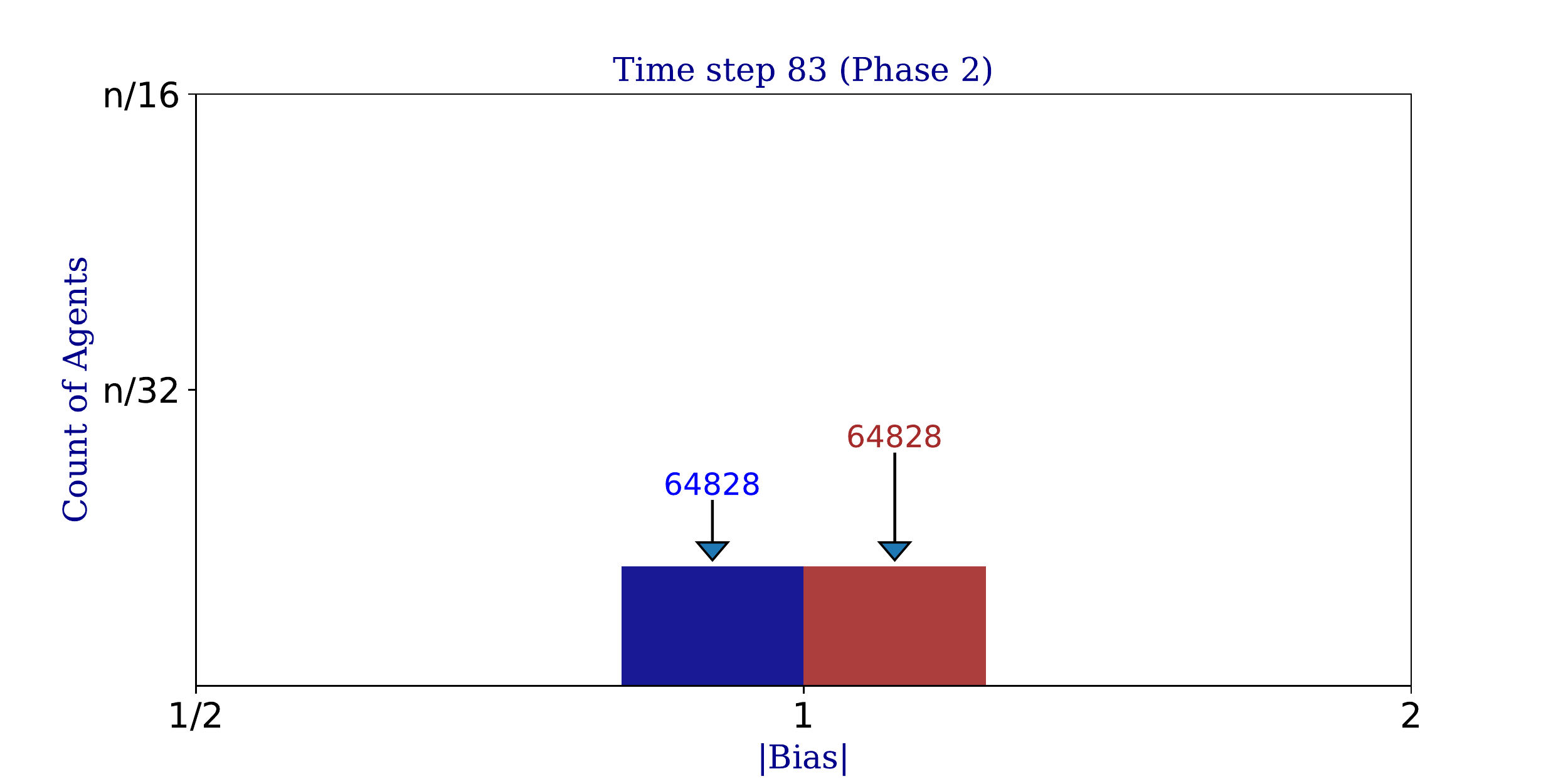}
         \caption{\footnotesize Snapshot S0. We enter \phaseMainAveraging\ with an equal number of $\fieldbias = \pm 1$. }
         \label{fig:snapshot_0_tie_gap}
     \end{subfigure}
     \hfill
     \begin{subfigure}[b]{0.49\textwidth}
         \centering
         \includegraphics[width=\textwidth, trim=10mm 0mm 20mm -2mm, clip=true]{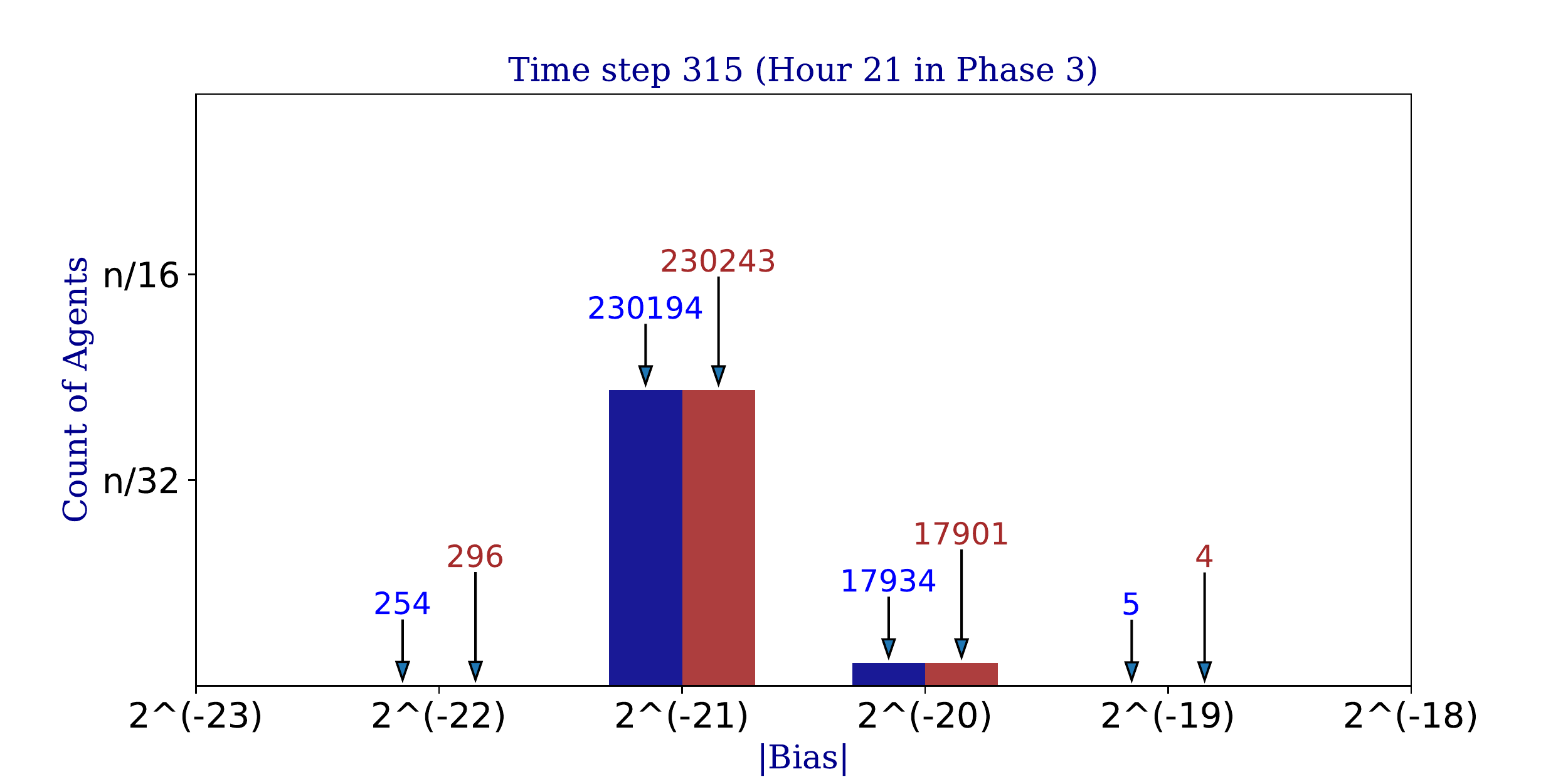}
         \caption{\footnotesize Snapshot S1. We see similar distributions as in \cref{fig:snapshot_1_constant_gap} all the way until the end of \phaseMainAveraging.}
         \label{fig:snapshot_1_tie_gap}
     \end{subfigure}
     \hfill
    \begin{subfigure}[b]{0.49\textwidth}
         \centering
         \includegraphics[width=\textwidth, trim=10mm 0mm 20mm -2mm, clip=true]{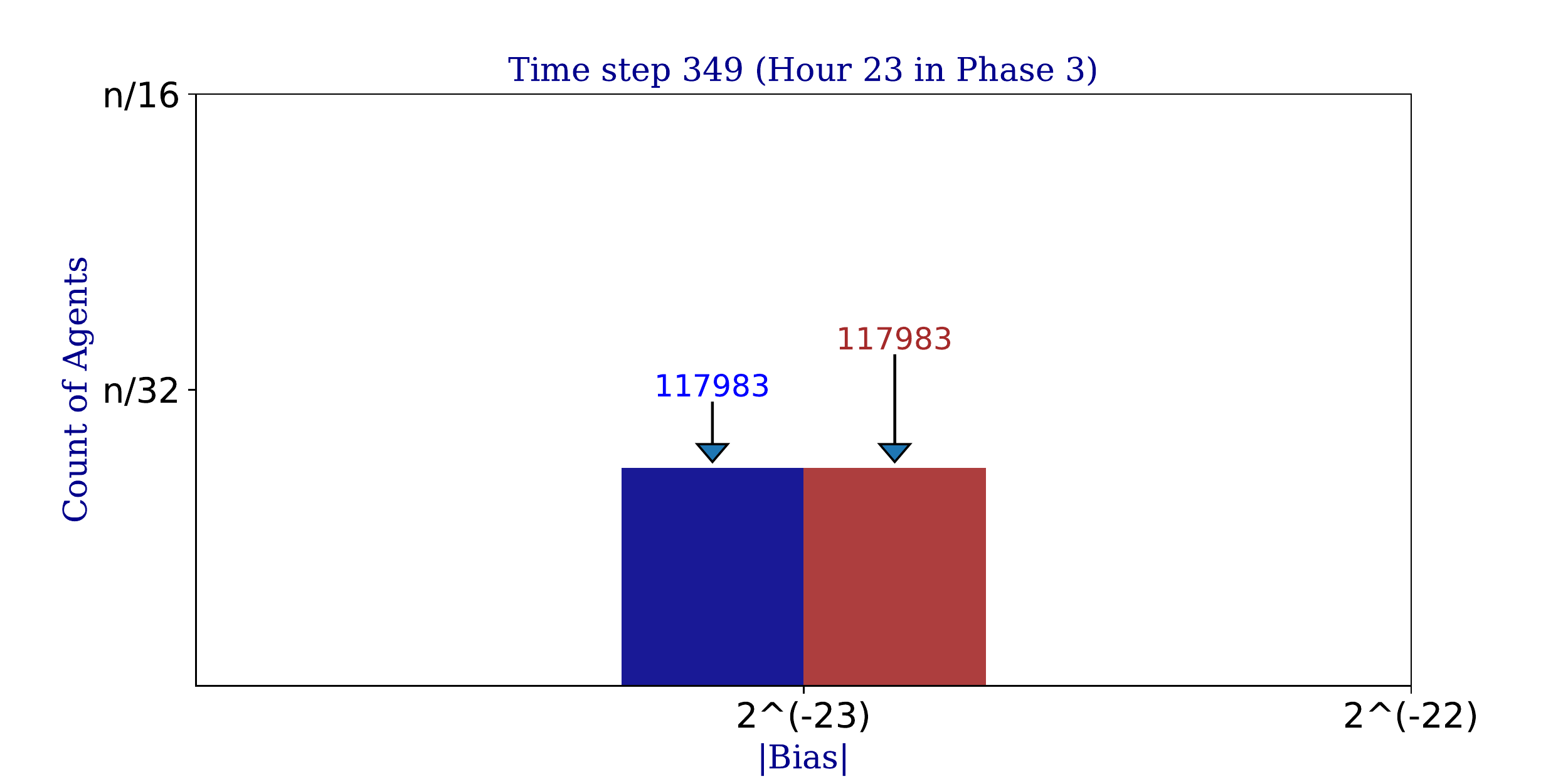}
         \caption{\footnotesize Snapshot S2. After reaching synchronous hour 23, split reactions bring all remaining biased agents down to $\fieldexponent = -23$, which is only possible with an initial tie.
         }
         \label{fig:snapshot_2_tie_gap}
     \end{subfigure}
     \hfill
        \caption{\footnotesize
        The case of a tie, with initial gap $g = 0$, again with $n = 5122666 \approx 2^{23}$, $p  = 0.1$, $k=2$.
        The simulation converges in \phaseConsensus. \cref{fig:opinion_linear} shows the counts of each opinion among the population of $\rolemain$ agents.
         Three special times S0, S1, S2, and the configurations of biased agents at these snapshots are shown in Figures~\ref{fig:snapshot_0_tie_gap}, \ref{fig:snapshot_1_tie_gap}, \ref{fig:snapshot_2_tie_gap}.
         See the link to animations of these plots at the end of this section.
         }
        \label{fig:simulation-tie}
\end{figure}

In this section we show simulation results, where the complete pseudocode of \cref{subsec:pseudocode} was translated into Java code available on 
GitHub~\cite{github_java_simulations}.
In these simulations, we stop the protocol once all agents reach \phaseConsensusTwo. For the low probability case that agents switch to the \stableBackup, the simulator prints an error indicating the agents should switch to slow back up, but as expected this was not observed in our simulations.
For all our plots, we collect data from simulations with $n \approx 2^{23}$, $p $ (drip probability) $ = 0.1 $. The first simulation in \cref{fig:clock_minute_hour} shows the relationship between $\fieldminute$ of $\roleclock$ agents and $\fieldhour$ of $\rolemain$ agents. 
Here we used $k = 5$ minutes per hour to show clearly the relationship and the discrete nature of the hours.

All remaining simulations used the even weaker value $k = 2$, to help see enough low probability behavior that the logic enforcing probability-1 correctness in later phases is necessary. 
We show 3 simulations corresponding to the 3 different types of initial gap. 
\cref{fig:constant-gap,fig:snapshots-constant} show constant initial gap $g = +2$. 
This is our ``typical'' case, where the simulation eventually stabilizes to the correct output in \phaseConsensusTwo. 
\cref{fig:simulation-linear} shows linear initial gap $g \approx \frac{n}{10}$, which stabilizes in \phaseConsensus\ after quickly cancelling all minority agents. 
\cref{fig:simulation-tie} shows initial tie $g = 0$, which stabilizes in \phaseDetectTie\ after all biased agents reach the minimum $\fieldexponent = -L = -23$.

All three simulations show various snapshots of configurations of the biased agents. These particular snapshots are at special times marked in \cref{fig:opinion-constant,fig:opinion_linear,fig:opinion_tie}. In all cases, an animation is available at GitHub~\cite{github_animation_simulations},
showing the full evolution of these distributions over all recorded time steps from the simulation.

\subsection{Algorithm pseudocode}
\label{subsec:pseudocode}

In this section we give a full formal description of the main algorithm.

Every agent starts with a read-only field $\fieldinput\in\{\A,\B\}$, a field $\fieldoutput \in \{\A,\B,\T\}$ corresponding to outputs that the majority is $\A$, $\B$, or a tie. The protocol is broken up into 11 consecutive phases, marked by the additional field $\fieldphase=0\in\{0,\ldots,10\}$. The phase updates via the epidemic reaction
    $u.\fieldphase, v.\fieldphase \gets \max(u.\fieldphase, v.\fieldphase).$
Some fields are only used in particular phases, to ensure the total state space is $\Theta(\log n)$.\footnote{
    Note that using two fields,
    both with $O(\log n)$ possible values,
    requires $O(\log^2 n)$ states,
    not $O(\log n)$.
} 
Such fields and the initial behavior of an agent upon entering a phase are described in the $\textbf{Init}$ section above each phase. 
Whenever an agent increments their $\fieldphase$, they execute $\textbf{Init}$ for the new phase (and sequentially for any phases in between if they happen to increment $\fieldphase$ by more than 1). We refer to the ``end of phase $i$'' to mean the time when the first agent sets $\fieldphase \gets i+1$.
Note that the agents actually enter each new phase by epidemic, so there is technically no well-defined ``beginning of phase $i$''. 
To simplify the analysis, we formally start our arguments for each phase assuming each agent is in the current phase,
although technically $\Theta(\log n)$ time will pass between the time the first agent enters phase $i$ and the last agent does.

Each timed phase $i$ each requires setting agents to count from $c_i \ln n$ down to 0, where the minimum required value of $c_i$ depends on the phase.
These constants can be derived from the technical analysis in Sections~\ref{sec:analysis-early-phases}-\ref{sec:analysis-final-phases}
but for brevity we avoid giving them concrete values in the pseudocode.
Our simulations (\cref{sec:simulation}) that used the same small constant $5\log_2(n)$ for all counters seem to work, but the proofs require larger constants to ensure the necessary behavior within each phase can complete with high probability $1-O(1/n^2)$. By increasing these constants $c_i$ (along with changing the phase clock constants $p$, $k$; see \phaseMainAveraging\ and \cref{thm:clock}), we could also push high probability bound to $1-O(1/n^c)$ for any desired constant $c$. For concreteness, use $1-O(1/n^2)$ for most high probability guarantees, since this is large enough to take appropriate union bounds and ensure the extra time from low probability failures does not contribute meaningfully to the total $O(\log n)$ time bound.

\begin{algorithm}
\caption*{\textbf{Nonuniform Majority}$_L(u,v)$.
Nonuniform majority algorithm for population sizes $n$ with $L = \ceil{\log n}$.
\\
{\bf Init:}
$\fieldphase \gets 0 \in \{0,1,\ldots,10\}$ and execute {\bf Init} for \phaseInitialize.
}
\label{alg:majority-nonuniform}
\begin{algorithmic}[1]
\If{$i.\fieldphase < j.\fieldphase$ where $\{i,j\} = \{u,v\}$}
    \For{$p = \{ i.\fieldphase + 1, \ldots, j.\fieldphase \}$}
        \State{execute {\bf Init} for {\bf Phase} $p$ on agent $i$}
    \EndFor
    \State{$i.\fieldphase \gets j.\fieldphase$}
\EndIf
\State{execute {\bf Phase} $u.\fieldphase(u,v)$}
\end{algorithmic}
\end{algorithm}

\phaseInitialize\ is a timed phase that splits the population into three subpopulations:
$\rolemain$ to compute majority,
$\roleclock$ to time the phases and the movement through exponents in \phaseMainAveraging,
and
$\rolereserve$ to aid in cleanup during \phaseReserveSplit.
An agent can only move into role $\roleclock$ or $\rolereserve$ by ``donating'' its opinion to a $\rolemain$ agent,
who can collect up to two other opinions in addition to their own,
leading to a \emph{bias} of up to $\pm 3$.
After this phase,
the populations of the three roles are near the expected one quarter $\rolemain$, one quarter $\roleclock$, and one half $\rolereserve$.
\cref{lem:phase-initialize} shows that all initial opinions have been given to assigned $\rolemain$ agents and these subpopulations are near their expected fractions, both with high probability.

It is likely that we could use a simpler population splitting scheme. 
For example, we could simply have each pair of agents with initially opposite opinions change to roles $\roleclock$ and $\rolereserve$.
However, this would mean the number of agents in the role $\rolemain$ after \phaseInitialize\ would depend on the initial gap.
The current method of population splitting gives us stronger guarantees on the number of each agent in each role,
simplifying subsequent analysis.

\begin{phase}[H]
\caption{Initialize Roles. Agent $u$ interacting with agent $v$.
\\
\textbf{Init} 
$\fieldrole \gets \roleMCR \in\{\rolemain,\roleclock,\rolereserve,
\roleMCR,
\roleCR\}$
\\
$\fieldassigned \gets \False \in \{\True,\False\}$
\\
if $\fieldinput = \A$, $\fieldbias \gets +1 \in \{-3,-2,-1, 0, +1, +2, +3\}$
\\
if $\fieldinput = \B$, $\fieldbias \gets -1 \in \{-3,-2,-1, 0, +1, +2, +3\}$
\\
we always maintain the invariant $\fieldopinion = \text{sign}(\fieldbias)\in\{-1,0,+1\}$
\\
if $\fieldrole = \roleclock$, $\fieldcounter \gets c_0\ln n \in \{0,\ldots,c_0\ln n\}$ only used in the current phase
}
\label{phase:initialize}
\begin{algorithmic}[1]
\If{$u.\fieldrole = v.\fieldrole = \roleMCR$}
\Comment{Allocate $\approx \frac{1}{2} \rolemain$ agents}
    \State{$u.\fieldrole \gets \rolemain$; $u.\fieldbias \gets u.\fieldbias + v.\fieldbias$}
    \label{line:phase-0-assign-main}
    \State{$v.\fieldbias \gets 0$; $v.\fieldrole \gets \roleCR$}
    \Comment{$v$ won't use its bias subsequently}
\EndIf
\If{$i.\fieldrole = \roleMCR$, $j.\fieldrole = \rolemain$, $j.\fieldassigned = \False$ where $\{i,j\} = \{u,v\}$}
    \State{$j.\fieldassigned \gets \True$; $j.\fieldbias \gets j.\fieldbias + i.\fieldbias$}
    \label{line:phase-0-main-takes-bias}
    \Comment{$\rolemain$ agents can assign 1 non-$\rolemain$ agent}
    \State{$i.\fieldbias \gets 0$; $i.\fieldrole \gets \roleCR$}
    \Comment{$i$ won't use its bias subsequently}
\EndIf
\If{$i.\fieldrole = \roleMCR$, $j.\fieldrole \neq \rolemain, \roleMCR$, $j.\fieldassigned = \False$ where $\{i,j\} = \{u,v\}$}
    \State{$j.\fieldassigned \gets \True$}
    \Comment{non-$\rolemain$ agents can assign 1 $\rolemain$ agent}
    \State{$i.\fieldrole \gets \rolemain$}
\EndIf
\If{$u.\fieldrole = v.\fieldrole = \roleCR$}
\Comment{Allocate $\approx \frac{1}{4} \roleclock$ agents, $\approx \frac{1}{4} \rolereserve$ agents}
    \State{$u.\fieldrole \gets \roleclock$; $u.\fieldcounter \gets \Theta(\log n)$}
    \State{$v.\fieldrole \gets \rolereserve$}
\EndIf
\If{$u.\fieldrole = v.\fieldrole = \roleclock$}
\Comment{time the phase once we have at least 2 $\roleclock$ agents}
    \State{execute \textbf{\countersubroutine}($u$), execute \textbf{\countersubroutine}($v$)}
\EndIf
\end{algorithmic}
\end{phase}

\begin{algorithm}
\caption*{\textbf{Standard Counter Subroutine}($c$). Agent $c$ with field $\fieldcounter$.}
\label{alg:standard-counter-subroutine}
\begin{algorithmic}[1]
\State{$c.\fieldcounter \gets c.\fieldcounter - 1$}
\If{$c.\fieldcounter = 0$}
    \State{$c.\fieldphase \gets c.\fieldphase + 1$}
    \Comment{move to next phase}
\EndIf
\end{algorithmic}
\end{algorithm}

\phaseDiscreteAveraging\ is a timed phase that averages the biases in $\rolemain$ agents;
with high probability at the end of the phase the $\fieldbias$ fields have three consecutive values, shown in \cref{lem:discrete-averaging}.
\phaseDiscreteAveraging\ expects no agent to remain in role $\roleMCR$ by this point,
signaling an error otherwise.\footnote{
    $\fieldrole = \roleMCR$ is an error because we need all agents not in role $\rolemain$ to have ``donated'' their bias to a $\rolemain$ agent.
    It is okay for $\roleCR$ agents to be undecided about $\roleclock$ versus $\rolereserve$.
    All such agents because $\rolereserve$, which WHP leaves sufficiently many agents in role $\roleclock$ to make the clock interactions sufficiently fast.
}

\begin{phase}[H]
\caption{Discrete Averaging. Agent $u$ interacting with agent $v$.
\\
\textbf{Init} 
if $\fieldrole = \roleMCR$, $\fieldphase \gets 10$ (error, skip to stable backup)
\\
if $\fieldrole = \roleCR$, $\fieldrole \gets \rolereserve$
\\
if $\fieldrole = \roleclock$, $\fieldcounter \gets c_1\ln n \in \{0,\ldots,c_1\ln n\}$ only used in the current phase
}
\label{phase:discrete-averaging}
\begin{algorithmic}[1]
\If{$u.\fieldrole = v.\fieldrole = \rolemain$}
    \State{$u.\fieldbias \gets \floor{\frac{u.\fieldbias + v.\fieldbias}{2}}$; $v.\fieldbias \gets \ceil{\frac{u.\fieldbias + v.\fieldbias}{2}}$}
\EndIf
\For{$c\in\{u,v\}$ with $c.\fieldrole = \roleclock$}
    \State{execute \textbf{\countersubroutine}(c)}
\EndFor
\end{algorithmic}
\end{phase}


\phaseConsensus\ (an untimed phase) checks to see if the entire minority population was eliminated in \phaseDiscreteAveraging\ by checking whether both positive and negative biases still exist. 
It assumes a starting condition where all $\fieldbias \in\{-1,0,+1\}$ (like the initial condition, but allowing some cancelling to have already happened). So the \textbf{Init} checks if any $|\fieldbias|>1$, which can only happen with low probability (since \phaseDiscreteAveraging\ WHP reaches the three consecutive values $\{-1,0,+1\}$), 
and we consider this an error and simply proceed immediately to \phaseStableBackup.
If not, the minority opinion is gone,
and the protocol will stabilize here to the correct output.
Otherwise, we proceed to the next phase;
note this phase is untimed and proceeds immediately upon detection of conflicting opinions. \cref{lem:discrete-averaging} shows that starting from a large initial gap, we will stabilize here with high probability.

\begin{phase}[H]
\caption{Output the Consensus. Agent $u$ interacting with agent $v$.
\\
\textbf{Init} 
if $|\fieldbias|>1$, $\fieldphase \gets 10$ (error, skip to stable backup)
\\
$\fieldbiases \gets \{\fieldopinion\} \subseteq\{-1, 0, +1\}$
}
\label{phase:consensus}
\begin{algorithmic}[1]
\State{$u.\fieldbiases, v.\fieldbiases \gets u.\fieldbiases \cup v.\fieldbiases$}
\Comment{union $\fieldbiases$}
\If{$\{-1,+1\}\subseteq\fieldbiases$}
    \State{$u.\fieldphase, v.\fieldphase \gets u.\fieldphase + 1$}
    \Comment{no consensus, move to next phase}
\ElsIf{$+1\in\fieldbiases$}
    \State{$u.\fieldoutput, v.\fieldoutput \gets \A$}
    \Comment{current consensus is $\A$}
\ElsIf{$-1\in\fieldbiases$}
    \State{$u.\fieldoutput, v.\fieldoutput \gets \B$}
    \Comment{current consensus is $\B$}
\ElsIf{$\fieldbiases = \{0\}$}
    \State{$u.\fieldoutput, v.\fieldoutput \gets \T$}
    \Comment{current consensus is $\T$}
\EndIf
\end{algorithmic}
\end{phase}

\phaseMainAveraging\ is where the bulk of the work gets done.
The biased agents (with non-zero $\fieldopinion$) have an additional field $\fieldlevel \in \{-L,\ldots,0\}$,
initially 0, 
where $L = \lceil\log_2(n)\rceil$, 
corresponding to holding $2^{\fieldlevel}$ units of mass. 
The unbiased $\unbiased$ agents have an additional field $\fieldTlevel = L \in \{0,\ldots,L\}$, and will only participate in split reactions with $-\fieldlevel > \fieldTlevel$. 
The field $\fieldTlevel$ is set by the $\roleclock$ agents, who have a field $\fieldminute = 0 \in \{0,\ldots,kL\}$,
which counts up as \phaseMainAveraging\ proceeds. 
Intuitively, there are $k$ minutes in an hour, so $\fieldhour = \lceil\frac{\fieldminute}{k}\rceil$ ranges from $0$ to $L$.\footnote{
    The drip reaction $C_i,C_i \to C_{i},C_{i+1}$ implemented by line~\ref{line:clock-drip} is asymmetric, but could be made symmetric, i.e., $C_i,C_i \to C_{i+1},C_{i+1}$ without affecting the analysis meaningfully.
}
\cref{fig:clock_minute_hour} shows how nonconsecutive minutes in the $\roleclock$ agents may overlap significantly, but non-consecutive hours in the $\unbiased$ agents have negligible overlap.
Once a $\roleclock$ agent has $\fieldminute$ at its maximum value $kL$, 
they initialize a new field 
$\fieldcounter = \Theta(\log n) \in \{0, \ldots, \Theta(\log n)\}$ 
to wait for the end of the phase
(waiting for any remaining $\unbiased$ agents to reach $\fieldhour = L$ and the distribution to settle).

See the overview in \cref{subsec:overview} for an intuitive description of this phase. \cref{thm:phase-3-tie-result} gives the main result of \phaseMainAveraging\ in the case of an initial tie, where all remaining biased agents are at the minimal $\fieldexponent = -L$. In the other case, \cref{thm:phase-3-majority-result} gives the main result that most of the population settle on the majority output with $\fieldexponent \in \{-l,-(l+1),-(l+2)\}$.

\begin{phase}[H]
\caption{Synchronized Rational Averaging. Agent $u$ interacting with agent $v$.
\\
\textbf{Init}
if $\fieldrole = \rolemain$ and $\fieldopinion \in \{-1,+1\}$, $\fieldlevel \gets 0 \in \{-L, \ldots, -1, 0\}$, and we define $\fieldbias = \fieldopinion \cdot 2^{\fieldlevel}$
\\
if $\fieldrole = \rolemain$ and $\fieldopinion = 0$, $\fieldTlevel \gets 0 \in \{0,\ldots,L\}$
\\
if $\fieldrole = \roleclock$, $\fieldminute \gets 0 \in \{0,\ldots,kL\}$ and $\fieldcounter \gets c_3\ln n \in \{0, \ldots, c_3\ln n\}$
}
\label{phase:main-averaging}
\begin{algorithmic}[1]
\If{$u.\fieldrole = v.\fieldrole = \roleclock$}
     \If{$u.\fieldminute \neq v.\fieldminute$}
        \State{$u.\fieldminute, v.\fieldminute \gets \max(u.\fieldminute, v.\fieldminute)$}
        \Comment{clock epidemic reaction}
        \label{line:clock-epidemic}
    \ElsIf{$u.\fieldminute < kL$} 
        \State{$u.\fieldminute \gets u.\fieldminute + 1$ with probability $p$}
            \Comment{clock drip reaction, $p=1$ in \cref{thm:clock}}
            \label{line:clock-drip}
    \Else
    \label{line:if-both-clocks-finished}
    \Comment{count only when both clocks finished}
        \State{execute \textbf{\countersubroutine}(u), \textbf{\countersubroutine}(v)}
    \EndIf
\EndIf
\If{$m.\fieldrole = \rolemain, m.\fieldopinion = 0$ and $c.\fieldrole = \roleclock$ where $\{m,c\}=\{u,v\}$}
    \State{$m.\fieldTlevel \gets \max\big(m.\fieldTlevel,\floor{\frac{c.\fieldminute}{k}}\big)$}
    \label{line:clock-update-reaction}
    \Comment{clock update reaction}
\ElsIf{$u.\fieldrole = v.\fieldrole = \rolemain$}
    \If{$\{u.\fieldopinion,v.\fieldopinion\} = \{-1,+1\}$ and $u.\fieldlevel = v.\fieldlevel = -h$}
        \State{$u.\fieldopinion, v.\fieldopinion \gets 0$; $u.\fieldhour, v.\fieldhour \gets h$}
        \Comment{cancel reaction}
    \EndIf
    \If{$t.\fieldopinion = 0$, $i.\fieldopinion \in \{-1,+1\}$ and $|t.\fieldTlevel| > |i.\fieldlevel|$, where $\{t,i\}=\{u,v\}$}
        \State{$t.\fieldopinion \gets i.\fieldopinion$}
        \Comment{split reaction}
        \State{$i.\fieldlevel, t.\fieldlevel \gets i.\fieldlevel - 1$}
        \Comment{sets $i.\fieldbias,t.\fieldbias \gets i.\fieldbias / 2$}
    \EndIf
\EndIf
\end{algorithmic}
\end{phase}

In \phaseDetectTie\ (an untimed phase), the population checks if all $\rolemain$ agents have reached the minimum $\fieldlevel = -L$, 
which only happens in the case of a tie. 
If so, the population will stabilize to the tie output. 
Otherwise, any $\rolemain$ agent with $\fieldlevel$ above $L$ can trigger the move to the next phase.

\begin{phase}[H]
\caption{Output Tie. Agent $u$ interacting with agent $v$.
\\
\textbf{Init} $\fieldoutput \gets \T$ }
\label{phase:detecttie}
\begin{algorithmic}[1]
\If{$|m.\fieldbias| > 2^{-L}$ where $m\in\{u,v\}$}
\Comment{stable if all $\fieldbias \in \left\{ -\frac{1}{2^L}, 0, +\frac{1}{2^L} \right\}$}
    \State{$u.\fieldphase, v.\fieldphase \gets u.\fieldphase + 1$}
    \Comment{end of this phase}
\EndIf
\end{algorithmic}
\end{phase}

If there was not a tie detected in \phaseDetectTie, then most agents should have the majority opinion, with $\fieldlevel \in \{-l, -(l+1), -(l+2)\}$ in a small range. 
The next goal is to bring all agents with $\fieldlevel > -l$ down to $\fieldlevel \leq -l$. This will be accomplished by the $\rolereserve$ agents, whose goal is to let exponents above $l$ do split reactions.

The $\rolereserve$ agents do this across two consecutive phases.
In \phaseReserveSample, the $\rolereserve$ agents become active, and will set $\fieldsample = \bot \in \{\bot,0,\ldots,L\}$ to the exponent of the first biased agent they meet, which is likely in $ \{-l, -(l+1), -(l+2)\}$. 
The $\roleclock$ agents now only hold a field $\fieldcounter = \Theta(\log n)\in\{0,\ldots,\Theta(\log n)\}$ to act as a simple timer for how long to wait until moving to the next phase.
This allows the $\rolereserve$ agents to adopt a distribution of $\fieldlevel$ values approximately equal to that of the $\rolemain$ agents.
This behavior is proven in \cref{lem:reserve-sample} and \cref{lem:reserve-split}.

\begin{phase}[H]
\caption{Reserves Sample exponent. Agent $u$ interacting with agent $v$.
\\
\textbf{Init}
if $\fieldrole = \rolereserve$, $\fieldsample \gets \bot \in \{\bot,-L,\ldots,-1,0\}$
\\
if $\fieldrole = \roleclock$, $\fieldcounter \gets c_5\ln n\in\{0,\ldots,c_5\ln n\}$
}
\label{phase:reserve-sample}
\begin{algorithmic}[1]
\If{$r.\fieldrole = \rolereserve$ and $m.\fieldrole = \rolemain, m.\fieldopinion \in \{-1,+1\}$ where $\{r,m\}=\{u,v\}$}
    \If{$r.\fieldsample = \bot$}
    \Comment{sample exponent of biased agent $m$}
        \State{$r.\fieldsample \gets m.\fieldlevel$}
    \EndIf
\EndIf
\For{$c\in\{u,v\}$ with $c.\fieldrole = \roleclock$}
    \State{execute \textbf{\countersubroutine}(c)}
\EndFor
\end{algorithmic}
\end{phase}

In \phaseReserveSplit,
the $\rolereserve$ agents can help facilitate more split reactions, 
with any agent at an exponent above their sampled exponent.
Because they have approximately the same distribution across all exponents as $\rolemain$ agents,
particular exponents $-l, -(l+1), -(l+2)$,
this allows them to bring all agents above $\fieldlevel = -l$ down to $-l$ or below.
Again, the $\roleclock$ agents keep a $\fieldcounter \in \{0,\ldots,c_6 \ln n\}$. \cref{lem:reserve-split} proves that \phaseReserveSplit\ works as intended, bringing all agents down to $\fieldlevel \leq -l$ with high probability.

\begin{phase}[H]
\caption{Reserve Splits. Agent $u$ interacting with agent $v$.
\\
\textbf{Init}
if $\fieldrole = \roleclock$, $\fieldcounter \gets c_6 \ln n \in \{0,\ldots,c_6 \ln n\}$
}
\label{phase:reserve-split}
\begin{algorithmic}[1]
\If{$r.\fieldrole = \rolereserve$ and $m.\fieldrole = \rolemain, m.\fieldopinion \in \{-1,+1\}$ where $\{r,m\}=\{u,v\}$}
    \If{$r.\fieldsample \neq \bot$ and $r.\fieldsample < m.\fieldlevel$}
        \State{$r.\fieldrole \gets \rolemain$; $r.\fieldopinion \gets m.\fieldopinion$}
        \Comment{split reaction}
        \label{line:reserve-split}
        \State{$r.\fieldlevel, m.\fieldlevel \gets m.\fieldlevel - 1$}
        \Comment{sets $r.\fieldbias,m.\fieldbias \gets m.\fieldbias / 2$}
    \EndIf
\EndIf
\For{$c\in\{u,v\}$ with $c.\fieldrole = \roleclock$}
    \State{execute \textbf{\countersubroutine}(c)}
\EndFor
\end{algorithmic}
\end{phase}

Now that all agents are $\fieldlevel \leq -l$, 
the goal of \phaseHighMinorityElimination\ is to eliminate any minority agents with exponents $-l, -(l+1), -(l+2)$.
This is done by letting the biased agents do generalized cancel reactions that allow their difference in exponents to be up to 2, 
while still preserving the mass invariant.
Again, the $\roleclock$ agents keep a $\fieldcounter \in \{0,\ldots, c_7 \ln n \}$.
\cref{lem:high-minority-elimination} shows that by the end of this phase, any remaining minority agents must have $\fieldexponent < -(l+2)$, with high probability.


\begin{phase}[H]
\caption{High-exponent Minority Elimination. Agent $u$ interacting with agent $v$
\\
\textbf{Init}
if $\fieldrole = \roleclock$, $\fieldcounter \gets c_7 \ln n \in\{0,\ldots, c_7 \ln n \}$
}
\label{phase:highminorityelimination}
\begin{algorithmic}[1]
\If{$u.\fieldrole = v.\fieldrole = \rolemain$ and $\{u.\fieldopinion,v.\fieldopinion\} = \{-1,+1\}$}
    \If{$u.\fieldlevel = v.\fieldlevel$}
        \State{$u.\fieldopinion, v.\fieldopinion \gets 0$}
        \Comment{cancel reaction}
    \ElsIf{$i.\fieldlevel = j.\fieldlevel + 1$ where $\{i,j\} = \{u,v\}$}
        \State{$i.\fieldlevel \gets i.\fieldlevel - 1$}
        \Comment{gap-1 cancel reaction}
        \State{$j.\fieldopinion \gets 0$}
        \Comment{example $\fieldbias$ update:  $+\frac{1}{4}, -\frac{1}{8} \to +\frac{1}{8}, 0$} 
    \ElsIf{$i.\fieldlevel = j.\fieldlevel + 2$ where $\{i,j\} = \{u,v\}$}
        \State{$j.\fieldopinion \gets i.\fieldopinion$}
        \Comment{gap-2 cancel reaction}
        \State{$i.\fieldlevel \gets i.\fieldlevel - 1$}
        \Comment{example $\fieldbias$ update: $+\frac{1}{4}, -\frac{1}{16} \to +\frac{1}{8}, +\frac{1}{16}$} 
        \State{$j.\fieldlevel \gets i.\fieldlevel - 2$}
    \EndIf
\EndIf
\For{$c\in\{u,v\}$ with $c.\fieldrole = \roleclock$}
    \State{execute \textbf{\countersubroutine}(c)}
\EndFor
\end{algorithmic}
\end{phase}

Now that all minority agents occupy exponents below $-(l+2)$,
yet a large number of majority agents remain at exponents $-l, -(l+1), -(l+2)$,
in \phaseLowMinorityElimination,
the algorithm eliminates the last remaining minority opinions at any exponent.
It allows opposite-opinion agents of \emph{any} two exponents to react and eliminate the smaller-exponent opinion.
The larger exponent $\frac{1}{2^i}$ ``absorbs'' the smaller $-\frac{1}{2^j}$, setting the smaller to mass $0$;
the larger now represents mass $\frac{1}{2^i} - \frac{1}{2^j}$,
which it lacks the memory to track exactly,
so it cannot absorb any further agents 
(though it can itself be absorbed by $-\frac{1}{2^m}$ for $m > i$).
\cref{lem:low-minority-elimination} shows that \phaseLowMinorityElimination\ eliminates any remaining minority agents with high probability.
A key property is that it cannot violate correctness since, although the biases held by some agents becomes unknown, 
the allowed transitions maintain the sign of the bias.
Thus the only source of error is failing to eliminate all minority agents,
detected in the next phase.

\begin{phase}[H]
\caption{Low-exponent Minority Elimination. Agent $u$ interacting with agent $v$
\\
\textbf{Init}
if $\fieldrole = \roleclock$, $\fieldcounter \gets c_8 \ln n \in\{0,\ldots, c_8 \ln n \}$
}
\label{phase:lowminorityelimination}
\begin{algorithmic}[1]
\If{$u.\fieldrole = v.\fieldrole = \rolemain$ and $\{u.\fieldbias,v.\fieldbias\} = \{-1,+1\}$}
    \If{$i.\fieldlevel > j.\fieldlevel$ and $i.\fieldfull = \False$ where $\{i,j\}=\{u,v\}$}
        \State{$i.\fieldfull \gets \True$}
        \Comment{consumption reaction}
        \State{$j.\fieldopinion \gets 0$}
    \EndIf
\EndIf
\For{$c\in\{u,v\}$ with $c.\fieldrole = \roleclock$}
    \State{execute \textbf{\countersubroutine}(c)}
\EndFor
\end{algorithmic}
\end{phase}

\phaseConsensusTwo\ (an untimed phase) acts exactly as \phaseConsensus, to check that agents have reached consensus.
Note that the initial check for $|\fieldbias| > 1$ is not required, since it is guaranteed to pass if we reach this point:
it passed in \phaseConsensus,
and the population-wide maximum $|\fieldbias|$ could only have decreased in subsequent phases.

\begin{phase}[H]
\caption{Output the Consensus. Exact repeat of \phaseConsensus.}
\label{phase:consensus2}
\end{phase}


In \phaseStableBackup, 
the agents give up on the fast algorithm,
having determined in \phaseConsensusTwo\ that it failed to reach consensus, or detected an earlier error with $\fieldrole$ or $\fieldbias$ assignment.
Instead they rely instead on a slow stable backup protocol.
This is a 6-state protocol that stably decides between the three cases of majority $\A$, $\B$, and tie $\T$. 
They only use the fields $\fieldoutput = \fieldinput \in \{\A,\B,\T\}$, and the initial field $\fieldactive = \True \in \{\True,\False\}$. \cref{lem:stable-backup} proves this 6-state protocol stably computes majority in $O(n\log n)$ time.



\begin{phase}[H]
\caption{Stable Backup. Agent $u$ interacting with agent $v$.
Similar to 6-state algorithm from \cite{BlackNinjasInTheDark},
slight modification of 4-state algorithm from~\cite{mertzios2014determining, draief2012convergence}.
\\
\textbf{Init}
$\fieldoutput \gets \fieldinput$, $\fieldactive \gets \True$}
\label{phase:backup}
\begin{algorithmic}[1]
\If{$u.\fieldactive = v.\fieldactive = \True$}
    \If{$\{u.\fieldoutput,v.\fieldoutput\} = \{\A,\B\}$}
        \State{$u.\fieldoutput, v.\fieldoutput \gets \T$}
        \Comment{cancel reaction}
        \label{line:stable-cancel-reaction}
    \ElsIf{$i.\fieldoutput\in\{\A,\B\}$ and $t.\fieldoutput = \T$ where $\{i,t\}=\{u,v\}$}
        \State{$t.\fieldoutput \gets i.\fieldoutput$, $t.\fieldactive \gets \False$}
        \Comment{biased converts unbiased}
        \label{line:stable-convert-unbiased}
    \EndIf
\EndIf
\If{$a.\fieldactive = \True$ and $p.\fieldactive = \False$ where $\{a,p\}=\{u,v\}$}
\label{alg:slow-stable-majority}
    \State{$p.\fieldoutput \gets a.\fieldoutput$}
    \Comment{active converts passive}
    \label{line:stable-convert-passive}
\EndIf
\end{algorithmic}
\end{phase}

\section{Useful time bounds}
\label{sec:timing-lemmas}



This section introduces various probability bounds which will be used repeatedly in later analysis.

We will use the following standard multiplicative Chernoff bound:

\begin{theorem}
\label{thm:chernoff-bound-multiplicative}
Let $X = X_1 + \ldots + X_k$ be the sum of independent $\{0,1\}$-valued random variables, with $\mu = \E[X]$. Then for any $\delta > 0$, $\Pr[X \geq (1 + \delta)\mu] \leq \exp(-\frac{\delta^2 \mu}{2 + \delta})$.
\end{theorem}

We use the standard Azuma inequality for supermartingales:

\begin{theorem}
\label{thm:azuma}
Let $X_0,X_1,X_2,\ldots$ be a supermartingale such that,
for all $i \in \mathbb{N}$,
$|X_{i+1} - X_i| \leq c_i$.
Then for all $n \in \mathbb{N}$ and $\epsilon > 0$,
$
    \Pr\qty[(X_n - X_0) - \E[X_n - X_0] \geq \epsilon]
    \leq
    \exp(- \frac{2 \epsilon^2}{\sum_{i=0}^n c_i^2}).
$
\end{theorem}
In application we will consider potential functions $\phi$ that decay exponentially, with $\E[\phi_{j+1}] \leq (1-\epsilon)\phi_j$. In this case, we will take the logarithm $\Phi = \ln(\phi)$, which we will show is a supermartingale upon which we can apply Azuma's Inequality to conclude that $\phi$ achieves a requisite amount of exponential decay.

We will also use the following two concentration bounds on heterogeneous sums of geometric random variables, due to Janson~\cite[Theorems 2.1, 3.1]{janson2018tail}:

\begin{theorem}
\label{thm:janson}
    Let $X$ be sum of $k$ independent geometric random variables with success probabilities $p_1,...,p_k$,
    let $\mu = \E[X] = \sum_{i=1}^k \frac{1}{p_i},$
    and let $p^* = \min\limits_{1 \leq i \leq k} p_i$.
    For all $\lambda \geq 1,$
    $\Pr[X \geq \lambda \mu] 
    \leq e^{- p^* \mu (\lambda - 1 - \ln \lambda)}.$
    For all $\lambda \leq 1,$
    $\Pr[X \leq \lambda \mu] \leq e^{- p^* \mu (\lambda - 1 - \ln \lambda)}.
    $
\end{theorem}

The expression $\lambda - 1 - \ln \lambda$ is hard to work with if we are not fixing exact constant values for $\lambda$. The following Corollary gives asymptotic approximation to the error bound:

\begin{corollary}
\label{cor:janson-wvhp}
    Let $X$ be sum of $k$ independent geometric random variables with success probabilities $p_1,...,p_k$,
    let $\mu = \E[X] = \sum_{i=1}^k \frac{1}{p_i},$
    and let $p^* = \min\limits_{1 \leq i \leq k} p_i$.
    For any $0 < \epsilon < 1$, we have
    $(1-\epsilon)\mu \leq X \leq (1+\epsilon)\mu$ with probability $1 - \exp(-\Theta(\epsilon^2 p^* \mu))$.
\end{corollary}

\begin{proof}
Setting $\lambda = 1 + a$ in \cref{thm:janson}, where $a = \epsilon$ in the upper bound and $a = -\epsilon$ in the lower bound, 
by Taylor series approximation of $\ln(1+a)$,  
$\lambda - 1 - \ln(\lambda) = a - \ln(1+a) \geq a - (a - \frac{a^2}{2} + \frac{a^3}{3}) = a^2(\frac{1}{2}-\frac{a}{3}) = \Theta(\epsilon^2)$. 
The stated inequality then follows from \cref{thm:janson}.
\end{proof}

The following lemmas give applications of \cref{thm:janson} and \cref{cor:janson-wvhp} to common processes that we will repeatedly analyze. When the fractions we consider are distinct and independent of $n$, \cref{cor:janson-wvhp} applies to give tight time bounds with very high probability. We also consider cases where we run a process to completion and bring a count to $0$. Here we must use \cref{thm:janson} and only get a high probability bound with times at most a constant factor above the mean.

First we consider the epidemic process, $i,s\rightarrow i,i$, moving from a constant fraction infected to another constant fraction infected.



\begin{lemma}
\label{lem:epidemic}
Let $0 < a < b < 1$.
Consider the epidemic process starting from a count of $a\cdot n$ infected agents.
The expected parallel time $t$ until there is a count $b\cdot n$ of infected agents is
\[
\E[t] 
= \frac{\ln(b-\frac{1}{n}) - \ln(1-b+\frac{1}{n}) - \ln(a) + \ln(1-a)}{2} 
\sim \frac{\ln(b) - \ln(1-b) - \ln(a) + \ln(1-a)}{2}.
\]
Let $c = \min(a,1-b+\frac{1}{n})$ and $0 < \epsilon < 1$.
Then $(1-\epsilon)\E[t] < t < (1+\epsilon)\E[t]$ 
with probability at least 
$1 - \exp[-\Theta(\epsilon^2 \E[t] n c)].$
\end{lemma}

\begin{proof}
When there are $i$ infected agents, the probability the next reaction infects another agent and increases this count is $\frac{i(n-i)}{\binom{n}{2}}$. The number of interactions $T$ for the count to increase from $a\cdot n$ to $b\cdot n$ is a sum of geometric random variables, with expected value
\begin{align*}
\E[T] &= \sum_{i = a \cdot n}^{b\cdot n - 1}
\frac{\binom{n}{2}}{i(n-i)}
\sim
\frac{1}{2}
\sum_{i = a \cdot n}^{b\cdot n - 1}
\frac{1}{\frac{i}{n}(1-\frac{i}{n})}
\sim \frac{n}{2}\int_{x = a}^{b-\frac{1}{n}}\frac{dx}{x(1-x)}
\\
&=  \frac{n}{2}\qty[\ln(x)-\ln(1-x)]_{a}^{b-\frac{1}{n}}
= n \cdot \frac{\ln(b-\frac{1}{n}) - \ln(1-b+\frac{1}{n}) - \ln(a) + \ln(1-a)}{2},
\end{align*}
giving the stated value of $\E[t]$ after converting from interactions to parallel time.
The minimum probability $p^* = \Theta(\min(a, 1-b+\frac{1}{n}))$, so using \cref{cor:janson-wvhp} we have Then $(1-\epsilon)\E[t] < t < (1+\epsilon)\E[t]$ 
with probability at least 
$1 - \exp[-\Theta(\epsilon^2 \E[T] p^*)].$
\end{proof}

Note that if $a, (1-b), \epsilon$ are all constants independent of $n$, then the bound above is with very high probability. If we wanted to consider the complete epidemic process (which starts with $a=1/n$ and ends with $(1-b) = 1/n$), we would have $\E[t] = \Theta(\log n)$ and minimum probability $p^* = \Theta(\frac{1}{n})$. Then the probability bound would become $1 - \exp[-\Theta(\log n)] = 1 - n^{-\Theta(1)}$, which is high probability, but requires carefully considering the constants to get the exact polynomial bound. In our proofs, we only use the very high probability case. For tight  bounds on a full epidemic with precise constants, see \cite{mocquard2016analysis, burman2021self}.

Next we consider the cancel reactions $a,b \rightarrow 0,0$, which are key to the majority protocol.

\begin{lemma}
\label{lem:cancel-reactions}
Consider two disjoint subpopulations $A$ and $B$ of initial sizes $|A|=a\cdot n$ and $|B| = b\cdot n$, where $0<b<a<1$. An interaction between an agent in $A$ and an agent in $B$ is a \emph{cancel reaction} which removes both agents from their subpopulations. Thus after $i$ cancel reactions we have $|A|=a\cdot n - i$ and $|B|=b\cdot n - i$.

The expected parallel time $t$ until $d\cdot n$ cancel reactions occur, where $d < b$ is
\[
\E[t] = \frac{\ln(b) - \ln(a) - \ln(b-d+\frac{1}{n}) + \ln(a-d+\frac{1}{n})}{2(a-b)}.
\]

Let $c = (a - d + \frac{1}{n})(b- d + \frac{1}{n})$ and $0 < \epsilon < 1$.
Then $(1-\epsilon)\E[t] < t < (1+\epsilon)\E[t]$ 
with probability at least 
$1 - \exp[-\Theta(\epsilon^2 \E[t] n c)]$.

The parallel time $t$ until all $b\cdot n$ cancel reactions occur has $\E[t] \sim \frac{\ln n}{2(a-b)}$ and satisfies $t \leq \frac{5\ln n}{2(a-b)}$ with high probability $1-O(1/n^2)$.
\end{lemma}

Again, the first bound is with very high probability if all constants $a,b,d$ are independent of $n$.

\begin{proof}
After $i$ cancel reactions, the probability of the next cancel reaction is $p \sim \frac{2|A|\cdot|B|}{n^2} = 2(a - \frac{i}{n})(b - \frac{i}{n})$, and the number of interactions until this cancel reaction is a geometric random variable with mean $p$. The number of interactions $T$ for $d\cdot n$ cancel reactions to occur is a sum of geometrics with mean
\begin{align*}
\E[T] &= \sum_{i = 0}^{d\cdot n - 1}\frac{1}{2(a - \frac{i}{n})(b - \frac{i}{n})}
\sim n\int_{x = 0}^{d-\frac{1}{n}}\frac{dx}{2(a-x)(b-x)}
= \frac{n}{2(a-b)}\int_{x = 0}^{d-\frac{1}{n}}\qty(\frac{1}{b-x}-\frac{1}{a-x})dx
\\
& = \frac{n}{2(a-b)}\bigg[-\ln(b-x)+\ln(a-x)\bigg]_0^{d-\frac{1}{n}}
= n\cdot\frac{\ln(b)-\ln(a)-\ln(b-d+\frac{1}{n})+\ln(a - d + \frac{1}{n})}{2(a-b)}.
\end{align*}
Translating to parallel time and using \cref{cor:janson-wvhp} gives the first result, where minimum probability $p^* = \Theta(c)$.

In the second case where $d = b$ and we are waiting for all cancel reactions to occur, then
\[
\E[t] \sim \frac{\ln(b) - \ln(a) + \ln(n) + \ln(a-b)}{2(a-b)} = \frac{\ln n}{2(a-b)},
\]
and $\mu = \E[T] \sim \frac{n\ln n}{2(a-b)}$.
Now the minimum geometric probability is when $i = bn - 1$ and $p^* \sim \frac{2(a-b)}{n}$. Choosing $\lambda = 5$ so that $\lambda - 1 - \ln\lambda > 2$, by \cref{thm:janson} we have
\[
\Pr\qty[t \geq \frac{5\ln n}{2(a-b)}]
\leq
\Pr\qty[T \geq \lambda\mu]
\leq
e^{-p^*\mu(\lambda - 1 - \ln\lambda)}
= \exp(-\frac{2(a-b)}{n}\cdot\frac{n\ln n}{2(a-b)}\cdot2)
= n^{-2}.
\]
Thus $t \leq \frac{5\ln n}{2a}$ with high probability.

Note we get the same result for this second case also by assuming a fixed minimal fraction $a-b$ in $|A|$ is always there to cancel the agents from $|B|$, and using the next \cref{lem:one-sided-cancel}.
\end{proof}

Now we consider a ``one-sided'' cancel process, where reactions $a,b \rightarrow a,0$ change only the $b$ agents into a different state. This process happens for example when $\roleclock$ agents change the $\fieldhour$ of $\unbiased$ agents.

\begin{lemma}
\label{lem:one-sided-cancel}
Let $0 < a, b_1, b_2, \epsilon < 1$ be constants with $b_1 > b_2$. Consider a subpopulation $A$ maintaing its size above $a\cdot n$, and $B$ initially of size $b_1\cdot n$. Any interaction between an agent in $A$ and in $B$ is \emph{meaningful}, and forces the agent in $B$ to leave its subpopulation.

The expected parallel time $t$ until the subpopulation $B$ reaches size $b_2\cdot n$ is 
\[
\E[t] = \frac{\ln(b_1) - \ln(b_2)}{2a},
\]
and satisfies $(1-\epsilon)\E[t] < t < (1+\epsilon)\E[t]$ with probability at least 
$1 - \exp[-\Theta(\epsilon^2 \E[t] n a b_2)]$.

The parallel time $t$ until the subpopulation $B$ reaches size $0$ has $\E[t] \sim \frac{\ln n}{2a}$ and satisfies $t \leq \frac{5\ln n}{2a}$ with high probability $1- O(1/n^2)$.
\end{lemma}

Again, the first bound is with very high probability if constants $a,b_2$ are independent of $n$.

\begin{proof}
When $|B| = i$, then the probability of a meaningful interaction $p \sim \frac{2ia}{n}$. Then the number of interactions before the next meaningful interaction is a geometric with probability $p$, and the total number $T$ of interaction is a sum of geometrics with mean
\[
\E[T] = \sum_{i = b_2n + 1}^{b_1n}\frac{n}{2ia}
= \frac{n}{2a}\qty(\sum_{i= 1}^{b_1n} \frac{1}{i} - \sum_{i=1}^{b_2n+1}\frac{1}{i})
\sim \frac{n}{2a}(\ln(b_1n) - \ln(b_2n+1))
\sim n\frac{\ln(b_1) - \ln(b_2+\frac{1}{n})}{2a}.
\]

Translating to parallel time and using \cref{cor:janson-wvhp} gives the first result, where minimum probability $p^* = \Theta(a \cdot b_2)$.

In the case where $b_2 = 0$, we have $\mu = \E[T] \sim \frac{n\ln n}{2a}$ and $\E[t] \sim \frac{\ln(n)}{2a}$. Now the minimum geometric probability $p^* = \frac{2a}{n}$. Choosing $\lambda = 5$ so that $\lambda - 1 - \ln\lambda > 2$, by \cref{thm:janson} we have
\[
\Pr\qty[t \geq \frac{5\ln n}{2a}]
\leq
\Pr\qty[T \geq \lambda\mu]
\leq
e^{-p^*\mu(\lambda - 1 - \ln\lambda)}
= \exp(-\frac{2a}{n}\cdot\frac{n\ln n}{2a}\cdot2)
= n^{-2}.
\]
Thus $t \leq \frac{5\ln n}{2a}$ with high probability.
\end{proof}
\section{Analysis of initial phases}
\label{sec:analysis-early-phases}

Define $\mains$, $\clocks$, and $\reserves$ to be the sub-populations of agents with roles $\rolemain$, $\roleclock$, $\rolereserve$ when they first set $\fieldphase = 1$. Let $|\mains| = m\cdot n, |\clocks| = c \cdot n, |\reserves| = r \cdot n$, where $m + c + r = 1$.

The population splitting of \phaseInitialize\ will set the fractions $m \approx \frac{1}{2}$, $c \approx \frac{1}{4}$, and $r \approx \frac{1}{4}$. The rules also ensure deterministic bounds on these fractions once all $\roleMCR$ agents have been assigned. The probability 1 guarantees on subpoplation sizes using this method will be key for a later uniform adaptation of the protocol. 
We first prove the behavior of the rules used for this initial top level split in \phaseInitialize\ between the roles $\rolemain$ and $\roleCR$.

\begin{lemma}
\label{lem:split-deterministic}
    Consider the reactions 
    \begin{align*}
        U,U   &\to S_f,M_f
        \\
        S_f,U &\to S_t,M_f
        \\
        M_f,U &\to M_t,S_f
    \end{align*}
    starting with $n$ $U$ agents.
    Let $u = \# U$, $s = \#S_f + \#S_t$, and $m = \#M_f + \#M_t$.
    This converges to $u = 0$ in expected time at most $2.5\ln n$ and in $12.5\ln n$ time with high probability $1 - O(1/n^2)$.
    Once $u = 0$, 
    $\frac{n}{3} \leq s,m \leq \frac{2n}{3}$ with probability 1,
    and for any $\epsilon > 0$,
    $\frac{n}{2}(1-\epsilon) \leq s,m \leq \frac{n}{2}(1+\epsilon)$ 
    with very high probability.
\end{lemma}

\begin{proof}
Let $s_f = \#S_f, s_t = \#S_t, m_f = \#M_f, m_t = \#M_t.$

First we consider the interactions that happen until $u = 2n / 3$. Note that while $u \geq 2n / 3$, the probability of the first reaction is $\sim\qty(\frac{u}{n})^2 \geq \frac{4}{9}$ and the probability of the other two reactions is $\sim 2\qty(\frac{u}{n})\qty(\frac{m_f + s_f}{n}) \leq 2\cdot \frac{2}{3} \cdot \frac{1}{3} = \frac{4}{9}$. Thus until $u = 2n/3$, each non-null interaction is the top reaction with at least probability $1/2$. Then by standard Chernoff Bounds,
with very high probability,
at least a fraction $\frac{1}{2} - \epsilon$ of these reactions are the top reaction, which create a count $s_f + m_f \geq \frac{2}{9}(1 - \epsilon)n > \frac{n}{5}$.

Now notice that this count $s_f + m_f$ can never decrease, so we have $s_f + m_f > \frac{n}{5}$ for all future interactions. Now we can use the rate of the second and third reactions to bound the completion time. The probability of decreasing $u$ is at least $2\qty(\frac{u}{n})(\frac{1}{5})$, so the number of interactions it takes to decrement $u$ is stochastically dominated by a geometric random variable with probability $p = \frac{2u}{5n}$. Then the number of interactions for $u$ to decrease from $\frac{2n}{3}$ down to $0$ is dominated by a sum $T$ of geometric random variables with mean
\[
\E[T] = \sum_{u = 1}^{2n / 3}\frac{5n}{2u} = \frac{5n}{2}\sum_{u = 1}^{2n / 3}\frac{1}{u} \sim \frac{5n}{2}\ln(2n/3) \sim \frac{5}{2}n \ln n.
\]

Now we will apply the upper bound of \cref{thm:janson}, using $\mu = \frac{5}{2}n\ln n$, $p^* = \frac{2}{5n}$ (when $u = 1$) and $\lambda = 5$, where $\lambda - 1 - \ln(\lambda) > 2$, so we get
\[
\Pr[T \geq \lambda \mu] \leq \exp(-p^* \mu (\lambda - 1 - \ln \lambda))
\leq \exp(-\frac{2}{5n} \cdot \frac{5}{2}n\ln n \cdot 2) = n^{-2}.
\]
Thus with high probability $1 - O(1/n^2)$, this process converges in at most $\frac{\lambda\mu}{n} = 12.5\ln n$ parallel time.

Observe that the reactions all preserve the following invariant:
$s_f + 2s_t = m_f + 2m_t$.
The probability-1 count bounds follow by observing that when $u = 0$ (i.e., we have converged) we have $s_f+s_t+m_f+m_t=n$.
Maximizing $s=s_f+s_t$ and minimizing $m=m_f+m+t$ is achieved by setting $s_f=2n/3, s_t=0, m_f=0, m_t=n/3$,
and symmetrically for maximizing $m$.

Assume we have $s > m$ at some point in the execution, so
$s_f + s_t > m_f + m_t$.
Subtracting the invariant equation gives
$-s_t > -m_t$,
which implies
$s_t < m_t$.
Together with the first inequality this implies that
$s_f > m_f$. 
Thus the rate of second reaction is higher than the rate of the third reaction,
so it is more likely that the next reaction changing the value $s-m$ decreases it.

By symmetry the opposite happens when $m > s$, so we conclude that the absolute value $|m - s|$ is more likely to decrease than increase. Thus we can stochastically dominate $|m - s|$ by a sum of independent coin flips.
The high probability count bounds follow by standard Chernoff bounds.
\end{proof}

Now we can prove bounds on the sizes $|\mains| = mn, |\clocks| = cn, |\reserves| = rn$ of $\rolemain$, $\roleclock$, and $\rolereserve$ agents from the population splitting of \phaseInitialize.

\begin{lemma}
\label{lem:phase-initialize}
    For any $\epsilon > 0$,
    with high probability $1 - O(1/n^2)$, 
    by the end of \phaseInitialize, 
    $|\roleMCR| = 0$, 
    $\frac{n}{2}(1-\epsilon) \leq |\mains| \leq \frac{n}{2}(1+\epsilon)$
    and
    $|\clocks|,|\reserves| \geq \frac{n}{4}(1-\epsilon)$.
    If $|\roleMCR| = 0$ 
    and \phaseDiscreteAveraging\ initializes without error, then with probability $1$, 
    $\frac{1}{3}n \leq |\mains|    \leq \frac{2}{3}n$, 
    $\frac{1}{6}n \leq |\reserves| \leq \frac{2}{3}n$, and 
    $2 \leq |\clocks|   \leq \frac{1}{3}n$.
\end{lemma}

\begin{proof}
The top level of splitting of $\roleMCR$ into $\roleCR$ and $\rolemain$ is equivalent to the reactions of \cref{lem:split-deterministic},
with $U = \roleMCR$,
$M = \rolemain$,
$S = \roleCR$,
and the $f$ and $t$ subscripts representing the Boolean value of the field $\fieldassigned$.
\cref{lem:split-deterministic} gives the stated bounds on $\rolemain$,
if there were no further splitting of $\roleCR$.

\cref{lem:split-deterministic} gives that with high probability,
all $\roleMCR$ are converted to $\roleCR$ and $\rolemain$ in $12.5\ln n$ time;
we begin the analysis at that point,
letting $s$ be the number of $\roleCR$ agents produced,
noting $n/3 \leq s \leq 2n/3$ with probability 1,
and 
for any $\epsilon'$
$\frac{n}{2}(1-\epsilon') \leq s \leq \frac{n}{2}(1+\epsilon')$ with high probability.

The splitting of $\roleCR$ into $\roleclock$ and $\rolereserve$ follows a simpler process that we analyze here.
Let $r = |\reserves|$ and $c = |\clocks|$ at the end of \phaseInitialize.
To see the high probability bounds on $r$ and $c$,
we model the splitting of $\roleCR$ by the reaction $U,U \to R,C$ during \phaseInitialize\ 
and $U \to R$ at the end of \phaseInitialize,
since all un-split $\roleCR$ agents become $\rolereserve$ agents upon leaving \phaseInitialize.

The reaction $U,U \to R,C$ reduces the count of $U$ from its initial value $s$ ($=n/2 \pm \epsilon' n / 2$  WHP)
to $\epsilon' s$,
with the number of interactions between each reaction when $\# U = l$ governed by a geometric random variable with success probability $O(l^2 / n^2)$.
Applying \cref{cor:janson-wvhp} with 
$k = s - \epsilon' s$,
$p_i = O((i+\epsilon' s)^2 / n^2)$ for $i \in \{1,\ldots,k\}$,
the reaction $U,U \to R,C$ takes $O(1)$ time to reduce the count of $U$ from its initial value $m$ ($=n/2 \pm \epsilon' n / 2$  WHP)
to $\epsilon' m$ with very high probability.
This implies that after $O(1)$ time, 
$r$ and $c$ are both at least 
\begin{align*}
    s/2 - \epsilon' s
    = 
    s (1/2 - \epsilon')
    &\geq 
    (n/2 - \epsilon' n / 2) (1/2 - \epsilon')
    \\&= 
    n/4 - \epsilon' n / 2 -  \epsilon' n / 4 + (\epsilon')^2 n /2
    \\&>
    n/4 - \epsilon' n 
    =
    n/4 (1 - 4 \epsilon')
\end{align*}
with very high probability.
Choosing $\epsilon = \epsilon' / 4$ gives the high probability bounds on $r$ and $c$. Thus we require $12.5\ln n + O(1) \leq 13\ln n$ time, and for appropriate choice of $\fieldcounter$ constant $c_0$, this happens before the first $\roleclock$ agent advances to the next phase with high probability.

We now argue the probability-1 bounds.
The bound $r \leq 2n/3$ follows from $|\mains| \geq n/3$.
The bound $r \geq n/6$ follows from $|\mains| \leq 2n/3$, so $\roleCR \geq n/3$ if no $U,U \to R,C$ splits happen, 
and the fact that at least half of $\roleCR$ get converted to $\rolereserve$:
exactly half by $U,U \to R,C$ and the rest by $U \to R$.

Although the reactions 
$U,U \to R,C$ and $U \to R$ can produce only a single $C$,
there must be at least two $\roleclock$ agents for \countersubroutine\ to count at all and end \phaseInitialize, so if \phaseDiscreteAveraging\ initializes, $c \geq 2$.
The bound $c \leq n/3$ follows from the fact that $c$ is maximized when $|\mains| = n/3$ and no $U \to R$ reactions happen, i.e., all $2n/3$ $\roleCR$ agents are converted via $U,U \to R,C$,
leading to $c = n/3$.
\end{proof}

We now reason about \phaseDiscreteAveraging, which has different behavior based on the initial gap $g$.

\begin{lemma}
\label{lem:discrete-averaging}
If the initial gap $|g|\geq 0.025|\mains|$, then we stabilize to the correct output in \phaseConsensus. If $|g| < 0.025|\mains|$, then at the end of \phaseDiscreteAveraging, all agents have $\fieldbias \in \{-1,0,+1\}$, and the total count of biased agents is at most $0.03|\mains|$. Both happen with high probability $1-O(1/n^2)$.
\end{lemma}

\begin{proof}
In the case where all agents enter \phaseDiscreteAveraging, none are still in $\roleMCR$, so every non-$\rolemain$ agent has given their $\fieldbias$ to a $\rolemain$ agent via \cref{line:phase-0-assign-main} or \cref{line:phase-0-main-takes-bias} of \phaseInitialize. Thus the initial gap $g = \sum_{m.\fieldrole = \rolemain} m.\fieldbias$.

\newcommand{\avgbias}{\lfloor \frac{g}{|\mains|} \rceil}
Let $\mu = \avgbias$ be the average $\fieldbias$ among all $\rolemain$ agents, rounded to the nearest integer. 
By \cite{mocquard2020stochastic}, we will converge to have all $\fieldbias \in \{\mu - 1, \mu, \mu +1 \}$, in $O(\log n)$ time with high probability $1 - O(1/n^2$). We use Corollary 1 of \cite{mocquard2020stochastic}, where the constant $K=O(\sqrt{n})$ and $\delta = \frac{1}{n^2}$. This gives that with probability $1-\delta$, all $\fieldbias \in \{\mu - 1, \mu, \mu +1 \}$ after a number of interactions 
\[
t \geq (n-1)(2\ln(K+\sqrt{n}) - \ln(\delta) - \ln(2)) 
\sim n(2\ln(\sqrt{n}) + \ln(n^2)) = 3n\ln n.
\]
Thus after time $3\ln n$, all $\fieldbias \in \{\mu - 1, \mu, \mu +1 \}$ with high probability $1-1/n^2$.

If $|g| > 0.5|\mains|$, then $|\mu| \geq 1$, so all remaining biased agents have the majority opinion, and we will stabilize in \phaseConsensus\ to the correct majority output.

If $|g| \leq 0.5|\mains|$, then $\mu = 0$, so now all $\fieldbias \in \{-1, 0, +1\}$. We will use \cref{lem:cancel-reactions}, with the sets of biased agents $A = \{a:a.\fieldbias = +1\}$ and $B = \{b:b.\fieldbias = -1\}$, which have initial sizes $|A| = a\cdot n$ and $|B| = b\cdot n$. 

In the first case where $0.025|\mains| \leq |g| \leq 0.5|\mains|$, we have $a - b \geq 0.025m$ (assuming WLOG that $A$ is the majority). Then by \cref{lem:cancel-reactions}, with high probability $1-O(1/n^2)$, the count of $B$ becomes $0$ in at most time $\frac{5\ln n}{2(a-b)} = \ln n\frac{5}{2\cdot 0.025m} = \frac{100}{m}\ln n \leq 201\ln n$. With all minority agents eliminated, we will again stabilize in \phaseConsensus\ with the correct output.

In the second case where $|g| < 0.025|\mains|$, we can use \cref{lem:cancel-reactions} with constant $d = b - 0.0025m$. Then even with maximal gap $a - b = 0.025m$, with very high probability in constant time we bring the counts down to $b = 0.0025m$ and $a = 0.0275m$. Thus the total count of biased agents is at most $(0.0025m + 0.0275m)n = 0.03|\mains|$.

Since all the above arguments take at most $O(\log n)$ time, for appropriate choice of counter constant $c_1$, the given behavior happens before the first $\roleclock$ agent advances to the next phase with high probability.
\end{proof}

\section{Analysis of main averaging \phaseMainAveraging}
\label{sec:analysis-main-phase}

The longest part in the proof is analyzing the behavior of the main averaging \phaseMainAveraging. The results of this section culminate in the following two theorems, one for the case of an initial tie, and the other for an initial biased distribution.

In the case of an initial tie, we will show that all biased agents have minimal $\fieldlevel= -L$ by the end of the phase. 

\begin{restatable}{theorem}{replicatedThmPhaseThreeTieResult}
\label{thm:phase-3-tie-result}
If the initial configuration was a tie with gap 0, then by the end of \phaseMainAveraging, all biased agents have $\fieldlevel = -L$, with high probability $1-O(1/n^{2})$.
\end{restatable}

Note that the where all biased agents have $\fieldlevel = -L$ gives a stable configuration in the next \phaseDetectTie, with all agents having $\fieldoutput = \T$. Thus from \cref{thm:phase-3-tie-result}, we conclude that from an intial tie, the protocol will stabilize in \phaseDetectTie\ with high probability. Conversely, we have already observed that this configuration can only be reached in the case of a tie, since the sum of all biased agents would bound the magnitude of the initial gap $|g| < 1$. Thus in the other case of a majority initial distribution, the agents will proceed through to \phaseReserveSample\ with probability 1.

In this majority case, we will show that a large majority of the $\rolemain$ agents have $\fieldopinion$ set to the majority opinion and $\fieldlevel \in \{-l, -(l+1), -(l+2)\}$ is in a consecutive range of 3 possible exponents, where the value $-l$ depends on the initial distribution. In addition, we show the upper tail above this exponent $-l$ is very small.

\begin{restatable}{theorem}{replicatedThmPhaseThreeMajorityResult}
\label{thm:phase-3-majority-result}
Assume the initial gap $|g| < 0.025|\mains|$. Let the exponent $-l = \floor{\log_2(\frac{g}{0.4|\mains|})}$.
Let $i = \text{sign}(g)$ be the majority opinion and $M$ be the set of all agents with $\fieldrole = \rolemain$, $\fieldopinion = i$, $\fieldlevel \in \{-l, -(l+1), -(l+2)\}$. Then at the end of \phaseMainAveraging, $|M| \geq 0.92|\mains|$ with high probability $1-O(1/n^{2})$.

In addition, the total mass above exponent $-l$ is $\mu_{(>-l)} = \sum\limits_{a.\fieldlevel > -l}|a.\fieldbias| \leq 0.002|\mains|2^{-l}$, and the total minority mass is $\beta_- = \sum\limits_{a.\fieldopinion = -i}|a.\fieldbias| \leq 0.004|\mains|2^{-l}$.
\end{restatable}

Note the assumption of small initial gap $g$ that we get from \phaseDiscreteAveraging\ is just for convenience in the proof, to reason uniformly about the base case behavior at hour $h=0$ for our inductive argument. The rules of \phaseMainAveraging\ would also work as intended with even larger initial gaps, just requiring a variant of the later analysis to acknowledge that the initial $\fieldexponent = 0$ is quite close to the final range 
$\fieldlevel \in \{-l,-(l+1),-(l+2)\}$.

\subsection{Clock synchronization theorems}

We will first consider the behavior of the $\roleclock$ agents in \phaseMainAveraging. The goal is to show their $\fieldminute$ fields remain tightly concentrated while moving from $0$ up to $kL$, summarized in \cref{thm:clock}. See \cref{fig:clock_distribution} for simulations of the clock distribution. The back tail behind the peak of the $\fieldminute$ distribution decays exponentially, since each agent is brought ahead by epidemic at a constant rate and thus their counts each decay exponentially. The front part of the distribution decays much more rapidly. With a fraction $f$ of agents at $\fieldminute = i$, the rate of the drip reaction is proportional to $pf^2$. This repeated squaring leads to the concentrations at the leading minutes decaying doubly exponentially. The rapid decay is key to showing that very few $\roleclock$ agents can get very far ahead of the rest.

\begin{figure}[!htbp]
     \centering
     \begin{subfigure}[b]{0.95\textwidth}
         \centering
         \includegraphics[width=\textwidth]{figures/clock_distribution_p_1.pdf}
         \caption{\footnotesize
         Simulating the clock rules with $p=1$.}
         \label{fig:clock_distribution_p_1}
     \end{subfigure}
     \hfill
        
     \centering
     \begin{subfigure}[b]{0.95\textwidth}
         \centering
         \includegraphics[width=\textwidth]{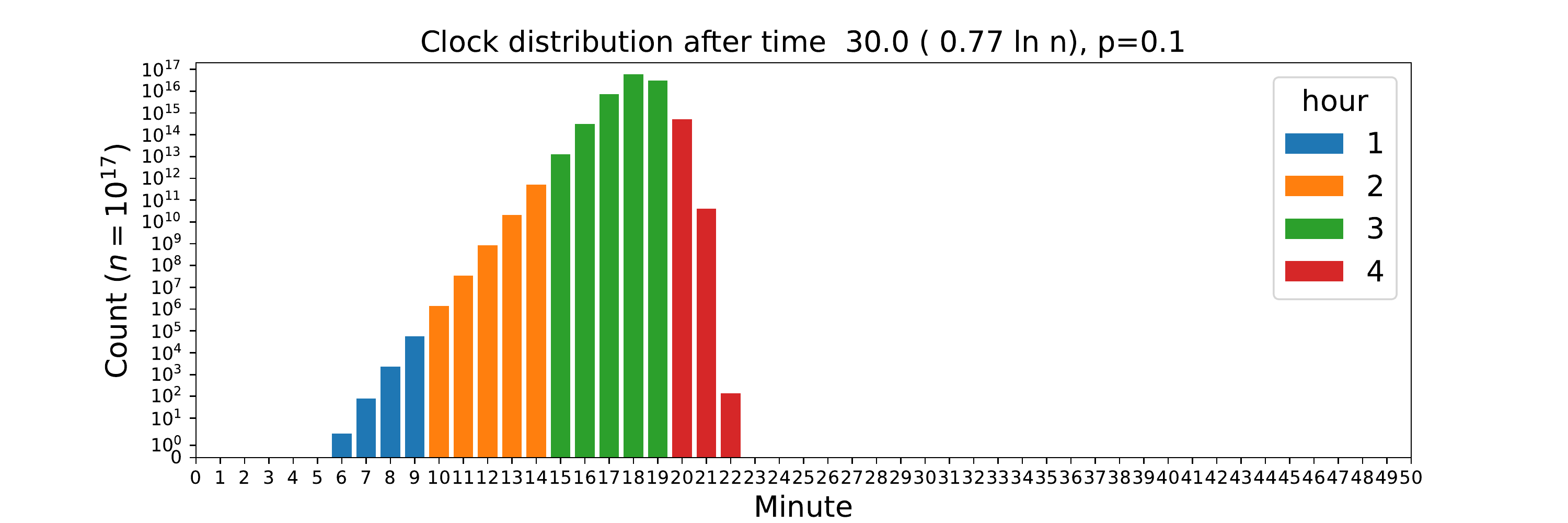}
         \caption{\footnotesize
         Simulating the clock rules with $p=\frac{1}{10}$.}
         \label{fig:clock_distribution_p_10}
     \end{subfigure}
     \hfill
     
    \centering
     \begin{subfigure}[b]{0.95\textwidth}
         \centering
         \includegraphics[width=\textwidth]{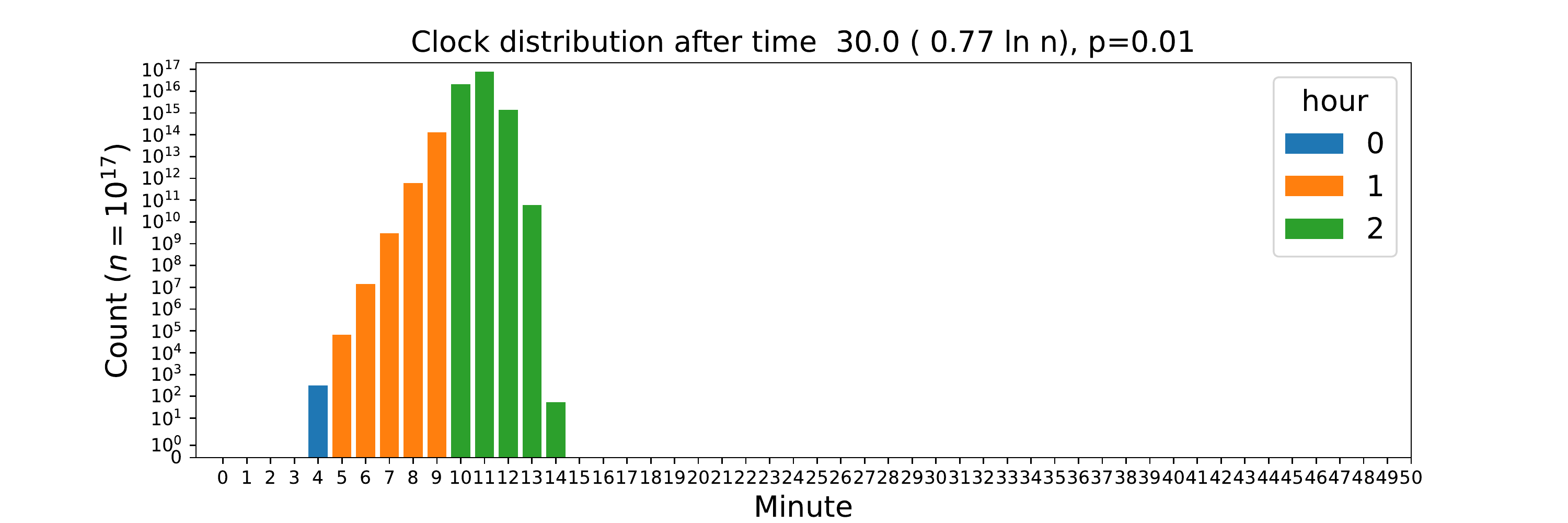}
         \caption{\footnotesize
         Simulating the clock rules with $p=\frac{1}{100}$.}
         \label{fig:clock_distribution_p_100}
     \end{subfigure}
     \caption{Simulating the clock rules on a large population of size $n = 10^{17}$, with $k=5$ minutes per hour and multiple values of $p$. 
     The $\fieldminute$ distribution gets tighter for smaller values of $p$. To make the $\fieldhour$ distribution tighter, we can also simply make $k$ larger.
     The full evolution of these distributions can be seen in the example notebook~\cite{examplenotebook}, along with the code that ran the large simulation to generate the data.}
     \label{fig:clock_distribution}
\end{figure}

While our algorithm runs for only $O(\log n)$ minutes,
the clock behavior we require can be made to continue for $O(n^c)$ minutes for arbitrary $c$.
Our proofs rely on induction on minutes, and all results in this section hold with very high probability. Thus, we can take a union bound over polynomially many minutes and still keep the very high probability guarantees. If the clock runs for a superpolynomial number of minutes, the results of this section no longer hold. 

We start by consider an entire population running the clock transitions (lines 1-5 of \phaseMainAveraging), so $|\clocks| = n$. The following lemmas describe the behavior of just this clock protocol. In our actual protocol, the clock agents are a subpopulation $|\clocks| = c \cdot n$, and these clock transitions only happen when two clock agents meet with probability $\binom{|\clocks|}{2}\big / \binom{n}{2}\sim \frac{1}{c^2}$. This more general situation is handled in later theorems applying to our exact protocol.

\paragraph{Definitions of values used in subsequent lemma statements.}
Throughout this section, we reference the following quantities.
For each minute $i$ and parallel time $t$, define $c_{\geq i}(t) = |\{c:c.\fieldminute \geq i\}| / |\clocks|$ to be the fraction of clock agents at minute $i$ or beyond at time $t$. Then define $t^+_{\geq i} = \min\{t:c_{\geq i}(t) > 0\}$, $t^{0.1}_{\geq i} = \min\{t:c_{\geq i}(t) \geq 0.1\}$ and $t^{0.9}_{\geq i} = \min\{t:c_{\geq i}(t) \geq 0.9\}$ to be the first times where this fraction becomes positive, hits $0.1$ and hits $0.9$.

We first show that the $c_{\geq i+1}$ is significantly smaller than $c_{\geq i}$ while both are still increasing, so the front of the clock distribution decays very rapidly. In our argument, we consider three types of reactions that change the counts at minutes $i$ or above:
\begin{enumerate}[(i)]
    \item \label{rxn:drip} Drip reactions $i,i \to i,(i+1)$ (\cref{line:clock-drip} in \phaseMainAveraging)
    \item \label{rxn:i-1epidemic} Epidemic reactions $j,k \to j,j$, for any minutes $j,k$ with $k \leq i < j$ (\cref{line:clock-epidemic} in \phaseMainAveraging)
    \item \label{rxn:iepidemic} Epidemic reactions $j,k \to j,j$, for any minutes $j,k$ with $k < i \leq j$ (\cref{line:clock-epidemic} in \phaseMainAveraging)
\end{enumerate}
Note we ignore the drip reactions $(i-1),(i-1) \to (i-1),i$, which would only help the argument, but we do not have guarantees on the count at minute $i-1$.

We will try to show the relationship pictured in~\cref{fig:clock_distribution}, that $c_{\geq (i+1)}(t) \approx c_{\geq i}(t)^2$. One challenge here is that this relationship will no longer hold with high probability for small values of $c_{\geq i}$. For example, if $c_{\geq i}(t) = n^{-1/2 - \epsilon}$, the desired relationship would require $c_{\geq i+1}(t) < \frac{1}{n}$, meaning the count above minute $i$ is 0. A drip reaction at minute $i$ could happen with non-negligible probability $\approx n^{-2\epsilon}$.

To handle this difficulty, we define the set of \emph{early drip agents} 
$D_{\geq i+1}(t)$,
the set of agents that moved above minute $i$ via a drip reaction at a time $\leq t$ when $c_{\geq i}(t) < n^{-0.45}$,
or that were brought above minute $i$ via an epidemic reaction with another early drip agent in $d_{\geq i+1}$. 
Note that the latter group of agents can move to minute $\geq i+1$ \emph{after} time $t^+_{\geq i+1}$.
Thus, this set represents the effect of any drip reactions that happen before $c_{\geq i}(t)$ grows large enough for our large deviation bounds to work. We then define $d_{\geq i+1}(t) = |D_{\geq i+1}(t)| / |\clocks|$ to be the fraction of early drip agents
(including effects that persist beyond this time due to epidemic reactions).
The set of agents above minute $i$ that comprise the fraction $c_{\geq (i+1)}$ is then partitioned into the early drip agents that comprise $D_{\geq i+1}$, and the rest. By first ignoring these early drip agents $D_{\geq i+1}$, we can show that the rest of $c_{\geq (i+1)}$ stays small compared to $c_{\geq i}$.

\begin{lemma}
\label{lem:clock-front-tail}
With very high probability, if $n^{-0.45} \leq c_{\geq i}(t) \leq 0.1$, then $c_{\geq i+1}(t) \leq 0.9pc_{\geq i}(t)^2 + d_{\geq i+1}(t)$.
\end{lemma}

\begin{proof}
The proof will proceed by induction on time $t$. As a base case, for all $t$ such that $c_{\geq i}(t) < n^{-0.45}$, the statement holds simply by definition of $d_{\geq i+1}(t)$, which is equal to $c_{\geq i+1}(t)$.

For the inductive step, to show the relationship $c_{\geq (i+1)}(t) < 0.9pc_{\geq i}(t)^2 + d_{\geq i+1}(t)$ holds at time $t$, we will use the inductive hypothesis that $c_{\geq (i+1)}(t - 0.1) < 0.9pc_{\geq i}(t - 0.1)^2 + d_{\geq i+1}(t - 0.1)$. Let $x(t) = c_{\geq i}(t)$ and $y(t) = c_{\geq i+1}(t) - d_{\geq i+1}(t)$, so we need to show $y(t) < 0.9px(t)^2$.

We first lower bound how much $x(t)$ grows by epidemic reactions, in order to show $x(t-0.1) < 0.84x(t)$. Using \cref{lem:epidemic}, the expected amount of time for an epidemic to grow from fraction $0.84x$ to $x$ is
$$
\frac{1}{2}\qty[\ln(x) - \ln(0.84x) + \ln(\frac{1-0.84x}{1-x})]
\leq
\frac{1}{2}\qty[- \ln(0.84) + \ln(\frac{1-0.84\cdot 0.1}{1-0.1})]
< 0.096,
$$
where we used the fact that $\frac{1-0.84x}{1-x}$ is nondecreasing and $x(t) \leq 0.1$.

The minimum probability $p^*$ is $\Theta(x(t)) = \Omega(n^{-0.45})$, and the expected number of interactions $\mu$ is $\Theta(n)$. So the application of \cref{thm:janson} will give probability $1-e^{-\Omega(n^{0.55})}$ for the epidemic to have grown enough within time $0.1$. Thus $x(t-0.1) < 0.84x(t)$ with very high probability.

Next we bound how much $y(t)$ grows, by both epidemic reactions and drip reactions. The probability of a drip reaction is at most $px(t)^2$, so the expected number of drip reactions in time $0.1$ is $0.1px(t)^2$. A standard Chernoff bound then gives that there are at most $0.11px(t)^2$ drip reactions with very high probability. We assume in the worst case all these drip reactions happen at time $t-0.1$, and then $y$ grows by epidemic starting from $z = y(t-0.1) + 0.11px(t)^2$.

By \cref{lem:epidemic}, the expected amount of time for an epidemic to grow from fraction $z$ to $1.23z$ is

$$
\frac{1}{2}\qty[\ln(1.23z) - \ln(z) + \ln(\frac{1-z}{1-1.23z})]
\geq
\frac{1}{2}\ln(1.23)
> 0.103.
$$
The minimum probability $p^* = \Omega(x(t)^2) = \Omega(n^{-0.9})$, and the expected number of interactions $\mu = \Theta(n)$. So the application of \cref{thm:janson} will give probability $1-e^{-\Omega(n^{0.1})}$.
Thus with very high probability
\begin{align*}
    y(t) & \leq 1.23[y(t-0.1) + 0.11px(t)^2]
    \\
    & \leq 1.23[0.9px(t-0.1)^2 + 0.11px(t)^2]
    \\
    & \leq 1.23[0.9p(0.84x(t))^2 + 0.11px(t)^2]
    < 0.9px(t)^2. \qedhere
\end{align*}
\end{proof}

In order to bound $c_{\geq (i+1)}(t)$ as a function of $c_{\geq i}(t)$ only (\cref{thm:clock-front-tail}),
we need to bound how large the set $D_{\geq i+1}$ of early drip agents can be. 
The strategy will be to show there is not enough time for the set $D_{\geq i+1}$ to grow very large starting from $t_{\geq i + 1}^+,$ 
the first time any agent appears at minute $\geq i+1$ (i.e., the first drip reaction into minute $i+1$), 
because there are only $O(\log\log n)$ minutes in the front tail (see~\cref{fig:clock_distribution}), 
so the time between $t^+_{\geq i + 1}$ and $t^{0.1}_{\geq i + 1}$ is $O(\log\log n)$.

We will first need an upper bound on how long it takes the clock to move from one minute to the next.

\begin{lemma}
\label{lem:clock-upper-bound}
$t_{\geq i + 1}^{0.1}-t_{\geq i}^{0.1}\leq 2.11 +\frac{1}{2}\ln\qty(\frac{1}{p})$ with very high probability.
\end{lemma}

\begin{proof}
First we argue that $c_{\geq i+1}(t_{i}^{0.1} + \frac{1}{2}) > 0.0045p$. If not, the count at minute $i$, $c_{\geq i}(t) - c_{\geq i+1}(t) > 0.1 - 0.0045 = 0.0955$ for all $t_{i}^{0.1} < t < t_{i}^{0.1} + 0.5$. Then, the probability of a drip reaction is at least $0.0955^2p > 0.0091p$. By standard Chernoff bounds, we then have that in time $\frac{1}{2}$, there are at least $0.0045p$ drip reactions with very high probability. Thus $c_{\geq i+1}(t_{i}^{0.1} + \frac{1}{2}) > 0.0045p$ from just those drip reactions alone.

Now we argue that the amount of time it takes for epidemic reactions to bring $c_{\geq (i+1)}$ up to $0.1$.
By \cref{lem:epidemic}, the expected amount of time for an epidemic to grow from fraction $0.0045p$ to $0.1$ is

$$
\frac{1}{2}\qty[\ln(0.1) - \ln(0.0045p) +
\ln(1-0.0045p) - \ln(0.9)]
<
\frac{\ln(0.1)-\ln(0.0045) - \ln(0.9)}{2}
-\frac{\ln p}{2}
<
1.603 - \frac{\ln p}{2}.
$$

As long as $p = \Theta(1)$, then the minimum probability $p^* = \Theta(p)$ is constant, and by \cref{lem:epidemic}, the epidemic takes at most time $1.61 - \frac{\ln p}{2}$ with very high probability.

In total, we then get that $t_{i+1}^{0.1}-t_{i}^{0.1}\leq 0.5+1.61 - \frac{\ln p}{2} = 2.11 +\frac{1}{2}\ln\qty(\frac{1}{p})$ with very high probability.
\end{proof}

Proving that the set $D_{\geq i+1}$ remains small will let us prove the main theorem about the front tail of the clock distribution:

\begin{theorem}
\label{thm:clock-front-tail}
With very high probability, if $n^{-0.4} \leq c_{\geq i}(t) \leq 0.1$, then $c_{\geq (i+1)}(t) < pc_{\geq i}(t)^2$.
\end{theorem}

\begin{proof}
The proof will proceed by induction on the minute $i$, where the base case is vacuous because $c_{\geq 0}(0) = 1$, so $c_{\geq 0}(t) > 0.1$ for all times $t$.

The inductive hypothesis will use two claims. The first is that the time $t^{0.1}_{\geq i} - t^+_{\geq i} = O(\log\log n)$ for minute $i$. The second is that $d_{\geq i+1}(t^{0.1}_{\geq i}) = O(n^{-0.85})$. Note that using this second claim along with \cref{lem:clock-front-tail} proves the Theorem statement at minute $i$:
when $n^{-0.4} \leq c_{\geq i}(t) \leq 0.1$, by \cref{lem:clock-front-tail} we have
$$c_{\geq (i+1)}(t) < 0.9pc_{\geq i}(t)^2 + d_{\geq i+1}(t) \leq pc_{\geq i}(t)^2,$$
because $c_{\geq i}(t)^2 \geq n^{-0.8}$, so the $d_{\geq i+1}$ term is negligible for sufficiently large $n$.

We will prove the first claim in two parts. First we argue that $t_{\geq i}^{0.1}-t_{\geq i}^{+} = O(\log \log n)$, because the width of the front tail is at most $2\log \log n$. We will show that at $t_{\geq i}^{+}$, when the first agent arrives at $\fieldminute = i$, we already have $c_{\geq j}(t_{\geq i}^{+}) \geq 0.1$, where $j = i - 2\log\log n$ (for $i < 2\log\log n$, we just have $j=0$ and there is nothing to show because the width of the front tail can be at most $i$). First we move $\log\log n$ levels back to $k = i - \log\log n$, to show $c_{\geq k}(t_{\geq i}^{+}) \geq n^{-0.4}$. 

Assume, for the sake of contradiction that $c_{\geq k}(t_{\geq i}^{+}) < n^{-0.4}$. 
Between $t_{\geq k}^+$ and $t_{\geq k}^{n^{-0.4}}$, consider the number of drips that happen from levels $k+1$ and above. By the inductive hypothesis, we have $c_{\geq k+1}(t)\leq p c_{\geq k}(t)^2 < n^{-0.8}$ during this whole time. Thus in any interaction of this period the probability of a drip above level $k+1$ is at most $p \cdot (n^{-0.8})^2 \leq n^{-1.6}$. The interval length is $t_{\geq k}^{n^{-0.4}}- t_{\geq k}^+ = O(\log\log n)$ by the inductive hypothesis, so
the probability of having at least $\log\log n$ drips during this interval is at most

\[
\binom{O(n\log\log n)}{\log\log n}(n^{-1.6})^{\log\log n}
\leq
\left(\frac{O(n\log\log n)}{n^{1.6}}\right)^{\log \log n}
= n^{-\omega(1)}.
\]
This implies $c_{\geq i}(t_{\geq k}^{n^{-0.4}})=0$ with very high probability.

Thus with very high probability, we already have $c_{\geq k}(t_{\geq i}^{+}) \geq n^{-0.4}$. Then we can iterate the inductive hypothesis $c_l(t) \leq p(c_{l-1}(t))^2$ for the $\log\log n$ minutes $l=k, k-1, \ldots, j$, which implies $c_{\geq j}(t_{\geq i}^{+}) \geq 0.1$. Now that we have shown the width of the front tail is at most $2\log\log n$, we use \cref{lem:clock-upper-bound}, which shows that each minute takes $O(1)$ time, so it takes $O(\log \log n)$ time between $t_{\geq i}^{+}$ when the first agent gets to $\fieldminute = i$ and $t_{\geq i}^{0.1}$, when the fraction at $\fieldminute \geq i$ reaches $0.1$.

Now we prove the second claim, arguing that in this $O(\log\log n)$ time, $d_{\geq i+1}$ can grow to at most $O(n^{-0.85})$. By definition of $d_{\geq i+1}$, the drip reactions that increase $d_{\geq i+1}$ happen with probability at most $p(n^{-0.45})^2 = pn^{-0.9}$. By a standard Chernoff bound, the number of drip reactions in $O(\log\log n)$ time is $O(\log\log n \cdot n^{-0.9}) = O(n^{-0.89})$. Then we assume in the worst case this maximal number of drip reactions happen, and then $d_{\geq i+1}$ can grow by epidemic. By \cref{lem:epidemic}, the time for an epidemic to grow from fraction $O(n^{-0.89})$ to $\Omega(n^{-0.85})$ is $\Omega(\log n) > O(\log\log n)$ with very high probability. Thus, with very high probability, we still have $d_{\geq i+1} = O(n^{-0.85})$ within $O(\log\log n)$ time.
\end{proof}

Now the relationship proven in \cref{thm:clock-front-tail} implies that $c_{\geq (i+1)}(t_{\geq i}^{0.1}) \leq 0.01p$. We can now use this bound to get lower bounds on the time required to move from one minute to the next.

\begin{lemma}
\label{lem:clock-lower-bound}
With very high probability, $t_{i+1}^{0.1} - t_{i}^{0.1} \geq \frac{1}{2}\ln(1+\frac{2}{9p})-0.01$.
\end{lemma}

\begin{proof}
We start at time $t_{\geq i}^{0.1}$, where by \cref{thm:clock-front-tail} we have $c_{\geq i+1}(t_{\geq i}^{0.1}) \leq 0.01p$ with very high probability.

Define $x = x(t) = c_{\geq i}(t)$ and $y = y(t) = c_{\geq i+1}(t)$.
The number of interactions for $Y = yn$ to increase by $1$ is a geometric random variable with mean $\frac{1}{\Pr[\text{\ref{rxn:drip}}] + \Pr[\text{\ref{rxn:i-1epidemic}}]}$, where
the drip reaction \ref{rxn:drip} has probability $\Pr[\text{\ref{rxn:drip}}]\sim p(x-y)^2 \leq p(1-y)^2$ and the epidemic reaction \ref{rxn:i-1epidemic} has probability $\Pr[\text{\ref{rxn:i-1epidemic}}]\sim 2y(1-y)$. Assuming in the worst case that $y(t_{i}^{0.1}) = 0.01p$, then the number of interactions $T = (t_{i+1}^{0.1}-t_{i}^{0.1})n$ for $Y$ to increase from $0.1pn$ to $0.1n$ is a sum of independent geometric random variables with mean 
\begin{align*}
\E[T] &\geq \sum_{Y=0.01pn}^{0.1n - 1} \frac{1}{p(1 - Y/n)^2 + 2(Y/n)(1-Y/n)}
\sim n\int_{0.01p}^{0.1} \frac{dy}{p(1-y)^2 + 2y(1-y)}
\\
&= n\int_{0.01p}^{0.1} \frac{dy}{(1-y)(p + (2-p)y)}
=
n\int_{0.01p}^{0.1} \frac{1/2}{1-y} + \frac{1-p/2}{p+(2-p)y}dy
\\
&= n 
\qty[-\frac{1}{2}\ln(1-y)+\frac{1}{2}\ln\qty(p+(2-p)y)]_{0.01p}^{0.1}
\\
&= \frac{n}{2}
\qty[-\ln(0.9)+\ln(1-0.01p)+\ln(p+0.1(2-p)) - \ln(p + 0.01p(2-p))]
\\
&= \frac{n}{2}
\qty[-\ln(0.9)+\ln(\frac{1-0.01p}{p-0.01p^2+0.02p})+\ln(0.9p+0.2)]
\\
&\geq\frac{n}{2}
\qty[-\ln(0.9)+\ln(\frac{0.9p+0.2}{1.02p})]
\geq
\frac{n}{2} \qty[-0.02 + \ln(1+\frac{2.04}{9p})]
\end{align*}

Note that the probability in the geometric includes the term $p(1-y)^2 \geq 0.81p$ since $y\leq 0.1$, thus the minimum geometric probability $p^* \geq 0.81p$ is bounded by a constant. Also the mean $\mu = \Theta(n)$ so by \cref{cor:janson-wvhp},
$\Pr[T \leq n(-0.01 + \frac{1}{2}\ln(1+\frac{2}{9p})] \leq \exp(-\Theta(n))$, so the time $t_{i+1}^{0.1}-t_{i}^{0.1} = T/n \geq \frac{1}{2}\ln(1+\frac{2}{9p})-0.01$ with very high probability.
\end{proof}

The worst case upper bound for the dripping probability, $p(1-y)^2$, used in the above Lemma, is weakest when $p=1$. We now give a special case lower bound that is stronger in the deterministic case with $p=1$.

\begin{lemma}
\label{lem:clock-lower-bound-p-1}
With very high probability, $t_{i+1}^{0.1} - t_{i}^{0.1} \geq 0.45$.
\end{lemma}

\begin{proof}
We assume in the worst case that $p=1$. Then by \cref{thm:clock-front-tail}, we have $c_{\geq i+1}(t_{i}^{0.1}) \leq 0.01$. This initial fraction will grow by epidemic to at most $0.01 \cdot e^{2 \cdot 0.45}(1+\epsilon)$ with very high probability by \cref{lem:epidemic}. We must also consider all agents that later make it to minute $i+1$ by a drip reaction, and how much they grow by epidemic. 

By time $t_i^{0.1} + s$, the fraction $c_{\geq i}$ could have increased to at most $0.1 + s$, since at most 1 agent can increase its minute in each interaction. Then the probability of a drip reaction at this time is at most $(0.1 + s)^2$. By standard Chernoff bounds, there will be at most $(1+\epsilon)(0.1 + s)^2n^{0.5}$ drip reactions in the next $n^{0.5}$ interactions with very high probability. Then by \cref{lem:epidemic}, these agents can grow by epidemic by at most a factor $(1+\epsilon)e^{2 \cdot (0.45 - s)}$ by time $t_i^{0.1} + 0.45$, with very high probability. Summing over consecutive groups of $n^{0.5}$ interactions at time $i\cdot \frac{n^{0.5}}{n}$ and using the union bound, we get a total bound
\begin{align*}
c_{\geq i+1}(t_{i}^{0.1} + 0.45) &\leq 0.01 \cdot e^{2 \cdot 0.45}(1+\epsilon)
+
\sum_{i=0}^{0.45n^{0.5}} (1+\epsilon)(0.1 + n^{-0.5}i)^2n^{0.5}\cdot e^{2 \cdot (0.45 - n^{-0.5}i)}
\\
&\sim (1+\epsilon)\cdot e^{0.9} \qty
[
0.01 + \int_0^{0.45}(1+s)^2 e^{-2s} ds
]
\leq
0.45.
\end{align*}

\end{proof}

We now summarize all bounds on the length of a clock minute in a single theorem:

\begin{theorem}
\label{thm:clock-bounds}
Let $t_{\geq i + 1}^{0.1}-t_{\geq i}^{0.1}$ be the time between when a fraction $0.1$ of agents have $\fieldminute \geq i$ and when a fraction $0.1$ of agents have $\fieldminute \geq i+1$.
Then with very high probability,
\[
\max\qty(0.45, \frac{1}{2}\ln(1+\frac{2}{9p})-0.01)
\leq
t_{\geq i + 1}^{0.1}-t_{\geq i}^{0.1}
\leq 
2.11 +\frac{1}{2}\ln\qty(\frac{1}{p})
\]
\end{theorem}

These bounds are shown in \cref{fig:clock-constants}, along with sampled minute times from simulation. This suggests the actual time per minute is roughly $0.75 + \frac{1}{2}\ln(\frac{1}{p})$.

\begin{figure}[!htbp]
    \centering
    \includegraphics[width=\textwidth]{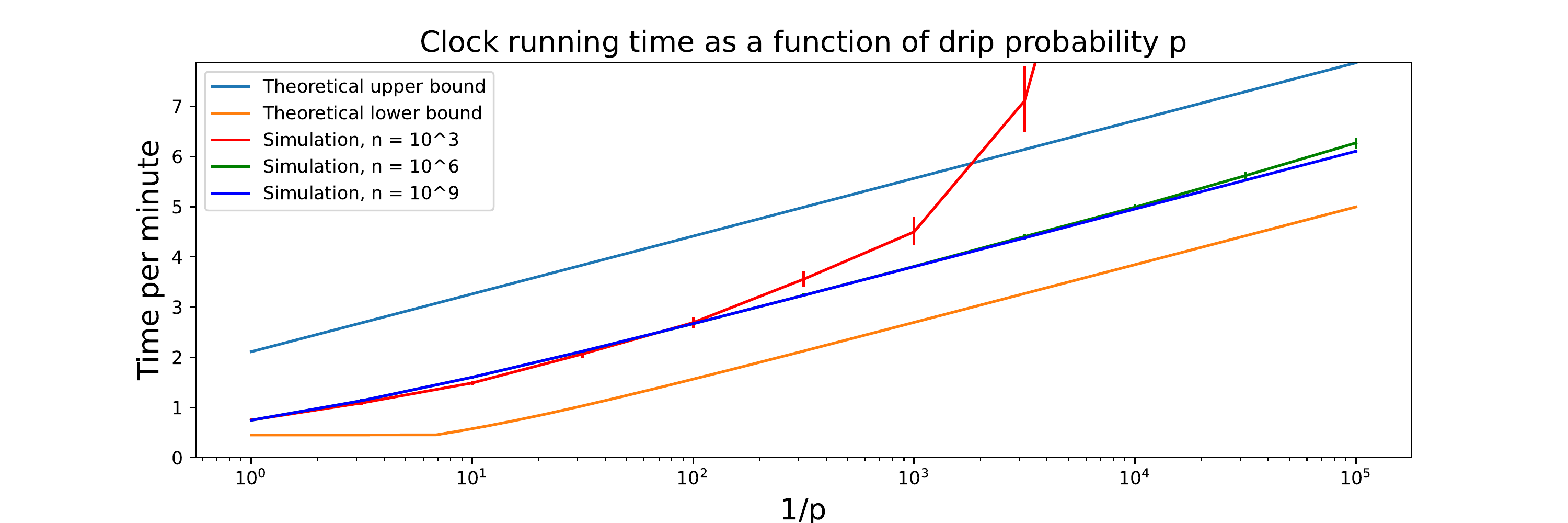}
    \caption{
    The upper and lower bounds from \cref{thm:clock-bounds} for the time of one clock minute, along with samples from simulation.
    For each value of $n$, 100 minute times were sampled, taking $t_{i+1}^{0.1} - t_i^{0.1}$ for $i=9, \ldots, 18$
    over 10 independent trials.
    All our proofs assume $p$ is constant, and for any fixed value of $p$, will only hold for sufficiently large $n$. The case $n=10^3$ shows that when $p = O(1/n)$, the bounds no longer hold. This is to be expected because the expected number of drips becomes too small for large deviation bounds to still hold.
    }
    \label{fig:clock-constants}
\end{figure}

We can now build from these theorems to get bounds on the values of $\fieldhour$.
For a clock agent $a$, define $a.\fieldhour = \floor{\frac{a.\fieldminute}{k}}$.
Define $\hourstart_h = \min\big(t: |\{a:a(t).\fieldhour \geq h\}| \geq 0.9|\clocks| \big)$ be the first time when the fraction of clock agents at hour $h$ or beyond reaches $0.9$ and $\hourend_h = \min\big(t: |\{a:a(t).\fieldhour > h\}| \leq 0.001|\clocks|\big)$ be the first time when the fraction of clock agents beyond hour $h$ reaches $0.001$. 
Define the \emph{synchronous hour h} to be the parallel time interval 
$[\hourstart_h, \hourend_h]$, 
i.e. when fewer than 0.1\%  are in any hour beyond $h$ and at least (90 - 0.1) = 89.9\% are in hour $h$. Note that if $\hourend_h < \hourstart_h$ then this interval is empty, but we show this happens with low probability.
We choose the threshold $0.001|\clocks|$ to be a sufficiently small constant for later proofs.

Recall $c = |\clocks|/n$ is the fraction of clock agents and
 $k$ is the number of minutes per hour.

\begin{theorem}
\label{thm:clock}
Consider a fraction $c$ of agents running the clock protocol, with $p=1$. Then
for every synchronous hour $h$, its length $\emph{\hourend}_h - \emph{\hourstart}_h \geq \frac{1}{c^2}[0.45k - 3.1]$, with very high probability.
The time between consecutive synchronous hours $\emph{\hourstart}_{h+1} - \emph{\hourstart}_h \leq \frac{1}{c^2}[2.11k + 2.2]$ with very high probability.
\end{theorem}

\begin{proof}
Note that the previous lemmas assumed $|\clocks| = n$, so the entire population was running the clock algorithm. In reality, we have a fraction $c = |\clocks| / n$ of clock agents. The reactions we considered only happen between two clock agents, which interact with probability $\sim c^2$. Thus we can simply multiply the bounds from our lemmas by the factor $\frac{1}{c^2}$ to account for the number of regular interactions for the requisite number of clock interactions to happen, which is very close to its mean with very high probability by standard Chernoff bounds.

Because our definition references time $t_{\geq i}^{0.9}$ when the fraction reaches $0.9$, we will first bound the time it takes an epidemic to grow from $0.1$ to $0.9$. By \cref{lem:epidemic}, this takes parallel time $t$, where
\[
\E[t]\sim\frac{\ln(0.9) - \ln(1 - 0.9) - \ln(0.1) + \ln(1-0.1)}{2} = \ln(9) < 2.2.
\]
Since $\ln(9)(1+\epsilon) < 2.2$, this completes within parallel time $2.2$ with very high probability.

Since $c.\fieldhour = h \iff hk \leq c.\fieldminute < (h+1)k$, the times
$\hourstart_h = t_{h k}^{0.9}$ and $\hourend_{h}=t_{h k}^{0.001}$.
For the upper bound, by \cref{lem:clock-upper-bound} we have $t_{i+1}^{0.1}-t_{i}^{0.1}\leq 2.11 / c^2$ with very high probability. We start at $t_{hk}^{0.9} \geq t_{hk}^{0.1}$, and sum over the $k$ minutes in hour $h$, then add at most time $\frac{2.2}{c^2}$ between $t_{(h+1)k}^{0.1}$ and $t_{(h+1)k}^{0.1}$. This gives that $\emph{\hourstart}_{h+1} - \emph{\hourstart}_h \leq \frac{1}{c^2}[2.11k + 2.2]$ with very high probability.

For the lower bound, by \cref{thm:clock-front-tail}, at time $t_{(h+1)k - 2}^{0.1}$, we have $c_{\geq (h+1)k} \leq (0.1^2)^2 = 10^{-4} < 0.001$, so $t_{(h+1)k - 2}^{0.1} < \hourend_{h}$. Then $\hourstart_{h} \leq t_{hk}^{0.1} + \frac{2.2}{c^2}$. 
Using
 \cref{lem:clock-lower-bound-p-1}, we have $t_{i+1}^{0.1}-t_{i}^{0.1}\leq 0.45 / c^2$ with very high probability, for each $i=hk, \ldots, hk + k - 3$. All together, this gives
 $\emph{\hourend}_h - \emph{\hourstart}_h \geq \frac{1}{c^2}[0.45(k - 2) - 2.2] = \frac{1}{c^2}[0.45k - 3.1]$.
\end{proof}

We will set $k = \kconstant$ to give us sufficiently long synchronous hours for later proofs. 
Since every hour takes constant time with very high probability, all $L$ hours will finish within $O(\log n)$ time.

We finally show one more lemma concerning how the $\roleclock$ agents affect the $\rolemain$ agents.
The $\roleclock$ agents will change the $\fieldTlevel$ of the $\unbiased$ agents via \cref{line:clock-update-reaction} of \phaseMainAveraging. There are a small fraction $0.001|\clocks|$ of $\roleclock$ agents that might be running too fast and thus have a larger $\fieldhour$ than the synchronized hour. We must now show these agents are only able to affect a small fraction of the $\rolemain$ agents. Intuitively, the $\roleclock$ agents with $\fieldhour > h$ bring up the hour of both $\unbiased$ agents and other $\roleclock$ agents. The following lemma will show they don't affect too many $\rolemain$ agents before also bringing in a large number of $\roleclock$ agents.

We will redefine $c_{> h} = c_{> h}(t) = |\{c:c.\fieldhour > h\}| / |\clocks|$ to be the fraction of $\roleclock$ agents beyond hour $h$, and similarly $m_{>h} = m_{>h}(t) = |\{m:m.\fieldhour > h\}| / |\mains|$ to be the fraction of $\rolemain$ agents beyond hour $h$.

\begin{lemma}
\label{lem:mains-small-hour}
For all times $t \leq \hourend_h$, we have
$m_{>h}(t) \leq 1.2 c_{> h}(\hourend_h) = 0.0012$ with very high probability.
\end{lemma}

\begin{proof}
We have $c_{>h}(\hourend_h) = 0.001$ by the definition of synchronous hour $h$.
Thus it suffices to show that $m_{>h}(t) \leq 1.2 c_{> h}(\emph{\hourend}_h)$.

Recall $c = |\clocks| / n$ and $m = |\mains| / n$, 
so $c \cdot c_{>h}$ and $m \cdot m_{>h}$ are rescaled to be fractions of the whole population. 

We will assume in the worst case that every $\rolemain$ agent can participate in the clock update reaction
\[
C_h, M_j \to C_h, M_h \text{ where } h > j,
\]
so the probability of clock update reaction is $2 c \cdot c_{>h} \cdot m(1-m_{>h})$.
Among the $\roleclock$ agents, we will now only consider the epidemic reactions $C_h, C_j \to C_h, C_h$ between agents in different hours $h > j$, so the probability of the clock epidemic reaction is $2 c^2 \cdot c_{>h}(1-c_{>h})$.

We use the potential $\Phi(t) = m_{> h}(t) - 1.1\cdot c_{> h}(t)$. Note that initially $\Phi(0) = 0$ since both terms are $0$. The desired inequality is $\Phi(\hourend_h) \leq 0.1c_{>h}(\hourend_h) = 0.0001$, and we will show this holds with very high probability by Azuma's Inequality because $\Phi$ is a supermartingale. The clock update reaction increases $\Phi$ by $\frac{1}{|\mains|} = \frac{1}{mn}$. The clock epidemic reaction decreases $\Phi$ by $\frac{1.1}{|\clocks|} = \frac{1.1}{cn}$. This gives an expected change
\begin{align*}
\E[\Phi(t+1/n)-\Phi(t)] &= \frac{1}{mn}\qty [ 2 c \cdot c_{>h} \cdot m(1-m_{>h})] - \frac{1.1}{cn}\qty [ 2 c^2 \cdot c_{>h}(1-c_{>h})]
\\
& = \frac{2c \cdot c_{>h}}{n}\qty[(1-m_{> h}) - 1.1(1- c_{>h})]
\\
& \leq \frac{2c \cdot c_{\geq i}}{n}\qty[1 - 1.1(0.999)] < 0.
\end{align*}

Thus the sequence $\big(\Phi_j = \Phi(\frac{j}{n})\big)_{j=0}^{n \cdot \hourend_h}$ is a supermartingale, with bounded differences
$|\Phi_{j+1}-\Phi_j| \leq 
\max\qty(\frac{1}{mn}, \frac{r}{cn}) =
O(\frac{1}{n})$. 
So we can apply Azuma's Inequality (\cref{thm:azuma}) to conclude
\[
\Pr[\Phi_{n \cdot \hourend_h} \geq \delta] \leq \exp \qty(-\frac{\delta^2}{2\sum_{j=0}^{n \cdot \hourend_h} O(\frac{1}{n})^2})
= \exp(-O(n)\delta^2),
\]
since by \cref{thm:clock}, we have time $\hourend_h = O(1)$ with very high probability. Thus we can choose $\delta = 0.1c_{>h}(\hourend_h) = 0.0001$ to conclude that $m_{>h}(\hourend_h) \leq 1.2c_{>h}(\hourend_h)$ with very high probability $1 - \exp(-\Omega(n)) = 1 - O(n^{-\omega(1)})$.

The lemma statement for times $t \leq \hourend_h$ simply follows from the monotonicity of the value $m_{>h}$, since agents only decrease their $\fieldhour$ field.
\end{proof}

\subsection{\phaseMainAveraging\ exponent dynamics}
We now analyze the behavior of the $\rolemain$ agents in \phaseMainAveraging. We will show their $\fieldlevel$ fields stay roughly synchronized with the $\fieldhour$ of the $\roleclock$ agents, decreasing from $0$ toward $-L$. 
We first use the above results on the clock to conclude that during synchronous hour $h$, the tail of either $\unbiased$ agents with $\fieldTlevel > h$ or biased agents with $\fieldlevel < -h$ is sufficiently small.

\begin{lemma}
\label{lem:level-lower-tail}
Until time $\emph{\hourend}_h$, the count $|\{\unbiased:\unbiased.\fieldTlevel > h\} \cup \{b:b.\fieldlevel < -h\}| \leq 0.0024|\mains|$ with very high probability.
\end{lemma}

\begin{proof}
By \cref{lem:mains-small-hour}, the count of $\unbiased$ agents that have been pulled up by a $\roleclock$ agent to $\fieldhour > h$ is at most $0.0012|\mains|$ with very high probability. 
Each of these agents could participate in a split reaction that results in two biased agents with $\fieldlevel < -h$, increasing the total count of $|\{\unbiased:\unbiased.\fieldTlevel > h\} \cup \{b:b.\fieldlevel < -h\}|$ by 1. Thus this total count can be at most twice as large: $\leq 0.0024|\mains|$.
\end{proof}

Our main argument will proceed by induction on synchronous hours.
During each synchronous hour $h$, we will first show that at least a constant fraction $0.77|\mains|$ of agents have $\fieldbias = \T$ with $\fieldhour = h$. Then we will show that split reactions bring most biased agents down to $\fieldlevel = -h$. Finally we will show that cancel reactions will reduce the count of biased agents, leaving sufficiently many $\unbiased$ agents for the next step of the induction.

Recall that the initial gap 
\[g = (|\{a:a.\fieldinput = \A\}| - |\{b:b.\fieldinput = \B\}|) = \sum_{m.\fieldrole = \rolemain} m.\fieldbias
\]
is both the difference between the counts of $\A$ and $\B$ agents in the initial configuration and the invariant value of the total bias.
We define $\beta_+ = \beta_+(t) = \sum_{a.\fieldopinion = +1} |a.\fieldbias|$ and $\beta_- = \beta_-(t) = \sum_{b.\fieldopinion = -1} |b.\fieldbias|$ to be the total unsigned bias of the positive $\A$ agents at time $t$ and negative $\B$ agents at time $t$. Then the net bias $g = \beta_+ - \beta_-$ is the invariant.

We also define $\mu = \mu(t) = \beta_+ + \beta_-$ to be the total unsigned bias, which we interpret as ``mass''. Note that split reactions preserve $\mu$, while cancel reactions strictly decrease $\mu$. The initial configuration has mass $\mu(0) = n$, and if we eliminate all of the minority opinion then we will get $\beta_- = 0$ and $\mu = g$.
Thus decreasing the mass shows progress toward reaching consensus. 
Also, if every biased agent decreased their exponent by 1, this would exactly cut the $\mu$ in half. We will show an upper bound on $\mu$ that cuts in half after each synchronous hour, which implies that on average all biased agents are moving down one exponent. 

We also define $\mu_{(>-i)}(t) = \sum_{m.\fieldlevel > -i}|m.\fieldbias|$ as the total mass of all biased agents above exponent $-i$. Note that a bound $\mu_{(>-i)} \leq x2^{-i+1}$ gives a bound on total count $|\{a:a.\fieldlevel > -i\}| \leq x$, since even if all agents above exponent $-i+1$ split down to exponent $-i + 1$, they would have count at most $x$. Also note that $\mu(t)$ and $\mu_{(>-i)}(t)$ are nonincreasing in $t$, since the mass above a given exponent can never increase.

This inductive argument on synchronous hours will stop working once we reach a low enough exponent that the gap has been sufficiently amplified. 
We define $g_i = g \cdot 2^{i}$, which we call the relative gap to hour $i$ / exponent $-i$, because if all agents had $\fieldlevel = -i$, then a gap $g_i = |\{a:a.\fieldopinion = +1\}| - |\{b:b.\fieldopinion = -1\}|$ would maintain the invariant $g = \sum_v v.\fieldbias = \sum_{a.\fieldopinion = +1} \frac{1}{2^i} - \sum_{b.\fieldopinion = -1} \frac{1}{2^i} = \frac{g_i}{2^i}$. Note that $g_0 = g$, and the relative gap doubles as we increment the hour and decrement the exponent. So if the $\rolemain$ agents have roughly synchronous exponents, there will be some minimal exponent where the relative gap has grown to exceed the number of $\rolemain$ agents $|\mains|$, and there are not enough agents available to maintain the invariant using lower exponents.

We now formalize this idea of a minimal exponent. In the case where there is an initial majority, $g_i \neq 0$, and because we assume the majority is $\A$, we have $g_i > 0$. Define the minimal exponent $-l = \floor{\log_2(\frac{g}{0.4|\mains|})}$ to be the unique exponent corresponding to relative gap $0.4|\mains| \leq g_l < 0.8|\mains|$. For larger exponents $-i \geq -l+5 = -(l - 5)$, we have $g_i \leq g_{l-5} < 0.025|\mains|$. Thus for hours $0,1,\ldots, l-5$ and exponents $0, -1, \ldots, -l + 5$, the gap is still very small. The small gap makes the inductive argument stronger, and we will use this small gap to show a high rate of cancelling keeps shrinking the mass $\mu$ in half each hour and keeps the count of $\unbiased$ agents large, above a constant fraction $0.77|\mains|$.

For the last few hours $l-4, \ldots, l$ and exponents $-l+4, \ldots, -l$, the doubling gap weakens the argument. Thus we have separate bounds for each of these hours. These weaker bounds acknowledge the fact that the majority count is starting to increase while the count of $\unbiased$ agents is starting to decrease.

\begin{theorem}
\label{thm:induction-on-hours}
With very high probability the following holds.
For all synchronous hours $h = 0, \ldots, l$, during times $[\emph{\hourstart}_h + \frac{2}{c}, \emph{\hourend}_h]$, the count $\unbiased_h$ of $\unbiased$ agents at $\fieldTlevel = h$,
obeys 
$|\unbiased_h| \geq \tau_h|\mains|$. 
Then by time $\emph{\hourstart}_h + \frac{2}{c} + \frac{41}{m}$, the mass $\mu_{(>-h)}$ above exponent $-h$ satisfies $\mu_{(>-h)} \leq 0.001 \cdot 2^{-h+1}$. Then by time $\emph{\hourstart}_h + \frac{2}{c} + \frac{47}{m} \leq \emph{\hourend}_h$, the total mass $\mu(\emph{\hourend}_h) \leq \rho_h|\mains| 2^{-h}$.

The constant $\rho_h = 0.1$ for $h \leq l-5$. Then we have $\rho_{l-4} = 0.104$, $\rho_{l-3} = 0.13$, $\rho_{l-2} = 0.212$, $\rho_{l-1} = 0.408$, and $\rho_{l} = 0.808$.

The constant $\tau_h \geq 0.97 - 2\rho_{(h-1)}$, so $\tau_h \geq 0.77$ for $h \leq l-4$, and the minimum value $\tau_l \geq 0.15$.
\end{theorem}

Note that in the case of a tie, $l$ is undefined since we always have gap $g_i = 0$. Here the stronger inductive argument will hold for all hours and exponents. In the tie case, we will technically define $l = L+5$ so the stronger $h \leq l-5$ bounds apply to all hours.

We will prove the three sequential claims via three separate lemmas. The first argument of \cref{lem:induction-t-count}, where the clock brings a large fraction of $\unbiased$ agents up to hour $h$, will need parallel time $\frac{2}{c}$. The second argument of \cref{lem:induction-upper-tail}, where the $\unbiased$ agents reduce the mass above exponent $-h$ by split reactions, will need parallel time $\frac{41}{m}$. The third argument of \cref{lem:induction-total-mass}, where cancel reactions at exponent $-h$ reduce the total mass, will need parallel time $\frac{6}{m}$. Thus the total time we need in a synchronous hour is
\[
\frac{2}{c} + \frac{47}{m}
\leq
\frac{49}{c}
\leq
\frac{17}{c^2}
\leq
\frac{0.45\cdot 45 - 3.1}{c^2},
\]
where we use that $c < \frac{1}{3} < m$ by \cref{lem:phase-initialize}. Thus by \cref{thm:clock}, since we have constant $k = \kconstant$ minutes per hour, each synchronous hour is long enough with very high probability.

The argument proceeds by induction, with each lemma using the previous lemmas and the inductive hypotheses at previous hours. The base case comes from \cref{lem:discrete-averaging}, where we have that the initial gap $|g| < 0.025m$, so $l-5 \geq 0$, and the starting mass is at most $0.03|\mains|$. This mass bound gives the base case for \cref{lem:induction-total-mass}, whereas since $h=0$ is the minimum possible hour, the base cases for \cref{lem:induction-t-count} and \cref{lem:induction-upper-tail} are trivial.

\begin{lemma}
\label{lem:induction-t-count}
With very high probability the following holds.
For each hour $h = 0, \ldots, l$, during the whole synchronous hour $[\emph{\hourstart}_h, \emph{\hourend}_h]$, the count of $\unbiased$ agents $|\unbiased| \geq (0.9976 - 2\rho_{(h-1)})|\mains|$. Then during times $[\emph{\hourstart}_h + \frac{2}{c}, \emph{\hourend}_h]$, the count of $\unbiased$ agents at $\fieldTlevel = h$, $|\unbiased_h| \geq (0.97 - 2\rho_{(h-1)})|\mains|$.
\end{lemma}

\begin{proof}
We show the first claim, that during synchronous hour $h$, the count $|\unbiased| \geq (0.9976 - 2\rho_{(h-1)})|\mains|$ by process of elimination. By \cref{lem:induction-total-mass}, by time $\hourend_{h-1} < \hourstart_{h}$, the total mass $\mu\leq \rho_{(h-1)}|\mains| 2^{-h+1}$, which implies the total count $\{a:a.\fieldlevel \geq -h\} \leq 2\rho_{(h-1)}|\mains|$ in all future configurations, since total mass is nonincreasing. Then by \cref{lem:level-lower-tail}, the count $|\{\unbiased:\unbiased.\fieldTlevel > h\} \cup \{b:b.\fieldlevel < -h\}| \leq 0.0024|\mains|$ until time $\hourend_h$. 
This leaves $|\mains|-0.0024|\mains|-2\rho_{(h-1)}|\mains| = (0.9976 - 2\rho_{(h-1)})|\mains|$ agents who must be $\unbiased$ agents with $\fieldTlevel \geq h$ until time $\hourend_h$. 

By definition of time $\hourstart_h$, there are at least $0.9|\clocks|$ $\roleclock$ agents with $\fieldhour \geq h$ that will bring these $\unbiased$ agents up to $\fieldTlevel = h$. We need to bring all but $0.0276|\mains|$ of these agents up to $\fieldhour = h$ to achieve the desired bound $|\unbiased_h| \geq (0.97 - 2\rho_{(h-1)})|\mains|$.
We can apply \cref{lem:one-sided-cancel}, with $a = 0.9c$, $b_1 = (0.9976 - 2\rho_{(h-1)})m < (0.9976 - 2\cdot 0.1)m$, and $b_2 = 0.0276m$ to conclude this happens after parallel time $t$, where $t < (1+\epsilon)\E[t]$ with very high probability. The expected time
\[
\E[t] = \frac{\ln(b_1) - \ln(b_2)}{2a}
\leq \frac{\ln(0.7976m) - \ln(0.0276m)}{2\cdot 0.9c}
\leq \frac{1.89}{c}.
\]
Thus for small constant $\epsilon > 0$, we have $t < (1+\epsilon)\frac{1.89}{c} < \frac{2}{c}$ with very high probability.


\end{proof}

In order to reason about the total mass $\mu_{(>-h)}$ above exponent $-h$, we will define the potential function $\phi_{(>-h)}$, where for a biased agent $a$ with $a.\fieldlevel = -i \geq -h+1$ at time $t$, $\phi_{(>-h)}(a,t) = 4^{-i + h - 1}$. The global potential 
\[
\phi_{(>-h)}(t) = \sum_{a.\fieldlevel > -h}\phi(a,t)
\geq \sum_{a.\fieldlevel > -h}4^{-h+1+h-1} = |\{a:a.\fieldlevel > -h\}|.
\]
Since $\phi_{(>-h)}$ upper bounds the count above exponent $h$, we can bound the mass $\mu_{(>-h)}(t) \leq 2^{-h+1}\phi_{(>-h)}(t)$. Also note that unlike the mass, $\phi_{(>-h)}(t)$ strictly decreases via split reactions since $4^{-i} > 4^{-i-1} + 4^{-i-1}$. This will let us show that $\phi_{(>-h)}$ exponentially decays to $0$ when there are a constant fraction of $\unbiased$ agents to do these splits.

We first show that by hour $h$ exponents significantly far above $-h$ are empty. Letting $q = \lfloor\frac{\ln n}{3}\rfloor$, we will show the potential $\phi_{(>-h+q)}$ hits $0$ by hour $h$.

\begin{lemma}
\label{lem:upper-levels-empty}
For each hour $h = 0, \ldots, l$, by time $\hourstart_{h}+\frac{2}{c}$, the maximum level among all biased agents is at most $-h + q$ with high probability $1 - O(1/n^{12})$.
\end{lemma}

\begin{proof}
The statement is vacuous for hours $h < q$, so we must only consider hours $h \geq q$. We use the potential
 $\phi_{(>-h+q)}$, and start the argument at time $t_{\text{start}} = \hourstart_{(h-q)}+\frac{2}{c} + \frac{41}{m}$ where inductively by \cref{lem:induction-upper-tail} $\phi_{(>-h+q)}(t_{\text{start}}) \leq 0.001|\mains|$. We must show that by time $t_{\text{end}} = \hourstart_h+\frac{2}{c}$, we have $\phi_{(>-h+q)}(t_{\text{end}}) = 0$.
 
 By \cref{lem:induction-t-count}, the count of $\unbiased$ agents with $\fieldhour \geq h - q$, $|\unbiased_{(\geq h-q)}| \geq 0.77|\mains|$ during all synchronous hours after $\hourstart_{(h-q)} + \frac{2}{c}$ and up through synchronous hour $l-5 \geq h - 5$. Thus the interval we consider consists of at least $q - 5$ synchronous hours.
 By \cref{thm:clock}, each synchronous hour takes at least time 
$
[1.48(\kconstant-2)-2.2]/c^2 \geq 17/c^2 \geq 51/m
$, since $c < \frac{1}{3} < m$ by \cref{lem:phase-initialize}. Thus the total time for this argument is at least $t_{\text{end}} - t_{\text{start}} \geq \frac{51(q-5)}{m}$ time.

For each of the $\frac{51(q-5)}{m}n$ interactions in this interval, we consider the expected change to $\phi_{(>-h+q)}$.
 For each biased agent $a$ with $a.\fieldlevel = -i > -h + q$, a split reaction will change the potential by $\Delta\phi_{a} \leq 4^{h-1}(2\cdot 4^{-i-1} - 4^{-i}) = 4^{h-1}(-\frac{1}{2}4^{-i}) = -\frac{1}{2}\phi_a$. Then in each interaction at parallel time $t$, the expected change in the potential
\begin{align*}
\E[\phi_{(>-h+q)}(t+1/n) - \phi_{(>-h+q)}(t)] 
& \leq \sum_{a.\fieldlevel > (-h+q)}\Pr[\text{$a$ splits}]\cdot \Delta\phi_{a}
\\&\leq \sum_{a.\fieldlevel > (-h+q)} \frac{2\cdot 0.77 m}{n} \cdot -\frac{1}{2}\phi_{a}
= -\frac{0.77 m}{n}\phi_{(>-h+q)}(t).
\end{align*}
Then we have $\E[\phi_{(>-h+q)}(t+1/n)|\phi_{(>-h+q)}(t)] \leq (1-\frac{0.77 m}{n})\phi_{(>-h+q)}(t)$. 

We now recursively apply this bound to all $\frac{51(q-5)}{m}n$ interactions beginning at time $t_{\text{start}}$:

\begin{align*}
\E\qty[\phi_{(>-h+q)}(t_{\text{end}}) | \phi_{(>-h+q)}(t_{\text{start}})] 
&\leq \qty( 1-\frac{0.77 m}{n} )^{\frac{51(q-5)}{m}n} \phi_{(>-h+q)}(t_{\text{start}})
\\
\E\qty[\phi_{(>-h+q)}(t_{\text{end}})]
&\leq 
\exp(-0.77 \cdot 51 \ln n / 3)\cdot 
\exp(51 \cdot 5 \cdot 0.77)\cdot 
0.001|\mains|
\\
&\leq n^{-13}
\exp(197)\cdot 
0.001mn
\\
&= O\qty(n^{1-13})
= O\qty(n^{-12}).
\end{align*}
Finally, since $\phi_{(>-h+q)}$ takes nonnegative integer values, we can apply Markov's Inequality to conclude $\Pr[\phi_{(>-h+q)}(t_{\text{end}}) > 0] = O(1/n^{12})$.
\end{proof}

Now we can reason about the potential $\phi_{(>-h)}$ during hour $h$, which will decrease by a large constant factor. The upper bound on $\phi_{(>-h)}$ gives the a bound on the mass of the upper tail $\mu_{(>-h)}$.

\begin{lemma}
\label{lem:induction-upper-tail}
For each hour $h = 0, \ldots, l$, by time $\emph{\hourstart}_h + \frac{2}{c} + \frac{41}{m}$, the potential $\phi_{(>-h)} \leq 0.001|\mains|$ with very high probability. This implies the mass above exponent $-h$ is $\mu_{(>-h)} \leq 0.001 |\mains| 2^{-h+1}$.
\end{lemma}

\begin{proof}
By \cref{lem:induction-t-count}, after time $\hourstart_h + \frac{2}{c}$, we have a count $|\unbiased_h| \geq \tau_h|\mains|$, where the weakest bound is at hour $l$, where $\tau_l = 0.15$.

Inductively, we have $\phi_{(>-h+1)}(\hourstart_h) \leq 0.001|\mains|$ by time $\hourstart_{(h-1)} + \frac{2}{c} + \frac{41}{m}$. By \cref{lem:induction-total-mass}, by time $\hourend_{(h-1)}$, the total mass $\mu \leq \rho_{(h-1)}|\mains|2^{-h+1} \leq \rho_{(l-1)}|\mains|2^{-h+1}$. Thus there are at most $\rho_{(l-1)}|\mains| = 0.408|\mains|$ biased agents with exponent $-h+1$, which lets us bound the potential $\phi_{(>-h)}(\hourstart_h + \frac{2}{c}) \leq (0.408 + 4\cdot 0.001)|\mains| = 0.412|\mains|$. Thus we must drop the potential by the constant factor $412$.

To show that $\phi$ drops by a constant factor, we will use $\Phi(t) = \ln(\phi_{(>-h)}(t))$, which will be a supermartingale. If agent $a$ at with exponent $-i$ splits, with $\phi_a = 4^{-i+h-1}$, this changes the potential $\phi$ by $\Delta\phi_a = -\frac{1}{2}\phi_a$. The potential $\Phi$ then changes by
\[
\Delta\Phi_a = \ln(\phi_{(>-h)} + \Delta\phi_a) - \ln(\phi_{(>-h)}) = \ln(1 + \frac{-\frac{1}{2}\phi_a}{\phi_{(>-h)}}) \leq - \frac{\phi_a}{2\phi_{(>-h)}}.
\]
The expected change in $\Phi$ is then
\[
\E[\Delta\Phi] \leq \sum_{a.\fieldlevel > -h}\Pr[\text{$a$ splits}]\cdot \Delta\Phi_{a}
\leq \sum_{a.\fieldlevel > -h} \frac{2\tau_h m}{n} \cdot -\frac{\phi_{a}}{2\phi_{(>-h)}}
= -\frac{0.15 m}{n}.
\]

We define the supermartingale $(\Phi_j)_{j=0}^{\frac{41}{m}n}$, where $\Phi_j = \Phi(\hourstart_h + \frac{2}{c} + \frac{j}{n})$. Then the desired inequality is $\phi_{>(h+1)}(\hourstart_h+\frac{2}{c}+\frac{41}{m}) \leq 0.001|\mains| \leq \phi_{>(h+1)}(\hourstart_h+\frac{2}{c}) / 412$, so we need to show
$\Phi_{\frac{41}{m}n} - \Phi_0 \leq \ln(\frac{1}{412})$. The expected value 
\[
\E[\Phi_{\frac{41}{m}n} - \Phi_0] \leq -\frac{0.15m}{n}\cdot\frac{41}{m}n = -6.15 \leq \ln(\frac{1}{412}) - 0.12.
\]

To apply Azuma's Inequality, we will need a bound on the difference $|\Phi_{j+1}-\Phi_j|\leq \max\qty|\frac{\phi_a}{2\phi_{(>-h)}}|$. During the interval we consider, $\phi_{(>-h)} \geq 0.001|\mains|$, since after that the desired inequality will hold.
By \cref{lem:upper-levels-empty}, by time $\hourstart_h + \frac{2}{c}$, the maximum exponent in the population is at most $-h + q$, where $q = \floor{ \frac{\ln n}{3} }$, so $\phi_a \leq 4^{(-h+q)+h-1} = 4^{q-1}$. Using the fact that $4^q = e^{\ln 4 \ln n / 3} = O(n^{0.47})$, we can bound the largest change in $\Phi$ as
\[
|\Phi_{j+1}-\Phi_j| \leq \frac{4^{q-1}}{0.002|\mains|}
= O\qty(\frac{4^q}{n})=O(n^{0.47}/n)=O(1/n^{0.53}).
\]

Now by Azuma's Inequality (\cref{thm:azuma}) we have
\[
\Pr\big[ (\Phi_{\frac{41}{m}n} - \Phi_{0}) - \E[\Phi_{\frac{41}{m}n} - \Phi_0] \geq 0.12] \leq \exp(-\frac{0.12^2}{2\sum_{j=0}^{\frac{41}{m}n} (O(n^{-0.53}))^2}) = \exp(-\Theta(n^{0.06})).
\]
Thus $\phi_{(>-h)}$ will drop by at least the constant factor $412$ with very high probability, giving $\phi_{(>-h)}(\hourstart_h + \frac{2}{c} + \frac{41}{m}) \leq 0.001|\mains|$.
\end{proof}

\begin{lemma}
\label{lem:induction-total-mass}
For each hour $h = 0, \ldots, l$, by time $\emph{\hourstart}_h + \frac{2}{c} + \frac{47}{m} \leq \emph{\hourend}_h$, the total mass $\mu \leq \rho_h |\mains| 2^{-h}$ with very high probability. The constant $\rho_h = 0.1$ for all $h \leq l+5$, then the last few constants $\rho_{l-4} = 0.104$, $\rho_{l-3} = 0.13$, $\rho_{l-2} = 0.212$, $\rho_{l-1} = 0.408$, and $\rho_{l} = 0.808$.
\end{lemma}

\begin{proof}
We start the argument at time $\hourstart_h + \frac{2}{c} + \frac{41}{m}$.
Let $A = \{a:a.\fieldopinion = +1, a.\fieldlevel = -h\}$ and  $B = \{b:b.\fieldopinion = -1, b.\fieldlevel = -h\}$. We will show that by time $\hourend_h$, large number of cancel reactions happen between the agents in $A$ and $B$ to reduce the total mass.
Inductively, by $\hourend_{h-1}$ we have total mass $\mu \leq \rho_{(h-1)} |\mains| 2^{-h+1}$. We assume in the worst case (since we need the total mass to be small) that we start with the maximum possible amount of mass $\mu = \rho_{(h-1)} |\mains| 2^{-h+1}$. We use the upper bound on the gap $g_h \leq \alpha|\mains|$, where $\alpha = 0.8\cdot 2^{h-l}$ and assume in the worst case (since larger gaps reduce the rate of cancelling reactions) the largest possible gap $\beta_+ - \beta_- = g_h2^{-h} = \alpha|\mains|2^{-h}$. This gives majority mass $\beta_+ = (\rho_{(h-1)} + \frac{\alpha}{2})|\mains|2^{-h}$ and minority mass $\beta_- = (\rho_{(h-1)} - \frac{\alpha}{2})|\mains|2^{-h}$. 

Since the minority is the limiting reactant in the cancelling reactions, we assume in the worst case the smallest possible minority count at exponent $-h$, where all mass outside of exponent $-h$ is the minority. Then the majority count at exponent $-h$ is $|A| = (\rho_{(h-1)} + \frac{\alpha}{2})|\mains|$. By \cref{lem:induction-upper-tail}, the mass above exponent $-h$ is $\mu_{(>-h)}(\hourstart_h + t_2) \leq 0.001 |\mains| 2^{-h+1}$. By \cref{lem:level-lower-tail}, the maximum count below exponent $-h$ is $0.0024|\mains|$, so the mass $\mu_{(<-h)} \leq 0.0024|\mains|2^{-h-1}$. This leaves a minority count at exponent $-h$ of $|B| = (\rho_{(h-1)} - \frac{\alpha}{2} - 0.002 - 0.0012)|\mains|$.

Now we will apply \cref{lem:cancel-reactions} with $a = (\rho_{(h-1)} + \frac{\alpha}{2})m$ and $b = (\rho_{(h-1)} - \frac{\alpha}{2} - 0.0032)m$. 
First we consider the cases where $h \leq l + 5$, where $\rho_h = \rho_{(h-1)} = 0.1$ There the bound on the gap $\alpha = 0.8\cdot 2^{(l-5)-l} = 0.025$. This gives $a = 0.1125m$ and $b = 0.0843m$. By \cref{lem:cancel-reactions}, the time $t$ to cancel a fraction $d = 0.05m$ from both $a$ and $b$ has mean
\[
\E[t] \sim \frac{\ln(b) - \ln(a) - \ln(b-d) + \ln(a - d)}{2(a-b)}
=
\frac{\ln(0.0843) - \ln(0.1125) - \ln(0.0343) + \ln(0.0625)}{2(0.0282m)}
<
\frac{5.53}{m},
\]
so $\E[t](1 + \epsilon) < \frac{6}{m}$ and $t < \frac{6}{m}$ with very high probability. Cancelling this fraction $d$ reduces the total mass by $2dn2^{-h} = 0.1|\mains|\cdot2^{-h} = \rho_{(h-1)}|\mains|2^{-h+1} - \rho_{h}|\mains|2^{-h}$. Then the total mass is at most $\rho_h|\mains|2^{-h}$ by time $\hourstart_h + \frac{2}{c} + \frac{47}{m}$.

For the remaining levels $h = l-4, l-3, l-2, l-1, l$, we will repeat the above argument and calculation, but now the bound will change so $\rho_h \neq \rho_{(h-1)}$. First our bound on the gap $\alpha$ will double as we move down each level. This causes the worst case fractions $a$ and $b$ to be spread further apart. Then $d$, the largest fraction which will cancel from both sides within $\frac{6}{m}$ time with high probability, will be smaller. This gives the new mass bound $\rho_h$, since from the equation $\rho_{(h-1)}|\mains|2^{-h+1} - \rho_{h}|\mains|2^{-h} = 2dn2^{-h}$, we have $d = (\rho_{(h-1)} - \frac{\rho_h}{2})m$ and $\rho_h = 2(\rho_{(h-1)}-\frac{d}{m})$. This relative total mass bound $\rho_h$ will be growing larger as $h$ decreases. This is to be expected, since we do in fact see the count of biased agents growing as they reach the final exponents before the count of $\unbiased$ agents runs out (see \cref{fig:exponent} and \cref{fig:opinion-constant}, where the value $-l = -19$).

The following table shows the constants used in the computations at each of the steps $h = l-5, l-4, \ldots, l$. The top row $h = l-5$ corresponds to the exact calculations above, then the following rows are the constants used in the same argument for lower levels.

\begin{tabular}{|c|c|c|c|c|c|c|}
    \hline
    $h = l - 5$ &
    $\alpha = 0.025$ &
    $a = 0.1125m$ &
    $b = 0.0843m$ &
    $d = 0.05m$ &
    $\E[t] < 5.53/m$ &
    $\rho_{(l-5)} = 0.1$ 
    \\
    \hline
    $h = l - 4$ &
    $\alpha = 0.05$ &
    $a = 0.125m$ &
    $b = 0.0718m$ &
    $d = 0.048m$ &
    $\E[t] < 5.83/m$ &
    $\rho_{(l-4)} = 0.104$ 
    \\
    \hline
    $h = l - 3$ &
    $\alpha = 0.1$ &
    $a = 0.154m$ &
    $b = 0.0508m$ &
    $d = 0.039m$ &
    $\E[t] < 5.66/m$ &
    $\rho_{(l-3)} = 0.13$ 
    \\
    \hline
    $h = l - 2$ &
    $\alpha = 0.2$ &
    $a = 0.23m$ &
    $b = 0.0268m$ &
    $d = 0.024m$ &
    $\E[t] < 5.29/m$ &
    $\rho_{(l-2)} = 0.212$ 
    \\
    \hline
    $h = l - 1$ &
    $\alpha = 0.4$ &
    $a = 0.412m$ &
    $b = 0.0088m$ &
    $d = 0.008m$ &
    $\E[t] < 2.95/m$ &
    $\rho_{(l-1)} = 0.408$ 
    \\
    \hline
    $h = l$ &
    $\alpha = 0.8$ &
    $a = 0.808m$ &
    $b = 0.0048m$ &
    $d = 0.004m$ &
    $\E[t] < 1.11/m$ &
    $\rho_{l} = 0.808$ 
    \\
    \hline
\end{tabular}

Note that for the last couple hours, the worst case for $b$ is very small, so the amount of cancelling that happens is negligible, and the relative mass bound essentially doubles. We will see later that the minority count does in fact sharply decrease, so it is accurate that the rate of cancellation drops to zero, and the count of biased majority agents is roughly doubling for these last couple rounds.
\end{proof}

\subsection{End of \phaseMainAveraging}

Now that we have proven \cref{thm:induction-on-hours}, we will finish analyzing separately the cases of an initial tie and an initial majority opinion.

In the case of a tie, the stronger bounds from \cref{thm:induction-on-hours} hold all through the final hour $h = L$. Thus at hour $L$ we still have a constant fraction $0.77|\mains|$ of $\unbiased$ agents. We now show that in the remaining $O(\log n)$ time in \phaseMainAveraging, while the $\roleclock$ agents with $\fieldhour = h$ decrement their counter in \cref{line:if-both-clocks-finished}, these $\unbiased$ agents can keep doing split reactions to bring all remaining biased agents down to exponent $-L$. This lets us now prove \cref{thm:phase-3-tie-result}, the main result of the section for the case of a tie.

\replicatedThmPhaseThreeTieResult*

\begin{proof}
We start by using \cref{thm:induction-on-hours}, where at hour $L$, the total mass $\mu(\hourend_L) \leq 0.1|\mains|2^{-L}$. Thus the count of biased agents is at most $0.1|\mains|$. We also have at least $0.77|\mains|$ $\unbiased$ agents with $\fieldTlevel = L$ (this count is actually slightly higher, as we could show almost all of the $0.9|\mains|$ $\unbiased$ agents reach $\fieldTlevel = L$ quickly by the same argument as \cref{lem:induction-t-count}, but this bound will suffice).

We use the potential $\phi_{(>-L)}$, where by \cref{lem:induction-upper-tail}, we have $\phi_{(>-L)}(\hourend_0) \leq 0.001|\mains|$. The goal is to now show that $\phi_{(>-L)} = 0$ within the $O(\log n)$ time it takes for any $\fieldcounter$ to hit $0$ and trigger the move to the next phase. This proof proceeds identically to the proof of \cref{lem:upper-levels-empty}, where we had shown that 
$\E[\phi_{(>-L)}(t+1/n)|\phi_{(>-L)}(t)] \leq (1-\frac{0.77 m}{n})\phi_{(>-L)}(t)$. We again recursively bound the expected value of $\phi_{(>-L)}$ after an additional $16\ln n$ time.
\begin{align*}
\E[\phi_{(>-L)}(\hourend_L + 16\ln n)|\phi_{(>-L)}(\hourend_0)] &\leq \qty(1-\frac{0.77 m}{n})^{16 n \ln n}\phi_{(>-L)}(\hourend_0)
\\
\E[\phi_{(>-L)}(\hourend_L + 16\ln n)]
&\leq \exp(-0.77\cdot 0.25 \cdot 16\ln n)\cdot 0.001|\mains|
\\
&\leq n^{-3.08}0.001mn
= O(1/n^{2}).
\end{align*}
Again we conclude by Markov's Inequality that $\Pr[\phi_{(>-L)}(\hourend_L + 16\ln n) > 0] \leq O(1/n^2)$. 

Finally we must argue that \phaseMainAveraging\ lasts at least until time $\hourend_L + 16\ln n$. We will bound the probability that a $\roleclock$ agent can decrement its $\fieldcounter$ down to $0$ before time $\hourend_L + 16\ln n$. Note the if statement on \cref{line:if-both-clocks-finished} will only decrement the counter when both $\roleclock$ agents have reached the final minute. Even if in the worst case an agent was at the final minute at the beginning of the phase, until time $\hourend_{L-1}$ the count of other $\roleclock$ agents with $\fieldminute = kL$ is at most $0.1p|\clocks| = 0.001|\clocks|$.
By \cref{thm:clock}, the time between consecutive hours is at most $\frac{3.86k}{c^2} < \frac{290}{c}$ using $k = \kconstant$ and the bound $c > 0.2$ with very high probability from \cref{lem:phase-initialize}.
Then the number of interactions until $\hourend_{L-1}$ is at most $\frac{290}{c}Ln$. In each interaction, the probability of the $\roleclock$ agent decrementing its counter is at most $2\cdot 0.001c \cdot \frac{1}{n}$, so the expected number of decrements $\E[X] \leq 290L\cdot 0.002 \leq 0.58(\log_2(n)) \leq 0.84\ln n $. Then by Chernoff bounds in \cref{thm:chernoff-bound-multiplicative}, with $\mu = 0.84\ln n, \delta = 5$, we have
\[
\Pr[X \geq 5.04\ln n] \leq \Pr[X \geq (1+\delta)\mu] \leq \exp(-\frac{\delta^2\mu}{2 + \delta}) = \exp(-3\ln n) = n^{-3}.
\]

Then after time $\hourend_{L-1}$, we assume that all $\roleclock$ agents have $\fieldminute = kL$. In this case, the probability a given $\roleclock$ agent decrements its counter is at most $\frac{2c}{n}$. Thus in $16n\ln n$ interactions, the expected number of decrements $\E[X] \leq 32c\ln n \leq 11\ln n$, using $c < \frac{1}{3}$ from \cref{lem:phase-initialize}. We again use \cref{thm:chernoff-bound-multiplicative}, with $\mu = 11\ln n, \delta =  0.9$, to conclude
\[
\Pr[X \geq 20.9\ln n] \leq \Pr[X \geq (1+\delta)\mu] \leq \exp(-\frac{\delta^2\mu}{2 + \delta}) \leq \exp(-3.07\ln n) \leq n^{-3}.
\]
Then by union bound, $\roleclock$ agents decrements their counter at most $5.04\ln n + 20.9\ln n < 26\ln n$ times with high probability $1 - O(1/n^2)$. In \phaseMainAveraging, we initialize $\fieldcounter \gets c_3\ln n$, so setting this counter constant $c_3 = 26$, we conclude that no agent moves to the next phase before time $\hourend_L + 16\ln n$ with high probability. This gives enough time to conclude that all biased agents have $\fieldlevel = -L$ by the end of \phaseMainAveraging.
\end{proof}

Now we consider the case where there is an initial majority, which we have assumed without loss of generality was $\A$. Here the argument of \cref{thm:induction-on-hours} stops working at exponent $-l$, because of the substantial relative gap $0.4|\mains| \leq g_l < 0.8|\mains|$. Next we will need to show that the count of $\unbiased$ gets brought to almost $0$ by split reactions with the majority. In order to show that the count of $\unbiased$ stays small however, we will need some way to bound future cancel reactions. To do this, we will now show that during the last few hours $l-4, \ldots, l$, the minority mass $\beta_-$ drops to a negligible amount. We will then use this tight bound to show the majority consumes most remaining $\unbiased$ agents, and finally that constraints on the majority mass $\beta_+$ will force the majority count to remain above $99\%$ at exponents 
$-l,-(l+1),-(l+2)$ 
for all reachable configurations in the rest of \phaseMainAveraging.

We first show inductively that the minority mass becomes negligible during the last few hours $l-5, \ldots, l$. We already showed a bound on the total mass in \cref{lem:induction-total-mass}, which used an upper bound on the gap. If we want the minority mass to be small, now the worst case is the smallest possible gap. Notice that just using the final mass bound $\mu < \rho_l|\mains|2^{-l} = 0.808|\mains|2^{-l}$ in the worst case with the smallest possible gap $g_l = 0.4$ at exponent $-l$, would imply $\beta^+ \approx 0.6|\mains|2^{-l}$ and $\beta^- \approx 0.2|\mains|2^{-l}$. To get a much tighter upper bound on minority mass, we will make a similar inductive argument to \cref{lem:induction-total-mass}, now using the lower bound for each gap $g_h$, to show enough minority mass cancels at exponent $-h$ during each hour $h$.

\begin{lemma}
\label{lem:induction-minority-mass}
For each hour $h = l-5,\ldots, l$, by time $\hourend_h$, the minority mass $\beta_- \leq \xi_h|\mains|2^{-h}$ with very high probability. The constants $\xi_{(l-5)} = 0.04375$, $\xi_{(l-4)} = 0.0375$, $\xi_{(l-3)} = 0.0267$, $\xi_{(l-2)} = 0.0145$, $\xi_{(l-1)} = 0.0056$, and $\xi_{l} = 0.004$.
\end{lemma}

\begin{proof}
For the base case $h = l-5$, by \cref{lem:induction-total-mass}, the total mass $\mu(\hourend_{(l-5)}) \leq 0.1|\mains|2^{-l+5}$. The relative gap $0.0125|\mains| \leq g_{(l-5)} < 0.025|\mains|$, so now in the worst case (for minority mass being small), we assume the smallest gap $g_{(l-5)} = 0.0125|\mains|$ and largest total mass $\mu(\hourend_{(l-5)}) = \beta_+ + \beta_- =  0.1|\mains|2^{-l+5}$. Since the invariant gap $g = \beta_+ - \beta_- = g_{(l-5)}2^{-l+5}$, we have $\beta_+ = 0.05625|\mains|2^{-l+5}$ and $\beta_- = 0.04375|\mains|2^{-l+5}$, giving an upper bound constant $\xi_{(l-5)} = 0.04375$.

Now we outline the inductive step, for each level $h = l-4, \ldots, l$. We will show the constants for the first step $h = l-4$ here, and then list constants for the remaining steps in a table below. Let $A = \{a:a.\fieldopinion = +1, a.\fieldlevel = -h\}$ and  $B = \{b:b.\fieldopinion = -1, b.\fieldlevel = -h\}$.

We start the argument at time $\hourstart_h + \frac{2}{c} + \frac{41}{m}$, and with the previous bound $\beta_-(\hourend_{(h-1)}) = \xi_{(h-1)}|\mains|2^{-h+1} = 0.0875|\mains|2^{-h}$. We will then assume in the worst case the smallest possible gap $g_h = \alpha|\mains| = 0.4\cdot 2^{h-l}|\mains| = 0.025|\mains|$. This gives $\beta_+ = (2\xi_{(h-1)} + \alpha)|\mains|2^{-h}$. We then assume in the worst case that all mass allowed outside exponent $-h$ belongs to the minority (so doesn't reduce by cancelling), giving majority count $|A| = (2\xi_{(h-1)} + \alpha)|\mains|$. We also assume the maximum amount of mass outside exponent $-h$. By \cref{lem:level-lower-tail}, before time $\hourend_h$, the count below exponent $-h$ is at most $0.0024|\mains|$, so the mass is at most $0.0024|\mains|2^{-h-1}$. By \cref{lem:induction-upper-tail}, after time $\hourstart_h + \frac{2}{c} + \frac{41}{m}$, the mass above exponent $-h$ is $\mu_{(>-h)} \leq 0.001\cdot 2^{-h+1}$. This leftover remaining mass at exponent $-h$ gives minority count $|B| = (2\xi_{(h-1)} - 0.0012 - 0.002)|\mains|$.

Now we will apply \cref{lem:cancel-reactions} with fractions $a = (2\xi_{(h-1)} + \alpha)m = 0.1125m$ and $b = (2\xi_{(h-1)} - 0.0036)m = 0.0843m$. These match the first fractions $a,b$ used in the proof of \cref{lem:induction-total-mass}, and again we use that the time $t$ to cancel a fraction $d = 0.05m$ from both $a$ and $b$ has mean
\[
\E[t] \sim \frac{\ln(b) - \ln(a) - \ln(b-d) + \ln(a - d)}{2(a-b)}
=
\frac{\ln(0.0843) - \ln(0.1125) - \ln(0.0343) + \ln(0.0625)}{2(0.0282m)}
<
\frac{5.53}{m},
\]
so $\E[t](1 + \epsilon) < \frac{6}{m}$ and $t < \frac{6}{m}$ with very high probability. Cancelling this fraction $d$ reduces the minority mass by $dn2^{-h} = 0.05|\mains|\cdot2^{-h}$, giving a new minority mass $\beta_- = (2\xi_{(h-1)} - \frac{d}{m})|\mains|2^{-h} = 0.0375|\mains|2^{-h}$. Thus for $\xi_{h} = 0.0375$, the minority mass is at most $\xi_h|\mains|2^{-h}$ by time $\hourstart_h + \frac{2}{c} + \frac{47}{m}$.

For the remaining levels $h = l-4, l-3, l-2, l-1, l$, we will repeat the above argument and calculation. 
The following table shows the constants used in the computations at each of the steps $h = l-4, l-3, \ldots, l$. The top row $h = l-4$ corresponds to the exact calculations above, then the following rows are the constants used in the same argument for lower levels.

\begin{tabular}{|c|c|c|c|c|c|c|}
    \hline
    $h = l - 4$ &
    $\alpha = 0.025$ &
    $a = 0.1125m$ &
    $b = 0.0843m$ &
    $d = 0.05m$ &
    $\E[t] < 5.53/m$ &
    $\xi_{(l-4)} = 0.0375$ 
    \\
    \hline
    $h = l - 3$ &
    $\alpha = 0.05$ &
    $a = 0.125m$ &
    $b = 0.0718m$ &
    $d = 0.048m$ &
    $\E[t] < 5.83/m$ &
    $\xi_{(l-3)} = 0.0267$ 
    \\
    \hline
    $h = l - 2$ &
    $\alpha = 0.1$ &
    $a = 0.154m$ &
    $b = 0.0508m$ &
    $d = 0.039m$ &
    $\E[t] < 5.66/m$ &
    $\xi_{(l-2)} = 0.0145$ 
    \\
    \hline
    $h = l - 1$ &
    $\alpha = 0.2$ &
    $a = 0.23m$ &
    $b = 0.0268m$ &
    $d = 0.024m$ &
    $\E[t] < 5.29/m$ &
    $\xi_{(l-1)} = 0.0056$ 
    \\
    \hline
    $h = l$ &
    $\alpha = 0.4$ &
    $a = 0.412m$ &
    $b = 0.0088m$ &
    $d = 0.008m$ &
    $\E[t] < 2.95/m$ &
    $\xi_{l} = 0.004$ 
    \\
    \hline
\end{tabular}

 Note that the structure of the proof yields the same sequence of values $a, b$ as the proof of \cref{lem:induction-total-mass}, so we have repeated the exact same applications of \cref{lem:cancel-reactions}. This time, however, we are using the cancelled fraction $d$ to show the minority mass becomes very small.
\end{proof}

Now \cref{lem:induction-minority-mass} gives that minority mass $\beta_- \leq 0.004|\mains|2^{-l}$ by time $\hourend_l$. The rest of the argument will not try to reason about probabilistic guarantees on what the distribution looks like. Instead, we will make claims purely about reachability, and show that we get to a configuration where invariants force the count of majority agents to remain high in any reachable configuration using the transitions of \phaseMainAveraging.

\begin{lemma}
\label{lem:large-majority-count}
By the time $\hourend_{(l+2)}$, the count of majority agents with $\fieldlevel \in \{-l,-(l+1),-(l+2)\}$ is at least $0.96|\mains|$, with very high probability. Then in all reachable configurations where all agents are still in $\phaseMainAveraging$, the count of majority agents with $\fieldlevel \in \{-l,-(l+1),-(l+2)\}$ is at least $0.92|\mains|$.
\end{lemma}

\begin{proof}
We argue by process of elimination. First we apply \cref{lem:level-lower-tail} to conclude that until time $\hourend_{(l+2)}$, the count $|\{\unbiased:\unbiased.\fieldTlevel > l+2\} \cup \{b:b.\fieldlevel < -(l+2)\}| \leq 0.0024|\mains|$ with very high probability.

Now \cref{lem:induction-minority-mass} gives that minority mass $\beta_- \leq 0.004|\mains|2^{-l}$ after time $\hourend_l$. Then we can bound the maximum minority count at exponent $-(l+2)$ and above by $0.016|\mains|$. In addition, these minority agents could eliminate more majority agents by cancel reactions. The count of agents they could cancel with is at most $0.016|\mains|$.

This leaves at least $|\mains| - 0.0024|\mains| - 0.016|\mains - |0.016|\mains| = 0.9656|\mains|$ agents that are either majority or $\unbiased$ agents at time $\hourend_l$. We will next show that in hour $l+2$, most of any remaining $\unbiased$ agents set $\fieldhour = l+2$ and then are consumed in split reactions with majority agents with $\fieldlevel > -(l+2)$.

The majority mass $\beta_+ \geq g = g_l 2^{-l} \geq 0.4|\mains|2^{-l}$. By \cref{lem:induction-upper-tail}, at most $0.001|\mains|2^{-l+1} = 0.002|\mains|2^{-l}$ of $\beta_+$ can come from exponents above $-l$. Thus we need mass $0.398|\mains|2^{-l}$ from majority agents at exponent $-l$ and below. Note that if the count $|\{a:a.\fieldopinion = +1, a.\fieldlevel \in \{-l,-(l+1)\}\}|\leq 0.19|\mains|$, then the maximum majority mass could achieve would be $0.81|\mains|$ with $\fieldlevel = -l-2$ and $0.19|\mains|$ with $\fieldlevel = -l$, which gives $(\frac{0.81}{4}+0.19)|\mains|2^{-l} < 0.398|\mains|2^{-l}$.
This implies we must have a count of at least $0.19|\mains|$ of majority agents at exponents $-(l+1)$ and $-l$ in all reachable configurations.

Next we show that all but at most $0.0036|\mains|$ of these $\unbiased$ agents gets brought up to $\fieldhour = l+2$. We start at time $\hourstart_{l+2}$, when the count $|\unbiased| \leq (0.9656 - 0.398)\mains = 0.5676|\mains|$, and the count of clock agents at $\fieldhour = l+2$ is at least $0.9|\clocks|$. Thus we can apply \cref{lem:one-sided-cancel}, with $a = 0.9c$, $b_1 = 0.5676m$, $b_2 = 0.0056m$ to conclude that at most $0.0056|\mains|$ of $\unbiased$ agents are left at the wrong hour after parallel time $t_1 < (1+\epsilon)\E[t_1]$, where the expected time
\[
\E[t_1] = \frac{\ln(b_1) - \ln(b_2)}{2a}
= \frac{\ln(0.5676m) - \ln(0.0036m)}{1.8c}
\leq \frac{2.82}{c}.
\]
Thus with very high probability $t_1 < \frac{3}{c}$.

Next the at least $0.19|\mains|$ majority agents at exponents $-(l+1)$ and $-l$ will eliminate these $\unbiased$ agents by split reactions. We show that at most $0.001|\mains|$ of $\unbiased_h$ agents are left. We again apply \cref{lem:one-sided-cancel}, with $a = 0.19m$, $b_1 = 0.5676m$, $b_2 = 0.01m$ to conclude this takes at most parallel time $t_2 < (1 + \epsilon)\E[t_2]$, where
\[
\E[t_2] = \frac{\ln(b_1) - \ln(b_2)}{2a}
= \frac{\ln(0.5676m) - \ln(0.001m)}{0.38m}
\leq \frac{16.69}{m}.
\]
Thus with very high probability $t_2 < \frac{17}{m}$.

Account for these counts $0.0032|\mains|$ and $0.001|\mains|$ of leftover $\unbiased$ agents, along with the maximal count $0.001|\mains|$ of biased agents above exponent $-l$, we are left with $0.9656|\mains| - 0.0032|\mains| - 0.001|\mains| - 0.001|\mains| = 0.96|\mains|$ agents that must be majority agents with $\fieldlevel \in \{-l,-(l+1),-(l+2)\}$. 

Now we argue that no reachable configurations can bring this count below $0.92|\mains|$. We have already accounted for the maximum number of minority agents that can do cancel reactions at exponents $-l,-(l+1),-(l+2)$. Thus we only have to argue that no amount of split reactions can bring this count down by $0.04|\mains|$. Note that reducing the count by $0.04|\mains|$ will reduce the mass at these exponents by at least $0.04|\mains|2^{-(l+2)}$. It will take at least $0.08|\mains|$ agents below exponent $-l-2$ to account for the same mass, for a total of $0.96|\mains| - 0.04|\mains| + 0.08|\mains| = |\mains|$. Thus it is not possible to move more than $0.04|\mains|$ majority agents below exponent $-(l+2)$ by split reactions, because there are not enough $\rolemain$ agents to account for the mass.
\end{proof}

We can now use the results of \cref{lem:large-majority-count}, \cref{lem:induction-upper-tail}, and \cref{lem:induction-minority-mass}, including the high probability requirements from all previous lemmas, to prove \cref{thm:phase-3-majority-result}, the main result for \phaseMainAveraging\ in the non-tie case:

\replicatedThmPhaseThreeMajorityResult*

\section{Analysis of final phases}
\label{sec:analysis-final-phases}

\subsection{$\rolereserve$ agent Phases \ref{phase:reserve-sample} and \ref{phase:reserve-split}}

We now consider the behavior of the $\rolereserve$ agents in \phaseReserveSample\ and \phaseReserveSplit. We first show that with high probability they are also able to set their $\fieldsample$ field during \phaseReserveSample.

\begin{lemma}
\label{lem:reserve-sample}
By the end of \phaseReserveSample, all $\rolereserve$ agents have $\fieldsample \neq \bot$, with high probability $1-O(1/n^2)$.
\end{lemma}

\begin{proof}
By \cref{thm:phase-3-majority-result}, we have at least $0.92|\mains|$ majority agents with $\fieldexponent \in \{-l,-(l+1),-(l+2)\}$ by the end of \phaseMainAveraging. Now we can analyze the process of $\rolereserve$ agents sampling the $\fieldexponent$ of the first biased agent they encounter in \phaseReserveSample\ with \cref{lem:one-sided-cancel}, where subpopulations $A = \{a:a.\fieldrole = \rolemain, a.\fieldbias = \pm 1\}$ with and $B = \reserves$, comprising fractions $a \geq 0.92m$ and $b_1 = r$. Then by \cref{lem:one-sided-cancel}, all $\rolereserve$ agents set their $\fieldsample$ within time $t$, where $t \leq \frac{5\ln n}{2\cdot 0.92m}$ with high probability $1-O(1/n^2)$. Thus for appropriate choice of counter constant $c_5$, this will happen before the first $\roleclock$ agent advances to the next phase.
\end{proof}

We next show that the $\rolereserve$ agents are able to bring the exponents of all biased agents down to at most $-l$ during \phaseReserveSplit.

\begin{lemma}
\label{lem:reserve-split}
By the end of \phaseReserveSplit, all biased agents have $\fieldlevel \leq -l$, with high probability $1 - O(1/n^2)$.
\end{lemma}

\begin{proof}
\cref{thm:phase-3-majority-result}, the mass above exponent $-l$ is $\mu_{(>-l)} \leq 0.001|\mains|2^{-l+1}$. We must show that all of this mass moves down to at least exponent $-l$ via split reactions with the $\rolereserve$ agents (\cref{line:reserve-split} in \phaseReserveSplit).
We consider the following two cases based on whether the size of $A_{-l} = \{a:a.\fieldopinion = +1, a.\fieldexponent = -l\}$.

In the first case, $|A| > 0.19|\mains|$ and we will show all agents get brought down to $\fieldexponent \leq -l$. Because of the sampling process in \phaseReserveSample, the count of $R_{-l} = \{r:r.\fieldrole = \rolereserve, r.\fieldsample = -l\}$ has size at least $|R_{-l}| \geq 0.18|\reserves|$  with very high probability, by standard Chernoff bounds. If all the agents above exponent $-l$ split down to exponent $-l$, they would have count at most $0.002|\mains|$, potentially all coming out of the initial count of $R_{-l}$. Thus for the entirety of \phaseReserveSplit, we have 
\[
|R_{-l}| \geq 0.18|\reserves| - 0.002|\mains| = n(0.18 r - 0.002m) \geq n(0.18 \cdot 0.24 - 0.002) \geq 0.04n,
\]
using the very high probability bound $r > 0.24$ from \cref{lem:phase-initialize}.

In the second case $|A| \leq 0.19|\mains|$. Now we cannot make guarantees on the size $|R_{-l}|$, so we will instead reason about $A_{-(l+1)} = \{a:a.\fieldopinion = +1, a.\fieldexponent = -(l+1)\}$ and $R_{-(l+1)} = \{r:r.\fieldrole = \rolereserve, r.\fieldsample = -(l+1)\}$. 

We first show that $|A_{-(l+1)}| > 0.59|\mains|$.
Recall by definition of $l$ and $g_l$ that the majority mass $\beta_+ \geq g_l 2^{-l} \geq 0.4|\mains|2^{-l}$. 
At most $0.001|\mains|2^{-l+1}$ can come from exponents above $-l$ and at most $0.19|\mains|2^{-l}$ comes from exponent $-l$. Thus there must be at least mass $(0.398 - 0.19)|\mains|2^{-l}$ from $A_{-(l+1)}$ and the exponents below. If $|A_{-(l+1)}| = 0.59|\mains|$, then even putting all $0.41|\mains|$ remaining $\rolemain$ agents at $\fieldexponent = -(l+2)$ would only give mass $0.59|\mains|2^{-(l+1)} + 0.41|\mains|2^{-(l+2)} = 0.3975|\mains|2^{-l}$. Thus we require $|A_{-(l+1)}| > 0.59|\mains|$. Again by standard Chernoff bounds from the sampling process of \phaseReserveSample, this implies the initial size of $|R_{-(l+1)}| \geq 0.58|\reserves|$. 

If all the agents above exponent $-l$ split down to exponent $-(l+1)$, they would have count at most $0.004|\mains|$, potentially all coming out of the initial count of $R_{-l}$.
In addition, we can lose count $|A| \leq 0.19|\mains|$ from $R_{-l}$ from additional split reactions. This implies that for the entirety of \phaseReserveSplit, we have count at least
\begin{align*}
|R_{-(l+1)}| &\geq 0.58|\reserves| - 0.004|\mains| - 0.19|\mains| = n(0.58 r - 0.004m - 0.19 m) 
\\
& \geq n(0.58\cdot 0.24 - (0.004 + 0.19) \cdot 0.51)
\geq 0.04n,
\end{align*}
using the very high probability bounds $r > 0.24$ and $m < 0.51$ from \cref{lem:phase-initialize}.

Thus in both cases we maintain a count of at least $0.04n$ $\rolereserve$ agents at level $-l$ or $-(l+1)$. For an agent $a$ with $a.\fieldexponent > -l$, the probability that in a given interaction agent $a$ does a split reaction (\cref{line:reserve-split} in \phaseReserveSplit) with an agent in $R_{-l} \cup R_{-(l+1)}$ is at least $\frac{2\cdot 0.04n \cdot 1}{n^2} = \frac{0.08}{n}$. Now we proceed as in the proofs of \cref{lem:upper-levels-empty} and \cref{thm:phase-3-tie-result}, analyzing the potential $\phi_{(>-l)}$ to show it hits 0 in $O(\log n)$ time.

Recall that for each biased agent $a$ with $a.\fieldexponent = -i > -l$ and local potential $\phi_a = 4^{-i+l-1}$, a split reaction changes the potential by $\Delta \phi_a \leq 4^{l-1}(2\cdot 4^{-i-1} - 4^{-i}) = 4^{l-1}(-\frac{1}{2}4^{-i}) = -\frac{1}{2}\phi_a$. Then in each interaction at parallel time $t$, the expected change in the potential
\begin{align*}
\E[\phi_{(>-l)}(t+1/n) - \phi_{(>-l)}(t)] 
& \leq \sum_{a.\fieldlevel > -l}\Pr[\text{$a$ splits}]\cdot \Delta\phi_{a}
\\&\leq \sum_{a.\fieldlevel > -l} \frac{0.08}{n} \cdot -\frac{1}{2}\phi_{a}
= -\frac{0.04}{n}\phi_{(>-l)}(t).
\end{align*}
Thus we have $\E[\phi_{(>-l)}(t+1/n)|\phi_{(>-l)}(t)] \leq (1-\frac{0.04}{n})\phi_{(>-l)}(t)$. 

When we begin the argument at time $t_{\text{start}}$ at the beginning of \phaseReserveSplit, we start with $\phi_{(>-l)}(t_{\text{start}}) \leq 0.001|\mains|$ from the final iteration of \cref{lem:induction-upper-tail}. We will wait until time $t_{\text{end}} = t_{\text{start}} + 75\ln n$, and now recursively bound the potential after these $75 n \ln n$ interactions:

\begin{align*}
\E\qty[\phi_{(>-l)}(t_{\text{end}}) | \phi_{(>-l)}(t_{\text{start}})] 
&\leq \qty( 1-\frac{0.04}{n} )^{75 n\ln n} \phi_{(>-l)}(t_{\text{start}})
\\
\E\qty[\phi_{(>-h+q)}(t_{\text{end}})]
&\leq 
\exp(-0.04 \cdot 75 \ln n)\cdot 
0.001|\mains|
\\
&\leq n^{-3}
\cdot 
0.001mn
= O\qty(n^{-2})
\end{align*}
Finally, since $\phi_{(>-l)}$ takes nonnegative integer values, we can apply Markov's Inequality to conclude $\Pr[\phi_{(>-l)}(t_{\text{end}}) > 0] = O(1/n^{2})$.
So after $75\ln n$ time in \phaseReserveSplit, all biased agents have $\fieldexponent \leq -l$, with high probability. For appropriate choice of counter constant $c_6$, this happens before any $\roleclock$ agent advances to the next phase.
\end{proof}

Recall $M$ is the set of all majority agents with $\fieldlevel \in \{-l,-(l+1),-(l+2)\}$, where $|M| \geq 0.92|\mains|$ at the end of \phaseMainAveraging\ by  \cref{thm:phase-3-majority-result}. We finally observe that after \phaseReserveSplit, $M$ is still large.

\begin{lemma}
\label{lem:reserve-keeps-majority-count}
At the end of \phaseReserveSplit, $|M| \geq 0.87|\mains|$ with very high probability.
\end{lemma}

\begin{proof}
By \cref{thm:phase-3-majority-result}, $|M| \geq 0.92|\mains|$ at the end of \phaseMainAveraging, so the count of all other main agents is at most $0.08|\mains|$. By standard Chernoff bounds on the sampling process in \phaseReserveSample, the count 
\[
|\{r:r.\fieldrole = \rolereserve, r.\fieldsample \notin \{-l,-(l+1),-(l+2)\}\}| \leq 0.09|\reserves| \leq 0.05|\mains|.
\]
Split reactions that consume these $\rolereserve$ agents are the only way to bring agents out of the set $M$.
Thus at the end of \phaseReserveSplit, the count is still at least
\[
|M| \geq 0.92|\mains| - 0.09|\reserves| \geq 0.87|\mains|,
\]
using $|\reserves| < \frac{5}{9}|\mains|$ with very high probability by \cref{lem:phase-initialize}. 
\end{proof}

\subsection{Minority elimination Phases \ref{phase:highminorityelimination} and \ref{phase:lowminorityelimination}}

Next we argue that in \phaseHighMinorityElimination, the agents in $M$ are able to eliminate all minority agents with $\fieldexponent \in \{-l,-(l+1),-(l+2)\}$. Note that these minority agents are able to do a cancel reaction with any agent in $M$, since by design \phaseHighMinorityElimination\ allows reactions between agents with an exponent gap of at most 2. We first argue that $|M|$ must stay large through the entirety of \phaseHighMinorityElimination:

\begin{lemma}
\label{lem:high-minority-elimination-majority-count}
At the end of \phaseHighMinorityElimination, $|M| \geq 0.8|\mains|$.
\end{lemma}

\begin{proof}
We use the bound on minority mass $\beta_- \leq 0.004|\mains|2^{-l}$ from \cref{thm:phase-3-majority-result}. This implies the minority count at exponent $-(l+4)$ and above is at most $0.064|\mains|$, since a $0.064|\mains|\cdot 2^{-(l+4)} \geq \beta$. Note that cancelling with these minority agents is the only way for a majority agent in $M$ to change its state in \phaseHighMinorityElimination. Using the previous bound from \cref{lem:reserve-keeps-majority-count}, the count of $M$ can decrease at most to $|M| = 0.87|\mains| - 0.064|\mains| \geq 0.8|\mains|$.
\end{proof}

Now we argue that these agents in $M$ eliminate all high exponent minority agents.

\begin{lemma}
\label{lem:high-minority-elimination}
At the end of \phaseHighMinorityElimination, all minority agents have $\fieldexponent < -(l+2)$, with high probability $1 - O(1/n^2)$.
\end{lemma}

\begin{proof}
From \cref{lem:reserve-split}, all minority agents already have $\fieldexponent \leq -l$. Let $B_{-l}, B_{-(l+1)}, B_{-(l+2)}$ be the sets of minority agents with $\fieldexponent = -l,-(l+1),-(l+2)$, respectively. We argue successively that $|B_{-l}| = 0$, then $|B_{-(l+1)}|$, then $|B_{-(l+2)}| = 0$.

The initial bound on minority mass $\beta_- \leq 0.004|\mains|2^{-l}$ from \cref{thm:phase-3-majority-result} implies that initially $|B_{-l}| \leq 0.004|\mains|$. Note that an interaction with any agent in $M$ will bring remove an agent from $B_{-l}$ via one of the \phaseHighMinorityElimination\ cancel reactions. Using the bound $|M| \geq 0.8|\mains|$ during the entire phase from \cref{lem:high-minority-elimination-majority-count}, we can use \cref{lem:one-sided-cancel} to bound the time for $|B_{-l}| = 0$. We have constants $a = 0.8m$ and $b_1 = 0.004m$. Then by \cref{lem:high-minority-elimination-majority-count}, with high probability $1-O(1/n^2)$, the time $t_1$ to eliminate the count of $B_{-l}$ is at most $t_1 \leq \frac{5\ln n}{2(0.8m - 0.004m)} \leq 6.41\ln n$, using the bound $m \geq 0.49$ from \cref{lem:phase-initialize}.

Next we wait to eliminate the count of $B_{-(l+1)}$, by a similar argument. Initially the count is at most $|B_{-(l+2)}| \leq 0.008|\mains|$, and this will all cancel in at most time $t_2$. By \cref{lem:high-minority-elimination-majority-count}, with high probability, $t_2 \leq \frac{5\ln n}{2(0.8m - 0.008m)} \leq 6.45\ln n$. The same argument for $B_{-(l+2)}$, initially $|B_{-(l+2)}| \leq 0.016|\mains|$, gives $t_3 \leq \frac{5\ln n}{2(0.8m - 0.016m)} \leq 6.51\ln n$.

Thus the entire process requires time at most $t_1 + t_2 + t_3 < 20\ln n$. So for appropriate choice of counter constant $c_7$, this will happen before the first $\roleclock$ agent advances to the next phase.
\end{proof}

Now we will prove that the remaining $0.8|\mains|$ majority agents in $M$ are able to eliminate all remaining minority agents in \phaseLowMinorityElimination.

\begin{lemma}
\label{lem:low-minority-elimination}
At the end of \phaseLowMinorityElimination, there are no more minority agents, with high probability $1 - O(1/n^2)$.
\end{lemma}

\begin{proof}
From \cref{lem:high-minority-elimination-majority-count} and \cref{lem:high-minority-elimination}, by the end of \phaseHighMinorityElimination, we have $|M| \geq 0.8|\mains|$ and all minority agents have $\fieldexponent < -(l+2)$. We assume in the worst case all $0.2|\mains|$ other $\rolemain$ agents have the minority opinion. Note by the consumption reaction in \phaseLowMinorityElimination, every minority agent will be eliminated by an agent in $M$ that still has $\fieldfull = \False$. Thus we can apply \cref{lem:cancel-reactions}, with $a = 0.8m$ and $b = 0.2m$, to conclude that all remaining minority agents are eliminated in time $t$, where $t \leq \frac{5\ln n}{2(0.8m - 0.2m)} \leq 8.5\ln n$, using $m \geq 0.495$ with very high probability from \cref{lem:phase-initialize}. So for appropriate choice of counter constant $c_8$, this will happen before the first $\roleclock$ agent advances to the next phase.
\end{proof}

\subsection{Fast Stabilization}

We next give a time bound for the stable backup:

\begin{lemma}
\label{lem:stable-backup}
The 6-state protocol in 
\phaseStableBackup\
stably computes majority in $O(n\log n)$ stabilization time, both in expectation and with high probability.
\end{lemma}

\begin{proof}
Note that agents in the initial state with $\fieldactive = \True$ and $\fieldoutput \in \{\A,\B\}$ can only change their state by a cancel reaction in \cref{line:stable-cancel-reaction} with another active agent with the opposite opinion.

First we consider the case where one opinion, without loss of generality $\A$, is the majority. We first wait for all active $\B$ agents to meet an active $\A$ agent and changed their state to active $\T$ via \cref{line:stable-cancel-reaction}. This is modelled by the ``cancel reaction'' process described in \cref{lem:cancel-reactions}, taking expected $O(n)$ time and $O(n\log n)$ time with high probability.
At this point, there are no active $\B$ agents and at least one active $\A$ agent remaining. We next wait for this active $\A$ agent to interact with all remaining active $\T$ agents, which will then become passive via \cref{line:stable-convert-unbiased}. This takes $O(n\log n)$ time in expectation and with high probability by a standard coupon-collector argument, after which point there are only active $\A$ agents and passive agents. Finally, we wait for these active $\A$ agents to convert all passive agents to $\fieldoutput = \A$ via \cref{line:stable-convert-passive}, taking another $O(n\log n)$ time by the same coupon-collector argument.

Next we consider the case where the input distribution is a tie. We first wait for all $n/2$ pairs of active $\A$ and $\B$ agents to cancel via \cref{line:stable-cancel-reaction}, taking $O(n)$ time in expectation and $O(n\log n)$ time with probability.
Consider the last such interaction. At this point, those two agents have become active $\T$ agents, and there are no active $\A$ or $\B$ agents left. Thus after interacting with one of the $\T$ agents, all passive agents will output $\T$ via \cref{line:stable-convert-passive}. This takes $O(n\log n)$ time in expectation and with high probability by the same coupon-collector argument.
\end{proof}

We are now able to combine all the results from previous sections to prove \cref{thm:nonuniform-algorithm}, giving guarantees on the behavior of the protocol:

\replicatedThmNonuniformAlgorithm*

\begin{proof}
We first argue that with high probability $1 - O(1/n^{2})$, the protocol \majorityNonuniform\ stabilizes to the correct output in $O(\log n)$ time. We consider three cases based on the size of the initial gap $|g|$. For the majority cases where $|g| > 0$, we assume without loss of generality $g > 0$, so $\A$ is the majority.

If $|g| \geq 0.025|\mains|$, then by \cref{lem:discrete-averaging}, we stabilize in \phaseConsensus\ to the correct output with high probability $1-O(1/n^2)$.

If $0 < |g| < 0.025|\mains|$, then by \cref{thm:phase-3-majority-result}, with high probability $1-O(1/n^{2})$, we end \phaseMainAveraging\ with at least 92\% of $\rolemain$ agents in the set $M$, with the majority opinion and $\fieldexponent \in \{-l,-(l+1),-(l+2)\}$. Then after \phaseReserveSample\ and \phaseReserveSplit, all biased agents have $\fieldexponent \leq -l$ with high probability $1 - O(1/n^2)$ by \cref{lem:reserve-split}. Then after \phaseHighMinorityElimination\ and \phaseHighMinorityElimination, there are no more minority agents with high probability $1-O(1/n^2)$ by \cref{lem:low-minority-elimination}. Thus in \phaseConsensusTwo, the majority opinion $+1$ will spread to all agents by epidemic, and we reach a stable correct configuration where all agents have $\fieldbiases = \{0, +1\}$ and $\fieldoutput = \A$.

If $g = 0$, then by \cref{thm:phase-3-tie-result}, with high probability $1 - O(1/n^{2})$, we end \phaseMainAveraging\ with all biased agents at $\fieldexponent = -L$. Thus in \phaseDetectTie, we reach a stable correct configuration where all agents have $\fieldoutput = \T$, and there are no agents with $\fieldexponent > -L$ to increment the phase.

Next we justify that \majorityNonuniform\ stably computes majority, since it is always possible to reach a stable correct configuration. By \cref{lem:phase-initialize}, we must produce at least 2 $\roleclock$ agents by the end of \phaseInitialize. Thus we must eventually advance through all timed phases 
\ref{phase:initialize}, \ref{phase:discrete-averaging}, \ref{phase:main-averaging}, \ref{phase:reserve-sample}, \ref{phase:reserve-split}, \ref{phase:highminorityelimination}, \ref{phase:lowminorityelimination}
using the $\fieldcounter$ field.
Also all agents have left the initial role $\roleMCR$, otherwise the $\textbf{Init}$ of \phaseDiscreteAveraging\ will trigger all agents to move to \phaseStableBackup\ by epidemic. Thus the sum of $\fieldbias$ across all $\rolemain$ agents is the initial gap $g$, and this key invariant is maintained through all phases. 

Using this invariant, if the agents stabilize in \phaseConsensus, they have a consensus on their $\fieldoutput$, the sign of their $\fieldbias$, and this consensus must be the sign of the sum of the biases, giving the sign of the initial gap. If the agents stabilize in \phaseDetectTie, all $|\fieldbias| \leq \frac{1}{2^L} \leq \frac{1}{n}$. This implies the gap $|g| \leq |\mains|\cdot \frac{1}{n} < 1$, so the gap $g = 0$ and the agents are correctly outputting $\T$. Finally, note that in \phaseLowMinorityElimination, when the agents set the field $\fieldfull = \True$, they can no longer actually store in memory the true bias they are holding. For example an agent with $\fieldopinion = +1, \fieldexponent = -i, \fieldfull = \True$, is actually representing the interval $\frac{1}{2^{i+1}} \leq \fieldbias < \frac{1}{2^i}$. Then we still have the guarantee that if we reach stable consensus in \phaseConsensusTwo, then this consensus is the sign of the sum of the biases and is thus correct. The final case is that we do not stabilize here, and then move to \phaseStableBackup, where they agents eventually stabilize to the correct output.

We finally justify that the expected stabilization time is $O(n\log n)$. By \cref{lem:stable-backup}, the stable backup \phaseStableBackup\ will stabilize in expected $O(n\log n)$ time. Note that the high probability guarantees are all at least $1- O(1/n^{2})$, so the time for the stable backup contributes at most $O\qty(\frac{n\log n}{n^{2}}) = o(1)$ to the total expected time. Now by \cref{lem:phase-initialize}, with very high probability we have at least $|\clocks| \geq 0.24n$ $\roleclock$ agents. In this case, the time upper bounds of \cref{thm:clock} and upper bounds on the $\fieldcounter$ times using standard Chernoff bound, imply that every timed phase lasts expected $O(\log n)$ time. If we do not stabilize in the untimed phases, then we also pass through by epidemic expected $O(\log n)$ time. Thus we either stabilize or reach the backup \phaseStableBackup\ in expected $O(\log n)$ time. There is a final very low probability case that $|\roleclock| = o(n)$, but we must at least have $|\roleclock| = 2$. Even in this worst case, the time upper bounds of \cref{thm:clock} and all counter rounds are at most polynomial in $n$, whereas the low probability of such a small $\roleclock$ population is smaller than any polynomial. Thus this event adds a negligible contribution, and we conclude the total expected stabilization time is $O(\log n)$.
\end{proof}
\section{Uniform, stable protocols for majority using more states}
\label{sec:uniform}

The algorithm described in \cref{sec:algorithm} is \emph{nonuniform}:
the set of transitions used for a population of size $n$ depends on the value $\ceil{\log n}$.
A uniform protocol~\cite{CMNPS11, doty2018exact, doty2019efficient} is a single set of transitions
that can be used in any population size.
Since it is ``uniformly specified'', the transition function is formally defined by a linear-space\footnote{
    That is, the Turing machine is allowed to use a bit more than the space necessary simply to read and write the input and output states, but not significantly more (constant-factor). 
    This allows it to do simple operations, such as integer multiplication, that require more than constant space, without ``cheating'' by allowing the internal memory usage of the Turing machine to vastly exceed that required to represent the states.
} Turing machine, where the space bound is the maximum space needed to read and write the input and output states.
The original model~\cite{AADFP06} used $O(1)$ states and transitions for all $n$ and so was automatically uniform,
but many recent $\omega(1)$ state protocols are nonuniform.
With the exception of the uniform variant in~\cite{Berenbrink_majority_2020},
all $\omega(1)$ state stable majority protocols are nonuniform~\cite{alistarh2015fastExactMajority, alistarh2017time, alistarh2018space, bilke2017brief, berenbrink2018population,  ben2020log3}.
The uniform variant in~\cite{Berenbrink_majority_2020} has a tradeoff parameter $s$ that, when set to $O(1)$ to minimize the states,
uses $O(\log n \log \log n)$ states 
and $O(\log^2 n)$ time.

In this section
we show that there is a way to make \majorityNonuniform\ in \cref{sec:algorithm} uniform,
retaining the $O(\log n)$ time bound,
but the expected number of states increases to $\Theta(\log n \log \log n)$.\footnote{
    We say ``expected'' because
    this protocol has a non-zero probability of using an arbitrarily large number of states. 
    The number of states will be $O(\log n \log \log n)$ in expectation and with high probability.
}
    
    
    


\subsection{Main idea of $O(\log n \log \log n)$ state uniform majority (not handling ties)}
\label{subsec:uniform-not-handle-ties}

Since \majorityNonuniform\ uses the hard-coded value 
$L = \ceil{\log n}$,
to make the algorithm uniform,
we require a way to estimate $\log n$
and store it in a field $L$ (called $\fieldestimatelogn$ below) of each agent.
For correctness and speed, it is only required that $\fieldestimatelogn$ be within a constant multiplicative factor of $\log n$.

G\k{a}sieniec and Stachowiak~\cite{GS18} show a uniform  $O(\log \log n)$ state, $O(\log n)$ time protocol
(both bounds in expectation and with high probability) that 
computes and stores in each agent a value $\ell \in \N^+$ that, with high probability, 
is within additive constant $O(1)$ of $\ceil{\log \log n}$
(in particular, WHP  
$\ell \geq \floor{\log \log n} - 3$~\cite[Lemma 8]{Berenbrink_majority_2020}),
so $2^\ell = \Theta(\log n)$.
(This is the so-called \emph{junta election} protocol used as a subroutine for a subsequent leader election protocol.)
Furthermore, 
agents approach this estimate from below, propagating the maximum estimate by epidemic $\ell',\ell \to \ell,\ell$ if $\ell' < \ell$.
This gives an elegant way to compose the size estimation with a subsequent nonuniform protocol $\mathcal{P}$ that requires the estimate:
agents store their estimate $\fieldestimatelogn$ of $\log n$ and use it in $\mathcal{P}$.
Whenever an agent's estimate $\fieldestimatelogn$ updates---always getting larger---it simply resets $\mathcal{P}$, i.e., sets the entire state of $\mathcal{P}$ to its initial state.
We can then reason as though all agents actually started with their final convergent value of $\fieldestimatelogn$.\footnote{
    One might hope for a stronger form of composition,
    in which the size estimation \emph{terminates},
    i.e., sets an initially $\False$ Boolean flag to $\True$ only if the size estimation has converged,
    in order to simply prohibit the downstream protocol $\mathcal{P}$ from starting with an incorrect estimate of $\log n$.
    However, when both states $A$ and $B$ are initially $\Omega(n)$,
    this turns out to be impossible;
    $\Omega(n)$ agents will necessarily set the flag to $\True$ in $O(1)$ time,
    long before the $O(\log n)$ time required for the size estimation to converge~\cite[Theorem 4.1]{doty2019efficient}.
}

To make our protocol uniform, but remove its correctness in the case of a tie, as we explain below,
all agents conduct this size estimation, stored in the field $\fieldestimatelogn$, in parallel with the majority protocol $\mathcal{P}$ of \cref{sec:algorithm}.
Each agent resets $\mathcal{P}$ to its initial state whenever $\fieldestimatelogn$ updates.
This gives the stated $O(\log n)$ time bound and $O(\log n \log \log n)$ state bound.
Note that in Phase 0, agents count from $\fieldcounter = \Theta(\log n)$ down to 0.
It is sufficient to set the constant in the $\Theta$ sufficiently large that all agents with high probability receive by epidemic the convergent final value of $\fieldestimatelogn$ significantly before any agent with the same convergent estimate counts down to 0.

Acknowledging that,
with small probability,
the estimate of $\log n$ could be too low for \phaseDetectTie\ to be correct, we simply remove \phaseDetectTie\ and do not attempt to detect ties.
So if we permit undefined behavior in the case of a tie (as many existing fast majority protocols do),
then this modification of the algorithm otherwise retains stably correct, $O(\log n)$ time behavior,
while increasing the state complexity to $O(\log n \log \log n)$.

\subsection{How to stably compute ties}

With low but positive probability, the estimate of $\log n$ could be too small.
For most phases of the algorithm, this would merely amplify the probability of error events 
(e.g., \phaseDiscreteAveraging\ doesn't last long enough for agents to converge on biases $\{-1,0,+1\}$) 
that later phases are designed to handle.
However, the correctness of \phaseDetectTie\ 
(which detects ties) 
requires agents to have split through at least $\log n$ exponents in \phaseMainAveraging.
Since the population-wide bias doubles each time the whole population splits down one exponent,
the only way for the whole population to split through $\log n$ exponents is for there to be a tie (i.e., the population-wide bias is 0, so can double unboundedly many times).
In this one part of the algorithm,
for correctness we require the estimate to be at least $\log n$ with probability 1.
(It can be much greater than $\log n$ without affecting correctness; an overestimate merely slows down the algorithm.)

To correct this error, we will introduce a stable backup size estimate, to be done in \phaseDetectTie. Note that there are only a constant number of states with $\fieldphase = 4$: $\roleclock$ agents do not store a counter in this phase, and $\rolemain$ agents that stay in this phase must have $\fieldbias \in 
\left\{ 0, \pm \frac{1}{2^L} \right\}$.
Thus we can use an additional $\Theta(\log n)$ states for the agents that are currently in $\fieldphase = 4$ to stably estimate the population size. 
If they detect that their estimate of $L$ was too small, they simply go to the stable backup \phaseStableBackup.

\paragraph{Stable computation of $\floor{\log n}$.}
The stable computation of $\log n$ has all agents start in state $L_0$, where the subscript represents the agent's estimate of $\floor{\log n}$.
We have the following transitions:
for each $i \in \N$,
$L_i, L_i \to L_{i+1}, F_{i+1}$
and,
for each $0 \leq j < i$,
$F_i, F_j \to F_i, F_i$.
Among the agents in state $L_i$,
half make it to state $L_{i+1}$,
reaching a maximum of $L_k$ at $k = \floor{\log_2 n}$.\footnote{
    More generally, the unique stable configuration encodes the full population size $n$ in binary in the following distributed way:
    for each position $i$ of a 1 in the binary expansion of $n$,
    there is one agent $L_i$.
    Thus, these remaining agents lack the space to participate in propagating the value $k = \floor{\log n}$ by epidemic,
    but there are $\Omega(n)$ followers $F$ to complete the epidemic quickly.
}
All remaining $F$ agents receive the maximum value $k$ by epidemic.
A very similar protocol was analyzed in~\cite[Lemma 12]{Berenbrink2019counting}.
We give a quick analysis below for the sake of self-containment.


\paragraph{Time analysis of stable computation of $\floor{\log n}$.}
The expected time is $\Theta(n \log n)$.
For the upper bound, for each $i$,
if $j$ is the count of $L_i$ agents,
then the probability of a $L_i,L_i \to \ldots$ reaction is $O(j^2/n^2)$,
so for any possible starting count $k$ of $L_i$ agents, 
it takes expected time at most 
$O \qty( \frac{1}{n} \sum_{j=1}^k \frac{n^2}{j^2} ) 
= O(n)$ 
for the count of $L_i$ agents to get from $k$ to 0 or 1.
Assuming in the worst case that no $L_{i+1},L_{i+1} \to \ldots$ reaction happens until all $L_i,L_i \to \ldots$ reactions complete,
then summing over $\floor{\log n}$ values of $i$ gives the $O(n \log n)$ time upper bound to converge on one $L_{k}$ agent, where $k = \floor{\log n}$.
It takes additional $O(\log n)$ time for the $F$ agents to propagate $k$ by epidemic.

For the lower bound, observe that for each value of $i$, the \emph{last} reaction $L_i,L_i \to \ldots$ occurs when the count of $L_i$ is either 2 (and will increase at most once more), 
or 3 (and will never increase again).
This takes expected time $\Theta(n)$.
Since the count of $L_i$ will increase at most once more,
at that time (just before the last $L_i,L_i \to \ldots$ reaction), 
the count of $L_{i-1}$ is at most 3,
otherwise \emph{two} more $L_i$'s could be produced, and this would not be the last $L_i,L_i \to \ldots$ reaction.
Thus, to produce the $L_i$ needed for this last $L_i,L_i \to \ldots$ reaction,
the previous $L_{i-1},L_{i-1} \to L_i,F_i$ reaction occurs when the count of $L_{i-1}$ is at most 5, also taking $\Theta(n)$ time.
Thus the final reaction at each level takes time $\Theta(n)$ and depends on a reaction at the previous level that also takes time $\Theta(n)$,
so summing across $\floor{\log n}$ levels gives $\Omega(n \log n)$ completion time.

It is shown in~\cite[Lemma 12]{Berenbrink2019counting} that the time is $O(n \log^2 n)$ with high probability,
i.e., slower than the expected time by a $\log n$ factor.
Since the probability is $O(1/n^2)$ that we need to rely on this slow backup, even this larger time bound contributes negligibly to our total expected time.

\subsection{Challenges in creating $O(\log n)$ state uniform algorithm}
\label{subsec:uniform-issues-logn}

It is worth discussing some ideas for adjusting the uniform protocol described above to attempt to reduce its space complexity to $O(\log n)$ states.
The primary challenge is to enable the population size $n$ to be estimated without storing the estimate in any agent that participates in the main algorithm
(i.e., the agent has role $\rolemain$, $\roleclock$, or $\rolereserve$).
If agents in the main algorithm do not store the size, then by~\cite[Theorem 4.1]{doty2019efficient} they will provably go haywire initially, 
with agents in every phase, totally unsynchronized,
and require the size estimating agents to reset them after having converged on a size estimate that is $\Omega(\log n)$.

The following method would let the $\roleStableSize$ agents reset main algorithm agents, 
without actually having to store an estimate of ${\log n}$ in the algorithm agents, 
but it only works with high probability.
The size estimating agents could start a junta-driven clock as in~\cite{GS18},
which is reset whenever they update their size estimate.
Then, as long as there are $\Omega(n)$ $\roleStableSize$ agents,
they could for a phase timed to last $\Theta(\log n)$ time, 
reset the algorithm agents by \emph{direct} communication 
(instead of by epidemic).
This could put the algorithm agents in a quiescent state where they do not interact with each other, 
but merely wait for the $\roleStableSize$ agents to exit the resetting phase, 
indicating that the algorithm agents are able to start interacting again.
Since there are $\Omega(n)$ size-estimating agents,
each non-size-estimating agent will encounter at least one of them with high probability in $O(\log n)$ time.

The problem is that the reset signal is not guaranteed to reach every algorithm agent.
There is some small chance that a $\rolemain$ agent with a bias different from its $\fieldinput$ does not encounter a $\roleStableSize$ agent in the resetting phase, so is never reset.
The algorithm from that point on could reach an incorrect result when the agent interacts with properly reset agents,
since the sum of biases across the population has changed.
In our algorithm, by ``labeling'' each reset with the value $\fieldestimatelogn$,
we ensure that no matter what states the algorithm agents find themselves in during the initial chaos before size computation converges,
every one of them is guaranteed to be reset one last time with the same value of $\fieldestimatelogn$.
The high-probability resetting described above seems like a strategy that could work to create a high probability uniform protocol using $O(\log n)$ time and states,
though we have not thoroughly explored the possibility.

But it seems difficult to achieve probability-1 correctness using the technique of 
``reset the whole majority algorithm whenever the size estimate updates,''
without multiplying the state complexity by the number of possible values of $\fieldestimatelogn$.
Since we did not need $\floor{\log n}$ exactly,
but only a value that is $\Theta(\log n)$,
we paid only $\Theta(\log \log n)$ multiplicative overhead for the size estimate,
but it's not straightforward to see how to avoid this using the resetting technique.
Of course, one could imagine that the savings could come from reducing the state complexity of the main majority-computing agents in the nonuniform algorithm.
However, reducing the nonuniform algorithm's state complexity to below the $\Omega(\log n)$ lower bound of~\cite{alistarh2018space}
would provably require the algorithm to be not monotonic or not output dominant.
(See \cref{sec:conclusion} for a discussion of those concepts.)
Another possible approach is to intertwine more carefully the logic of the majority algorithm with the size estimation,
to more gracefully handle size estimate updates without needing to reset the entire majority algorithm.

\section{Conclusion}
\label{sec:conclusion}

There are two major open problems remaining concerning the majority problem for population protocols.

\paragraph{Uniform $O(\log n)$-time, $O(\log n)$-state majority protocol}
Our main $O(\log n)$ state protocol, 
described in \cref{sec:algorithm}, 
is \emph{nonuniform}:
all agents have the value $\ceil{\log n}$ encoded in their transition function.
The uniform version of our protocol described in \cref{sec:uniform} uses $O(\log n \log \log n)$ states.
It remains open to find a uniform protocol that uses $O(\log n)$ time and states.

\paragraph{Unconditional $\Omega(\log n)$ state lower bound for stable majority protocols}
The lower bound of $\Omega(\log n)$ states for (roughly) sublinear time majority protocols shown by Alistair, Aspnes, and Gelashvili~\cite{alistarh2018space}
applies only to stable protocols satisfying two conditions:
\emph{monotonicity} and \emph{output dominance}.

Recall that a \emph{uniform} protocol is one where a single set of transitions works for all population sizes; 
nonuniform protocols typically violate this by having an estimate of the population size 
(e.g., the integer $\ceil{\log n}$)
embedded in the transition function.
Monotonicity is a much weaker form of uniformity satisfied by nearly all known nonuniform protocols.
While allowing different transitions as the population size grows,
monotonicity requires that the transitions used for population size $n$ must also be correct for all smaller population sizes $n' < n$
(i.e., an \emph{overestimate} of the size cannot hurt),
and furthermore that the transitions be no slower on populations of size $n'$ than on populations of size $n$
(though the transitions designed for size $n$ may be slower on size $n'$ than the transitions intended for size $n'$).
Typically the nonuniform estimate of $\log n$ is used by the protocol to synchronize phases taking time $\approx T = \Theta(\log n)$, 
by having each agent individually count from $T$ down to 0 
(a so-called ``leaderless phase clock'', our \countersubroutine).
$\Theta(\log n)$ is the time required for an epidemic to reach the whole population and communicate some message to all agents before the phase ends.
Most errors in such protocols are the result of some agent not receiving a message before a phase ends, i.e., the clock runs atypically faster than the epidemic.
If the size estimate is significantly \emph{larger} than $\log n$, this slows the protocol down, but only increases the probability that all agents receive intended epidemic messages before a phase ends; thus such protocols are monotone.

In our nonuniform protocol \majorityNonuniform, 
which gives each agent the value $L = \ceil{\log n}$,
giving a different value of $L$ does not disrupt the stability (only the speed) of the protocol,
with one exception:
\phaseMainAveraging\ must go for at least $\log n$ exponents,
or else 
(i.e., if $L$ is an underestimate)
\phaseDetectTie\ could incorrectly report a tie.
If we consider running \phaseMainAveraging\ on a smaller population size $n'$ than the size $n$ for which it was designed,
$L=\ceil{\log n}$ of $\ceil{\log n'}$ will be an \emph{overestimate},
which does not disrupt correctness.
Furthermore, since $\roleclock$ agents are counting to $L$ in each case,
they are just as fast on population size $n'$ as on size $n$.
This implies that \phaseMainAveraging,
thus the whole protocol \majorityNonuniform, 
is monotone.

Output dominance references the concept of a \emph{stable} configuration $\vec{c}$,
in which all agents have a consensus opinion that cannot change in any configuration subsequently reachable from $\vec{c}$.
A protocol is \emph{output dominant} if in any stable configuration $\vec{c}$,
adding more agents with states already present in $\vec{c}$ maintains the property that every reachable stable configuration has the same output (though it may disrupt the stability of $\vec{c}$ itself).
This condition holds for all known stable majority protocols,
including that described in this paper,
because they obey the stronger condition that adding states already present in $\vec{c}$ does not even disrupt its stability.
Such protocols are based on the idea that two agents with the same opinion cannot create the opposite opinion,
so stabilization happens exactly when all agents first converge on a consensus output.

To see that our protocol is output dominant,
define a configuration to be \emph{silent} if no transition is applicable
(i.e., all pairs of agents have a null interaction);
clearly a silent configuration is also stable.
Although the definition of stable computation allows non-null transitions to continue happening in a stable configuration,
many existing stable protocols have the stronger property that they reach a silent configuration with probability 1, including our protocol.
It is straightforward to see that any silent configuration has the property required for output dominance,
since if no pair of states in a configuration can interact nontrivially,
their counts can be increased while maintaining this property.
(One must rule out the special case of a state with count 1 that can interact with another copy of itself,
which does not occur in our protocol's stable configurations.)

Monotonicity can be seen as a natural condition that all ``reasonable'' non-uniform protocols must satisfy, but output dominance arose as an artifact that was required for the lower bound proof strategy of~\cite{alistarh2018space}.
It remains open to prove an unconditional (i.e., removing the condition of output dominance) lower bound of $\Omega(\log n)$ states for \emph{any} stable monotone majority protocol taking polylogarithmic time,
or to show a stable polylogarithmic time monotone majority protocol using $o(\log n)$ states, which necessarily violates output dominance.
If the unconditional lower bound holds, then our protocol is simultaneously optimal for both time and states.
Otherwise, it may be possible to use $o(\log n)$ states to stably compute majority in polylogarithmic stabilization time with a non-output-dominant protocol.
In this case, there may be an algorithm simultaneously optimal for both time and states, or there may be a tradeoff.

\paragraph{$O(\log n)$ time, $o(\log n)$ state non-stable protocol}
Berenbrink, Els\"{a}sser, Friedetzky, Kaaser, Kling, and Radzik~\cite{Berenbrink_majority_2020} showed a non-stable majority protocol (i.e., it has a positive probability of error) using
$O(\log \log n)$ states and converging in 
$O(\log^2 n)$ time. 
(See \cref{sec:intro} for more details.)
Is there a protocol with $o(\log n)$ states solving majority in $O(\log n)$ time?

We close with questions unrelated to majority.

\paragraph{Fast population protocol for parity}
The majority problem ``$X_1 > X_2$?'' is a special case of a \emph{threshold} predicate,
which asks whether a particular weighted sum of inputs $\sum_{i=1}^k w_i X_i > c$ exceeds a constant $c$.
(For majority, $w_1=1, w_2=-1, c=0$; our protocol extends straightforwardly to other values of $w_i$, though not other values of $c$.)
The threshold predicates, together with the \emph{mod} predicates, characterize the \emph{semilinear} predicates, which are precisely the predicates stably computable by $O(1)$ state protocols~\cite{AAE06} with no time constraints
(though $\Theta(n)$ time to stabilize is sufficient for all semilinear predicates~\cite{AAE08} and necessary for ``most''~\cite{belleville2017time}).
A representative mod predicate is \emph{parity}: asking whether an odd or even number of agents exist.
(More generally, asking the parity of the number of agents with initial opinion $A$.)
Like majority, parity is solvable by a simple protocol in $O(n)$ time,\footnote{
    Agents start in state $L_1$, undergo reactions $L_i, L_j \to L_{(i+j) \mod 2}, F_{(i+j) \mod 2}$;
    \quad
    $L_i,F_j \to L_i,F_i$;
    \quad
    $F_i,F_j \to F_i,F_i$,
    i.e., $F$ agents adopt the parity of $L$ agents, and two $F$ agents with different parities adopt that of the sender.
    The first two reactions would take time $O(n \log n)$ for the last remaining $L_i$ to update all $F$ agents directly;
    the last reaction reduces this time to $O(n)$~\cite[Theorem 5]{AAE08}.
} 
and it is known to require $\Omega(n)$ time for any $O(1)$ state protocol to stabilize~\cite{belleville2017time}.
Techniques from~\cite{alistarh2017time} can be used to show that stabilization requires close to linear time even allowing up to $\frac{1}{2} \log \log n$ states.
An interesting open question is to consider allowing $\omega(1)$ states in deciding parity.
Can it then be decided in polylogarithmic time?
Is there a parity protocol simultaneously optimal for both polylogarithmic time and states, such as the $O(\log n)$ time, $O(\log \log n)$ state protocol for leader election~\cite{berenbrink2020optimal}?
Or is there a tradeoff?

\paragraph{Fast population protocols for function computation}
The transition $X,Q \to Y,Y$, starting with sufficiently many excess agents in state $Q$,
computes the function $f(x) = 2x$,
because if we start with $x$ agents in state $X$,
eventually $2x$ agents are in state $Y$,
taking time $O(\log n)$ to stabilize~\cite{CheDotSolNaCo}.
The similar transition $X,X \to Y,Q$ computes $f(x) = \floor{x/2}$,
but it takes time $\Theta(n)$ to stabilize, 
as does any $O(1)$-state protocol computing any linear function with a coefficient not in $\N$, 
as well as ``most'' non-linear functions such as $\min(x_1,x_2)$ (computable by $X_1,X_2 \to Y,Q$) and $\max(x_1,x_2)$~\cite{belleville2017time}.
Can such functions be computed in sublinear time by using $\omega(1)$ states?

\bibliography{refs}
\bibliographystyle{plain}

\end{document}